\documentclass[11pt, final, english]{article}

\usepackage[english]{babel}    
\usepackage[utf8]{inputenc}    
\usepackage[T1]{fontenc}       

\usepackage{amsmath}           
\usepackage{amssymb}          
\usepackage[amsmath,thmmarks,hyperref]{ntheorem} 
\usepackage{stmaryrd}

\usepackage[sort]{cite}        
\usepackage{setspace}          
\usepackage{exscale}           
\usepackage{xspace}            
\usepackage{bbm}               
\usepackage{mathtools}         
\usepackage{mathrsfs}          

\usepackage[expansion=false    
           ]{microtype}        
\usepackage[nottoc]{tocbibind} 
\usepackage[a4paper]{geometry} 
\usepackage[shortcuts]{extdash}
\usepackage{paralist}
\usepackage[final=true,        
           pdfpagelabels       
           ]{hyperref}         
\usepackage{fixme}

\allowdisplaybreaks

\newcommand{\refitem}[1] {\textit{\ref{#1}.)}}

\numberwithin{equation}{section}

\makeatletter
\g@addto@macro\bfseries{\boldmath}
\makeatother

\newcommand{\NN}{\mathbbm{N}}
\newcommand{\ZZ}{\mathbbm{Z}}
\newcommand{\RR}{\mathbbm{R}}
\newcommand{\CC}{\mathbbm{C}}
\newcommand{\PP}{\mathbbm{P}}
\newcommand{\DD }{\mathbbm{D}}
\newcommand{\Unit}{\mathbbm 1}

\newcommand{\E}{\mathrm{e}}
\newcommand{\I}{\mathrm{i}}

\newcommand{\cc}[1]{\overline{#1}}
\newcommand{\argument}{\,\cdot\,}
\newcommand{\at}[2][]{#1|_{#2}}
\newcommand{\av}{\mathrm{av}}

\newcommand{\group}[1]{\mathrm{#1}}
\newcommand{\Smooth}{\mathcal{C}^\infty}
\newcommand{\Stetig}{\mathcal{C}}
\newcommand{\lie}[1]{\mathfrak{#1}}
\newcommand{\momentmap}{\mathcal{J}}
\newcommand{\racts}{\mathbin{\triangleleft}}

\newcommand{\States}{\mathcal{S}}
\newcommand{\Weyl}{\mathcal{W}}

\newcommand{\algebra}[1]{\mathcal{#1}}
\newcommand{\A}{\algebra{A}}
\newcommand{\B}{\algebra{B}}
\newcommand{\Hermitian}{\mathrm{H}}
\newcommand{\Dom}{\mathcal{D}}
\newcommand{\Adbar}{\mathcal{L}^*}

\newcommand{\D}{\mathrm d}
\newcommand{\integral}[3]{\int_{#1} #2 \,#3}
\newcommand{\reductionMap}[2][]{#1[#2#1]_\mu}
\newcommand{\reductionMapSign}{[\argument]_\mu}

\newcommand{\Levelset}{\mathcal Z}
\newcommand{\Mred}{M_{\mu\mred}}
\newcommand{\MredExt}{\widehat M_{\mu\mred}}
\newcommand{\Schwartz}{\mathscr{S}}

\newcommand{\semialgebraicSet}[3][]{\mathcal Z^+#1(#2,#3#1)}

\newcommand{\punkt}{\,.}
\newcommand{\komma}{\,,}

\newcommand{\Polynomials}{\mathscr{P}}

\newcommand{\acts}{\triangleright}
\newcommand{\id}{\mathrm{id}}
\newcommand{\tr}{\mathrm{tr}}

\newcommand{\Cotangent}{\mathrm{T}^*}

\DeclareFontFamily{U}{FdSymbolF}{}
\DeclareFontShape{U}{FdSymbolF}{m}{n}{
	<-7.1> FdSymbolF-Book
	<7.1-> FdSymbolF-Book
}{}
\DeclareSymbolFont{delimiters}{U}{FdSymbolF}{m}{n}
\DeclareMathDelimiter{\llangle}{\mathopen}{delimiters}{"92}{delimiters}{"92}
\DeclareMathDelimiter{\rrangle}{\mathclose}{delimiters}{"98}{delimiters}{"98}

\DeclarePairedDelimiter{\ordinarySet}{\{}{\}}
\DeclarePairedDelimiter{\ordinaryIP}{\langle}{\rangle}
\DeclarePairedDelimiter{\ordinaryKom}{[}{]}

\newcommand{\set}[3][]{\ordinarySet[#1]{\,#2 \;#1|\; #3\,}}
\newcommand{\poi}[3][]{\ordinarySet[#1]{\,#2\mathbin{,}#3\,}}
\newcommand{\dupr}[3][]{\ordinaryIP[#1]{\,#2 \,,\, #3\,}}
\newcommand{\skal}[3][]{\ordinaryIP[#1]{\,#2 \,#1|\, #3\,}}
\newcommand{\kom}[3][]{\ordinaryKom[#1]{\,#2\,,\,#3\,}}

\newcommand{\genSId}[2][]{#1\langle\!#1\langle\,#2\,#1\rangle\!#1\rangle_{\ast\mathrm{id}}}
\newcommand{\genSAlg}[2][]{#1\langle\!#1\langle\,#2\,#1\rangle\!#1\rangle_{\ast\mathrm{alg}}}
\newcommand{\genQM}[2][]{#1\langle\!#1\langle\,#2\,#1\rangle\!#1\rangle_{\mathrm{qm}}}
\newcommand{\genPO}[2][]{#1\langle\!#1\langle\,#2\,#1\rangle\!#1\rangle_{\mathrm{po}}}
\newcommand{\abs}[2][]{#1|#2#1|}
\newcommand{\seminorm}[3][]{#1|#1|#3#1|#1|_{#2}}
\newcommand{\mred}{\textup{-}\mathrm{red}}

\theoremheaderfont{\normalfont\bfseries}
\theorembodyfont{\itshape}
\newtheorem{lemma}{Lemma}[section]
\newtheorem{proposition}[lemma]{Proposition}
\newtheorem{theorem}[lemma]{Theorem}

\newtheorem{corollary}[lemma]{Corollary}
\newtheorem{definition}[lemma]{Definition}

\theoremheaderfont{\normalfont\bfseries}
\theorembodyfont{\normalfont}
\newtheorem{example}[lemma]{Example}

\theorembodyfont{\rmfamily}
\theorembodyfont{\itshape}
\newtheorem*{reduction}{Reduction of States:}

\theorembodyfont{\normalfont}
\makeatletter
\def\thmhead@plain#1#2#3{%
	\thmname{#1}\thmnumber{\@ifnotempty{#1}{ }\@upn{#2}}%
	\thmnote{ {\the\thm@notefont#3}}}
\let\thmhead\thmhead@plain
\makeatother

\theoremheaderfont{\scshape}
\theorembodyfont{\normalfont}
\theoremstyle{nonumberplain}
\theoremseparator{:}
\theoremsymbol{\hbox{$\boxempty$}}
\newtheorem{proof}{Proof}
\theoremstyle{empty}

\author{
  \textbf{Philipp Schmitt}%
  \thanks{\href{mailto:schmitt@math.uni-hannover.de}{\texttt{schmitt@math.uni-hannover.de}},
  }
  \\[0.5cm]
  Institut für Analysis\\
  Leibniz Universität Hannover\\
  Welfengarten 1, 30167 Hannover\\
  Germany\\[1cm]
  \textbf{Matthias Schötz}%
  \thanks{\href{mailto:schotz@impan.pl}{\texttt{schotz@impan.pl}},
    former 'Boursier de l'ULB (Université libre de Bruxelles)',
    this work was supported by the Fonds de la Recherche Scientifique (FNRS) and the
    Fonds Wetenschappelijk Onderzoek - Vlaaderen (FWO) under EOS Project n$^\circ$30950721.
  }
  \\[0.5cm]
  Instytut Matematyczny PAN\\
  ul. Śniadeckich 8, 00-656 Warszawa\\
  Poland\\[1cm]
}

\title{Symmetry Reduction of States I}

\date{June 2023}

\begin{document}
	\maketitle
	\begin{abstract}
		We develop a general theory of symmetry reduction of states on (possibly non-commutative)
		$^*$\=/algebras that are equipped with a Poisson bracket and a Hamiltonian action of a commutative Lie algebra $\lie g$.
		The key idea advocated for in this article is that the ``correct'' notion of positivity on a $^*$\=/algebra $\A$ is not
		necessarily the algebraic one, for which positive elements are sums of Hermitian squares $a^*a$ with $a\in \A$,
		but can be a more general one that depends on the example at hand, like pointwise positivity on $^*$\=/algebras of functions
		or positivity in a representation as operators. The notion of states (normalized positive Hermitian linear functionals)
		on $\A$ thus depends on this choice of positivity on $\A$, and the notion of
		positivity on the reduced algebra $\A_{\mu\mred}$ should be such that states on $\A_{\mu\mred}$
		are obtained as reductions of certain states on $\A$.
		We discuss three examples in detail: Reduction of the $^*$\=/algebra of smooth functions on 
		a Poisson manifold $M$, reduction of the Weyl
		algebra with respect to translation symmetry, and reduction of the polynomial algebra with respect
		to a $\group{U}(1)$-action.
  \end{abstract} \ \\[-0.2cm]
{\small
	 }
    \newpage
	\tableofcontents
\begin{onehalfspace}
\section{Introduction}
Symmetry reduction, like Marsden--Weinstein reduction of symplectic manifolds or coisotropic reduction of Poisson manifolds, uses a well-behaved action of a symmetry group
to reduce the number of degrees of freedom of the system at hand. 
Roughly speaking, there exist two approaches: 
The geometric approach considers a symplectic or Poisson manifold $M$,
and symmetry reduction amounts to restricting to a levelset $\Levelset_\mu$ of
fixed momentum $\mu$ and dividing out the action of the corresponding symmetry group.
This way, one obtains a reduced symplectic or Poisson manifold $M_{\mu\mred}$.
In general, the ordering of the two steps is important,
but they commute in well-behaved cases.
Dual to the geometric approach is the algebraic approach,
which considers the associated Poisson algebra of functions $\A = \Smooth(M)$, 
usually referred to as the ``algebra of observables'' in physics.
Here, one divides out the vanishing ideal of $\Levelset_\mu$ 
and restricts to the subalgebra of functions that are 
invariant under the action of the symmetry group, 
thus one obtains a reduced algebra $\A_{\mu\mred} \cong \Smooth(M_{\mu\mred})$. 
This algebraic approach has
the advantage that it allows for a non-commutative generalization applicable to quantum physics by 
considering more general
algebras $\A$. See e.g.\ \cite{ortega.ratiu:momentumMapsAndHamiltonianReduction,
	arms.cushman.gotay:universalReductionProcedure,
	marsden.weinstein:reductionOfSymplecticManifoldsWithSymmetry,
	marsden.ratiu:reductionOfPoissonManifolds,
	bordemann.herbig.waldmann:BRSTCohomologyAndPhaseSpaceReduction
}.

The aim of this article is to examine a ``bidual'' version, 
the reduction of states on the algebra of observables.
It seems reasonable to expect that states on the reduced algebra should 
correspond to states on the original algebra in some way. 
This, however, is not trivially fulfilled in a naive approach:
The definition of states requires a notion of positivity on the algebra of observables;
and while there is a canonical notion of positivity on every $^*$\=/algebra by declaring Hermitian 
squares to be positive,
this is not enough to obtain a reasonable theory of reduction of states in general.
An example of this is discussed in Section~\ref{sec:polynomials}.
Ordered $^*$\=/algebras \cite{cimpric:representationTheoremForArchimedeanQuadraticModules, 
	popovych:positivstellensatzAndFlatFunctionals, powers:SelfadjointAlgebrasOfUnboundedOperators2,
	schoetz:equivalenceOrderAlgebraicStructure, schoetz:preprintGelfandNaimarkTheorems, 
	schmuedgen:StrictPositivstellensatzForWeylAlgebra, 
	schmuedgen:StrictPositivstellensatzForEnvelopingAlgebras, 
	schmuedgen:nonCommutativeRealAlgebraicGeometry},
which, from a different point of view, are also discussed in (non-commutative) real algebraic
geometry as ``$^*$\=/algebras equipped with a quadratic module'', offer the required flexibility, 
and allow us to make the correspondence between states on the original and reduced algebra 
precise.

Ordered $^*$\=/algebras are unital associative
$^*$\=/algebras over the field of complex numbers, for which the
real linear subspace $\A_\Hermitian \coloneqq \set{a\in \A}{a=a^*}$ of Hermitian elements is 
endowed with a
partial order fulfilling some compatibilities. 
A state on an ordered $^*$\=/algebra $\A$ is a normalized Hermitian linear functional $\omega 
\colon \A \to \CC$, positive with respect to the order on $\A_\Hermitian$.
Such a state $\omega$ associates to any observable $a \in \A$ 
its ``expectation value'' $\dupr \omega a$. 
The setting of ordered $^*$\=/algebras is general enough to cover a great 
number of examples, especially
the smooth complex-valued functions on a manifold $M$ with the pointwise order, $\A = \Smooth(M)$,
or the adjointable endomorphisms on a pre-Hilbert space~$\Dom$, $\A = \Adbar(\Dom)$.
Examples of states are evaluation functionals at points of $M$ or vector functionals 
$a \mapsto \skal \psi {a( \psi )}$  
corresponding to normalized vectors $\psi \in \Dom$.
Despite their generality, ordered $^*$\=/algebras
still allow for the development of some non-trivial results concerning their structure and representations and therefore might
be seen as a suitable generalization of $C^*$\=/algebras that also comprises unbounded examples.

We develop our theory of reduction in Section~\ref{sec:reduction:general} for ordered 
$^*$\=/algebras $\A$, endowed with a Poisson bracket and a Hamiltonian action of a commutative Lie 
algebra $\lie g$, i.e.\ an action induced by a momentum map $\momentmap \colon \lie g \to \A$.
Denote the $\lie g$-invariant elements of $\A$ by $\A^{\lie g}$.
We define the reduction $\A_{\mu\mred}$ of $\A$ for any ``momentum'' 
$\mu \in \lie g^*$, $\xi \mapsto \dupr \mu \xi$
by a universal property and show that the reduction always exists.
The construction is built in such a way that it behaves well with respect to states
under some minor technical assumptions:

\begin{reduction}
There is a bijective correspondence between states $\omega$ on $\A_{\mu\mred}$
and states $\hat\omega$ on $\mathcal A^{\lie g}$ which satisfy 
$\dupr{\hat\omega}{(\momentmap(\xi)-\dupr \mu \xi \Unit)^2} = 0$. 
Moreover, under the technical assumption of the existence of an averaging operator,
all such states can be obtained by restriction of states on $\A$. 
\end{reduction}
This might be seen as a partial justification of our setting: 
On the one hand, the assumptions we made are sufficiently strong to obtain a reasonable theory of 
symmetry reduction of states. But on the other hand, there are still plenty of examples,
which we discuss in Sections~\ref{sec:PoiMan}--\ref{sec:polynomials},
showing that our assumptions are not too strong.

That the interpretation of the reduction of states as the ``bidual'' of a geometric reduction 
procedure is not just a mere heuristic, 
can best be seen in the example $\mathcal A = \Smooth(M)$ of smooth functions on a Poisson manifold
$M$ with the pointwise order, which we discuss in detail in Section~\ref{sec:PoiMan}.
Indeed, $\Smooth(M)$ constitutes an example of an ordered $^*$\=/algebra
and assigning to every point $x\in M$ its evaluation
functional $\delta_x \colon \Smooth(M) \to \CC$ allows one to identify the manifold $M$
with the unital $^*$\=/homomorphisms $\Smooth(M) \to \CC$, 
which in turn are just the extreme points of the convex set of states on $\Smooth(M)$.
Given any smooth Hamiltonian action of a connected commutative Lie group $G$ on $M$,
the general reduction procedure for ordered $^*$\=/algebras and states, applied to this special 
example of $\Smooth(M)$ and its evaluation functionals,
results in a Poisson algebra $\mathcal{W}^\infty(M_{\mu\mred})$ of functions on a reduced topological space $M_{\mu\mred}$
and the evaluation functionals on $M_{\mu\mred}$.
Under the usual additional regularity assumptions this procedure is equivalent to Marsden--Weinstein reduction.

As a first non-commutative example we discuss the Weyl algebra of canonical commutation relations 
in Section~\ref{sec:weyl}.
In the Schrödinger representation as differential operators on the Schwartz space 
$\Schwartz(\RR^{1+n})$,
the momentum map for the translation symmetry is simply given by the usual momentum operators $-\I 
\frac{\partial}{\partial x_j}$.
It will be shown that the reduction with respect to one of the momentum operators yields the Weyl 
algebra on $\Schwartz(\RR^{n})$ with the operator order.
In this example, there are no states on the Weyl algebra satisfying the above reducibility 
condition (due to the lack of an averaging operator), yet the reduction procedure
still produces the expected result.

More involved non-commutative examples arise in non-formal deformation quantization,
but their detailed study will be postponed to future projects. 
In this article we discuss, as our last example in Section~\ref{sec:polynomials},
the case of the polynomial algebra on $\CC^{1+n}$
with the standard Poisson bracket. By reduction with respect to a $\group{U}(1)$-action one obtains
algebras of polynomials on e.g.~$\CC\PP^n$ or the hyperbolic disc $\DD^n$. This is the classical 
limit of some well-known
non-formal star products, which can also be obtained by symmetry reduction of the Wick star product 
on $\CC^{1+n}$, see e.g.~\cite{beiser.waldmann:FrechetAlgebraicDeformationOfThePoincareDisc, 
bordemann.brischle.emmrich.waldmann:PhaseSpaceReductionForStarProducts.ExplicitConstruction, 
bordemann.brischle.emmrich.waldmann:SubalgebrasWithConvergingStarProducts, 
cahen.gutt.rawnsley:QuantisationOfKaehlerManifolds3, 
esposito.schmitt.waldmann:OnlineComparisonContinuityWickStarProducts, 
kraus.roth.schoetz.waldmann:OnlineConvergentStarProductOnPoincareDisc, 
schmitt.schoetz:PreprintWickRotationsInDQ}.
These examples are especially relevant as the starting point for studying non-formal deformations 
of $^*$\=/algebras.
But these examples also demonstrate that it is in general not sufficient to simply consider 
$^*$\=/algebras with the canonical
algebraic order given by sums of Hermitian squares. In these cases, finding a suitable algebraic characterization of the order on the reduced algebras
is well-known to be a non-trivial problem already in the commutative case, but one that has been solved in great generality with the Positivstellensatz of Krivine and Stengle or similar results \cite{krivine:AnneauxPreordonnes, marshall:extendingArchimedeanPositivstellensatzToTheNonCOmpactCase, stengle:aNullstellensatzAndAPositivstellensatzInSemialgebraicGeometry, schmuedgen:kMomentProblemForCompactSemialgebraicSets}.
This raises the question whether and how similar algebraic characterizations can also be obtained in the non-commutative case, where they would
be especially valuable because the idea of a pointwise order on an easy-to-describe reduced manifold is no longer applicable.
For the star product on $\CC\PP^n$, this problem will be solved in \cite{schmitt.schoetz:SymmetryReductionOfStatesII}.

\section{Notation and Preliminaries} \label{sec:preliminaries}
The notation essentially follows \cite{schoetz:equivalenceOrderAlgebraicStructure}. See also \cite{schmuedgen:invitationToStarAlgebras}
for an introduction to $^*$\=/algebras and quadratic modules on them (but be aware of some differences in notation).

The natural numbers are denoted by $\NN \coloneqq \{1,2,3,\dots\}$ and $\NN_0 \coloneqq \{0\} \cup \NN$, and the fields of
real and complex numbers are $\RR$ and $\CC$, respectively. A \emph{quasi-order} on a set $X$ is a reflexive and transitive
relation $\lesssim$ on $X$. Given two sets $X$ and $Y$, both equipped with a quasi-order $\lesssim$,
and a map $\Phi \colon X \to Y$, then $\Phi$ is said to be \emph{increasing} (or \emph{decreasing}) if $\Phi(x) \lesssim \Phi(x')$
(or $\Phi(x) \gtrsim \Phi(x')$) holds for all $x,x'\in X$ with $x\lesssim x'$. If $\Phi$ is injective and increasing,
and if additionally $x\lesssim x'$ holds for all $x,x' \in X$ for which $\Phi(x) \lesssim \Phi(x')$, then $\Phi$ is called an \emph{order embedding}.

\subsection{Ordered \texorpdfstring{$^*$-Algebras}{*-Algebras}} \label{sec:preliminaries:orderedsalg}
A \emph{$^*$\=/algebra} $\A$ is a unital associative $\CC$-algebra 
endowed with an antilinear involution $\argument^* \colon \A \to \A$ such that $(ab)^* = b^* a^*$ holds
for all $a,b\in \A$. Its unit will be denoted by $\Unit$, or, more explicitly, $\Unit_\A$. It is not required that $\Unit \neq 0$,
which means that $\{0\}$ is a $^*$\=/algebra.
An element $a\in \A$ is called \emph{Hermitian} if $a=a^*$, and $\A_\Hermitian \coloneqq \set{a\in \A}{a=a^*}$
clearly is a real linear subspace of $\A$. A \emph{quadratic module} on a $^*$\=/algebra $\A$ is a subset $\mathcal{Q}$
of $\A_\Hermitian$ that fulfils
\begin{equation}
  a + b \in \mathcal{Q}
  \komma\quad\quad
  d^*a\,d \in \mathcal{Q}\komma
  \quad\quad\text{and}\quad\quad
  \Unit \in \mathcal{Q}
  \label{eq:axiomsQM}
\end{equation}
for all $a,b\in \mathcal{Q}$ and all $d\in \A$. Similarly, a \emph{quasi-ordered $^*$\=/algebra}
is a $^*$\=/algebra $\A$ endowed with a reflexive and transitive relation $\lesssim$ on $\A_\Hermitian$ that additionally fulfils the conditions
\begin{equation}
  a+c \lesssim b+c
  \komma\quad\quad
  d^*a\,d \lesssim d^*b\,d\komma
  \quad\quad\text{and}\quad\quad
  0 \lesssim \Unit
  \label{eq:axioms}
\end{equation}
for all $a,b,c\in \A_\Hermitian$ with $a\lesssim b$ and all $d\in\A$. We simply refer to the 
relation $\lesssim$ as the \emph{order} on $\A$.
An element $a\in\A_\Hermitian$ is called \emph{positive} if
$0 \lesssim a$, and the set of all positive Hermitian elements of $\A$ will be denoted by $\A_\Hermitian^+ \coloneqq \set{a\in \A_\Hermitian}{0 \lesssim a}$.
It is easy to check that $\A_\Hermitian^+$ is a quadratic module on $\A$. Conversely, any quadratic module $\mathcal{Q}$
on any $^*$\=/algebra $\A$ allows one to define a relation $\lesssim$ on $\A_\Hermitian$ for $a,b\in \A_\Hermitian$ as $a\lesssim b$ if and only if $b-a\in \mathcal{Q}$,
and then $\A$ with $\lesssim$ is a quasi-ordered $^*$\=/algebra for which $\A^+_\Hermitian = \mathcal{Q}$.
An \emph{ordered $^*$\=/algebra} is a quasi-ordered $^*$\=/algebra whose order $\lesssim$ is also 
antisymmetric,
hence a partial order. Equivalently, a quasi-ordered $^*$\=/algebra $\A$ is an ordered $^*$\=/algebra if and only if $(-\A^+_\Hermitian) \cap \A^+_\Hermitian = \{0\}$.
In the case of ordered $^*$\=/algebras, the order relation is usually written as $\le$.

Ordered $^*$\=/algebras, or quadratic modules on $^*$\=/algebras, occur e.g.~in the literature on representations of $^*$\=/algebras,
sometimes as ``$^*$\=/algebras equipped with an admissible cone'' as in \cite{powers:SelfadjointAlgebrasOfUnboundedOperators2}.
They can be seen as abstractions of the $^*$\=/algebras of (possibly unbounded) adjointable endomorphisms on a pre-Hilbert
space, similar to the way $C^*$\=/algebras are abstractions of bounded operators on a Hilbert space: In 
\cite{cimpric:representationTheoremForArchimedeanQuadraticModules} it is shown that the order gives rise to a $C^*$\=/seminorm on bounded elements,
and in sufficiently well-behaved cases one can generalize constructions or representation theorems from $C^*$\=/algebras to ordered $^*$\=/algebras,
see e.g. \cite{schoetz:preprintGelfandNaimarkTheorems, schoetz:equivalenceOrderAlgebraicStructure}. 
Quadratic modules are also studied in real algebraic geometry, especially in the 
commutative case where they
describe properties of the cone of sums of squares of real polynomials. However, some ideas of real algebraic geometry can also be adapted
to the non-commutative case, see e.g.~\cite{schmuedgen:nonCommutativeRealAlgebraicGeometry} for an overview.

A linear map $\Phi \colon \A \to \B$ between two $^*$\=/algebras $\A$ and $\B$ is called \emph{Hermitian}
if $\Phi(a^*) = \Phi(a)^*$ for all $a\in \A$, or equivalently if $\Phi(a) \in \B_\Hermitian$ for all $a\in \A_\Hermitian$.
A \emph{unital $^*$\=/homomorphism} is such a Hermitian linear map $\Phi \colon \A \to \B$ that additionally fulfils
$\Phi(\Unit_\A) = \Unit_\B$ and $\Phi(aa') = \Phi(a) \Phi(a')$ for all $a,a'\in \A$. If both $\A$ and $\B$ are quasi-ordered $^*$\=/algebras,
then a Hermitian linear map $\Phi \colon \A \to \B$ is said to be \emph{positive} if its restriction
to an $\RR$-linear map from $\A_\Hermitian$ to $\B_\Hermitian$ is increasing, or equivalently if $\Phi(a) \in \B^+_\Hermitian$
for all $a\in \A^+_\Hermitian$. Similarly, $\Phi$ is said to be an \emph{order embedding} if its 
restriction to Hermitian elements
is an order embedding. Especially if $\B=\CC$ (with the usual order on $\CC_\Hermitian = \RR$), then
we write $\A^*$ for the dual space of $\A$ whose elements are \emph{linear functionals} $\omega \colon \A \to \CC$,
and use the bilinear dual pairing $\dupr{\argument}{\argument} \colon \A^* \times \A \to \CC$, $(\omega,a) \mapsto \dupr{\omega}{a}$
to denote the evaluation of a linear functional $\omega \in \A^*$ on an algebra element $a \in \A$.
Similarly, we write $\A^{*}_\Hermitian$ for the real linear subspace of Hermitian linear functionals on $\A$
and $\A^{*,+}_\Hermitian$ for the convex cone of positive Hermitian linear functionals therein.
Note that clearly $\lambda \omega + \mu \rho \in \A^{*,+}_\Hermitian$ for all $\omega,\rho \in \A^{*,+}_\Hermitian$
and all $\lambda,\mu \in {[0,\infty[}$, and that the Cauchy--Schwarz inequality for the positive Hermitian sesquilinear form
$\A \times \A \ni (a,b) \mapsto \dupr{\omega}{a^*b} \in \CC$ for $\omega \in \A^{*,+}_\Hermitian$ shows that
\begin{equation}
  \abs[\big]{\dupr{\omega}{a^*b}}^2 \le \dupr{\omega}{a^*a}\dupr{\omega}{b^*b}
  \label{eq:CS}
\end{equation}
for all $a,b\in \A$. A \emph{state} on $\A$ is a positive Hermitian
linear functional that fulfils the normalization $\dupr{\omega}{\Unit} = 1$ and the set of states on $\A$ will
be denoted by $\States(\A)$. Setting $a \coloneqq \Unit$ in \eqref{eq:CS} shows that $\dupr{\omega}{\Unit} = 0$
implies $\omega = 0$, so $(-\A^{*,+}_\Hermitian) \cap \A^{*,+}_\Hermitian = \{0\}$ and every non-zero positive Hermitian
linear functional can be rescaled to a state. This allows one to reformulate most statements for 
positive Hermitian linear functionals
to equivalent statements for states.

If $\A$ is a quasi-ordered $^*$\=/algebra, then any \emph{unital $^*$\=/subalgebra} $S$ of $\A$, i.e.\ a linear subspace
$S\subseteq \A$ with $\Unit \in S$ which is stable under $\argument^*$ and closed under multiplication, is again a $^*$\=/algebra and
becomes a quasi-ordered $^*$\=/algebra with the restriction of the order of $\A$. We will always endow unital $^*$\=/subalgebras
with this restricted order. Similarly, if $\mathcal{I}$ is a $^*$\=/ideal of $\A$, i.e.~a linear subspace 
$\mathcal{I}\subseteq \A$ which is stable under $\argument^*$ and which fulfils $a b \in \mathcal{I}$ for all $a\in \A$, $b\in \mathcal{I}$,
then the quotient vector space $\A / \mathcal{I}$ becomes a $^*$\=/algebra in a unique way by 
demanding that the canonical
projection $[\argument] \colon \A \to \A/\mathcal{I}$ be a unital $^*$\=/homomorphism. 
This quotient $^*$\=/algebra 
even becomes a quasi-ordered $^*$\=/algebra with the order whose quadratic module of positive 
elements is
$\set{[a]}{a\in \A^+_\Hermitian}$; this order will be called the \emph{quotient order}. This way, $[\argument] \colon \A \to \A/\mathcal{I}$ becomes a positive unital $^*$\=/homomorphism,
and it is easy to check that the usual universal property of quotients is fulfilled: Whenever $\Phi \colon \A \to \B$ is a positive unital $^*$\=/homomorphism
(or, more generally, positive Hermitian linear map) to any quasi-ordered $^*$\=/algebra $\B$ such that 
$\mathcal{I} \subseteq \ker \Phi \coloneqq \Phi^{-1}(\{0\})$,
then there exists a unique positive unital $^*$\=/homomorphism (or positive Hermitian linear map) $\phi \colon \A / \mathcal{I} \to \B$
that fulfils $\Phi = \phi \circ [\argument]$.

\subsection{Constructing Quadratic Modules} \label{subsec:quadraticModules}
There are two canonical classes of examples of ordered $^*$\=/algebras, namely \emph{ordered $^*$\=/algebras of functions},
which are unital $^*$\=/subalgebras of the ordered $^*$\=/algebra $\CC^X$ of all complex-valued functions on a set $X$
with the pointwise operations and pointwise order, and \emph{$O^*$\=/algebras}, which are
unital $^*$\=/subalgebras of the ordered $^*$\=/algebra $\Adbar(\Dom)$ of all adjointable endomorphisms $a \colon \Dom \to \Dom$
on a pre-Hilbert space $\Dom$ with inner product $\skal{\argument}{\argument}$ (antilinear in the first, linear in the second argument)
with the standard operator order.
Here \emph{adjointable} is to be understood in the algebraic sense, i.e.~$a \colon \Dom \to \Dom$ is adjointable
if there exists a (necessarily unique and linear) $a^* \colon \Dom \to \Dom$ such that $\skal{\phi}{a(\psi)} = \skal{a^*(\phi)}{\psi}$
holds for all $\phi,\psi \in \Dom$. The operator order for $a,b\in \Adbar(\Dom)_\Hermitian$ is 
determined by $a\le b$ if and only if
$\skal{\psi}{a(\psi)} \le \skal{\psi}{b(\psi)}$ for all $\psi \in \Dom$.

There are essentially two possibilities to endow a $^*$\=/algebra with a suitable order:
Either by demanding that certain Hermitian elements should be positive, or by demanding that 
certain Hermitian linear functionals should be positive.
More precisely, given a $^*$\=/algebra $\A$ and any subset $S$ of $\A_\Hermitian$, then
\begin{equation}
  \genQM{S} \coloneqq \set[\Big]{\sum\nolimits_{m=1}^M a_m^* s_m a_m }{M\in \NN_0;\, a_1,\dots,a_M \in \A;\, s_1, \dots,s_M \in S \cup \{\Unit\} }
  \label{eq:genQM}
\end{equation}
is the \emph{quadratic module generated by $S$}, which clearly is the smallest (with respect to 
$\subseteq$) quadratic module on $\A$ that contains $S$.
As a special case, let
\begin{equation} \label{eq:algebraicOrder}
  \A^{++}_\Hermitian \coloneqq \genQM{\emptyset}
\end{equation}
be the quadratic module generated by the empty set, hence the smallest quadratic module on $\A$. Its elements are, by construction, the sums of \emph{Hermitian squares}
$a^*a$ with $a\in \A$, and will be referred to as \emph{algebraically positive elements}. The resulting \emph{algebraic order} on $\A$
gives a canonical, non-trivial way to turn any $^*$\=/algebra into a quasi-ordered $^*$\=/algebra. However, in many examples this is not
the ``correct'' one (the meaning of which, of course, depends on the context). A Hermitian linear functional on a $^*$\=/algebra $\A$
which is positive with respect to this algebraic order will be called \emph{algebraically positive},
and an \emph{algebraic state} therefore is a normalized algebraically positive Hermitian linear functional.
For example, the usual order on the Hermitian elements of a $C^*$\=/algebra
can be described as the one whose positive elements are those with spectrum contained in ${[0,\infty[}$, or equivalently as the one
whose positive elements are precisely the algebraically positive ones.

In the commutative case, quadratic modules that are closed under multiplication are especially interesting (and referred to as ``preordering'').
Thus for a commutative $^*$\=/algebra $\A$ and any subset $S$ of $\A_\Hermitian$, the \emph{preordering generated by $S$} is
\begin{align}
  \genPO{S} \coloneqq \genQM[\Big]{ \set[\Big]{ \prod\nolimits_{m=1}^M s_m }{M\in \NN;\,s_1,\dots,s_M \in S} }
  \,,
  \label{eq:genPO}
\end{align}
which is the smallest (with respect to $\subseteq$) quadratic module on $\A$ that is closed under multiplication and contains $S$.

Moreover, quasi-ordered $^*$\=/algebras can also be constructed by demanding that certain algebraic 
states be positive:
Let $\A$ be a $^*$\=/algebra and let $\argument \acts \argument \colon \A \times \A^* \to \A^*$ be the left action of the multiplicative monoid of $\A$ 
on $\A^*$ by conjugation, i.e.~$\dupr{a\acts \omega}{b} \coloneqq \dupr{\omega}{a^*b\,a}$ for all $a,b\in \A$ and all $\omega \in \A^*$.
A set of algebraically positive Hermitian linear functionals on a $^*$\=/algebra $\A$ that is stable under this action gives rise to an order on $\A_\Hermitian$:

\begin{proposition} \label{proposition:inducedOrder} 
  Let $\A$ be a $^*$\=/algebra and $S$ a set of algebraic states on $\A$ such that 
  $a \acts \omega \in S$ holds for all $\omega \in S$ and $a\in \A$ with $\dupr{\omega}{a^*a} = 1$. Then
  \begin{align}
    \mathcal{Q} \coloneqq \set[\big]{a\in \A_\Hermitian}{\dupr{\omega}{a} \ge 0\textup{ for all }\omega \in S}
    \label{eq:inducedOrder:quadraticmodule}
  \end{align}
  is a quadratic module on $\A$, so $\A$ can be turned into a quasi-ordered $^*$\=/algebra with $\A_\Hermitian^+ = \mathcal{Q}$. Similarly,
  \begin{align}
    \mathcal{I} \coloneqq \set[\big]{a\in \A}{\dupr{\omega}{a} = 0\textup{ for all }\omega \in S}
    \label{eq:inducedOrder:ideal}
  \end{align}
  is a $^*$\=/ideal of $\A$ and the quotient $^*$\=/algebra $\A / \mathcal{I}$ with the quotient order is an ordered $^*$\=/algebra.
  Moreover, for every $\omega \in S$ there exists a unique state $\check{\omega}$ on $\A / \mathcal{I}$
  fulfilling $\check{\omega}\circ [\argument] = \omega$ with $[\argument] \colon \A \to \A / \mathcal{I}$ the canonical projection onto the quotient,
  and
  \begin{align}
    (\A / \mathcal{I})^+_\Hermitian
    &=
    \set[\big]{[a] \in (\A / \mathcal{I})_\Hermitian }{ \dupr{\check{\omega}}{[a]} \ge 0 \textup{ for all }\omega \in S}
    \label{eq:inducedOrder:quotient}
    \punkt
  \end{align}
\end{proposition}
\begin{proof}
  Note that for $\omega \in S$ and $a\in \A$ one either has $\dupr{\omega}{a^*a} > 0$, hence 
  $\dupr{\omega}{a^*a}^{-1} (a\acts \omega) = (\dupr{\omega}{a^*a}^{-1/2} a) \acts \omega \in S$, or $\dupr{\omega}{a^*a} = 0$,
  in which case $\dupr{a\acts \omega}{\Unit} = 0$ and therefore $a\acts \omega = 0$ as a consequence of the Cauchy--Schwarz inequality.
  It is now straightforward to check that $\mathcal{Q}$ is a quadratic module. It is also clear that
  $\mathcal{I}$ is a linear subspace of $\A$ and stable under $\argument^*$,
  and $\mathcal{I}$ is a right ideal, hence a $^*$\=/ideal, because for all $a\in \mathcal{I}$, $b\in \A$ and $\omega \in S$ one has
  $\dupr{\omega}{a b} = \frac{1}{4} \sum_{k=0}^3 \I^{k} \dupr{(b+\I^k\Unit) \acts \omega}{a} = 0$.
  
  The quotient order on $\A/\mathcal{I}$ is even a partial order:
  Given $[a] \in (\A / \mathcal{I})_\Hermitian$ with $[0] \le [a] \le [0]$, then there exist representatives
  $\hat{a}_1,\hat{a}_2 \in [a] \cap \A_\Hermitian$ such that $0 \lesssim \hat{a}_1$ and $\hat{a}_2 
  \lesssim 0$, so
  $0 \le \dupr{\omega}{\hat{a}_1} = \dupr{\omega}{\hat{a}_2} \le 0$ for all $\omega \in S$ because $\hat{a}_1 - \hat{a}_2 \in \mathcal{I}$,
  which shows that $[a] = \mathcal{I} = [0]$. Moreover, essentially by definition of $\mathcal{I}$ and the quotient order,
  every $\omega \in S$ descends to a unique state $\check{\omega} \in \States(\A/\mathcal{I})$ fulfilling
  $\check{\omega}\circ [\argument] = \omega$.
  
  The inclusion ``$\subseteq$'' in \eqref{eq:inducedOrder:quotient} follows from the definitions of the quadratic module $\mathcal{Q}$
  and the quotient order.
  Conversely, let $[a] \in (\A / \mathcal{I})_\Hermitian$ be given such that $\dupr{\check{\omega}}{[a]} \ge 0$ for all $\omega \in S$,
  and choose any Hermitian representative $\hat{a} \in [a] \cap \A_\Hermitian$ (for example, take the real part $\hat{a} \coloneqq (\tilde{a}^*+\tilde{a})/2$
  of an arbitrary representative $\tilde{a} \in [a]$). Then $\dupr{\omega}{\hat{a}} = \dupr{\check{\omega}}{[a]} \ge 0$ for all $\omega \in S$
  shows that $\hat{a} \in \mathcal{Q} = \A^+_\Hermitian$, so $[a] \in (\A / \mathcal{I})_\Hermitian^+$.
\end{proof}
The order on $\A$ that was constructed in the above Proposition~\ref{proposition:inducedOrder}, 
\eqref{eq:inducedOrder:quadraticmodule}
will be called the one \emph{induced by $S$}.
If $\A$ is a quasi-ordered $^*$\=/algebra, then we especially say that \emph{its order is induced 
by its states}
if the given order on $\A$ is the one induced by $\States(\A)$, or equivalently, if for all $a\in 
\A_\Hermitian \setminus \A_\Hermitian^+$
there exists $\omega \in \States(\A)$ such that $\dupr{\omega}{a} < 0$. It is a standard 
consequence of the Hahn--Banach theorem that this
is the case if and only if $\A^+_\Hermitian$ is closed in some locally convex topology on $\A_\Hermitian$, e.g.~the strongest one.
Identity \eqref{eq:inducedOrder:quotient} just means that the order on the quotient $\A / \mathcal{I}$ in Proposition~\ref{proposition:inducedOrder} 
is induced by its states.

Ordered $^*$\=/algebras $\A$ whose order is induced by their states have some desirable properties.
For example, if $a\in \A$ fulfils $\dupr{\omega}{a} = 0$ for all $\omega \in \States(\A)$, then also $\dupr{\omega}{a+a^*} = 0$
and $\dupr{\omega}{\I(a-a^*)} = 0$ for all $\omega \in \States(\A)$, which implies that $0 \le a+a^* \le 0$ and $0 \le \I(a-a^*) \le 0$,
hence $a = 0$. Similarly:
\begin{proposition} \label{proposition:positivitycriterium}
  Let $\Phi \colon \A \to \B$ be a unital $^*$\=/homomorphism between two quasi-ordered $^*$\=/algebras $\A$ and $\B$,
  and assume that the order on $\B$ is induced by its states. Then $\Phi$ is
  positive if and only if $\rho \circ \Phi \in \States(\A)$ holds for all $\rho \in \States(\B)$.
\end{proposition}
\begin{proof}
  If $\Phi$ is positive and $\rho$ a state on $\B$, then $\rho \circ \Phi$ is again positive, hence a state on $\A$.
  Conversely, if $\rho \circ \Phi \in \States(\A)$ holds for all $\rho \in \States(\B)$, then $\Phi$ is positive
  because $\dupr{\rho}{\Phi(a)} = \dupr{\rho \circ \Phi}{a} \ge 0$ shows that $\Phi(a) \in \B^+_\Hermitian$ for every $a\in \A^+_\Hermitian$.
\end{proof}
For example, the pointwise order on an ordered $^*$\=/algebra $\A \subseteq \CC^X$ of functions on a set $X$ is the one induced by
the set $\set{\delta_x}{x\in X}$ of \emph{evaluation functionals} $\delta_x \colon \A \to \CC$, $a\mapsto \dupr{\delta_x}{a} \coloneqq a(x)$.
Similarly, the operator order on an $O^*$\=/algebra $\A \subseteq \Adbar(\Dom)$ on a pre-Hilbert space $\Dom$ is the one induced by
the set $\set{\chi_\psi}{\psi \in \Dom \text{ with $\seminorm{}{\psi} = 1$}}$ of \emph{vector 
functionals} 
$\chi_\psi \colon \A \to \CC$, $a\mapsto \dupr{\chi_\psi}{a} \coloneqq \skal{\psi}{a(\psi)}$.

Relating quadratic modules that are constructed ``analytically'' as in \eqref{eq:inducedOrder:quadraticmodule} to suitable ``algebraically'' constructed quadratic modules
as in \eqref{eq:genQM} or \eqref{eq:genPO} is a typical problem of (possibly non-commutative) real algebraic geometry.
The most famous results are Artin's solution of Hilbert's 17th problem and the Positivstellensatz of Krivine and Stengle 
\cite{krivine:AnneauxPreordonnes, stengle:aNullstellensatzAndAPositivstellensatzInSemialgebraicGeometry}
that give an algebraic description of the pointwise order on polynomial algebras.

\subsection{Eigenstates}

Especially for those ordered $^*$\=/algebras that occur as algebras of observables in physics, the 
notion of states (which describe
the actual state of a physical system) is of major importance, generalizing the concept of vector states on $O^*$\=/algebras.
There is also an abstraction of the idea of vector states constructed out of eigenvectors of an operator:

\begin{definition} \label{definition:eigenstates}
  Let $\A$ be a quasi-ordered $^*$\=/algebra and $a\in \A$. An \emph{eigenstate of $a$} is a state $\omega$ on $\A$ that fulfils 
  $\dupr{\omega}{a^*a} = \abs{\dupr{\omega}{a}}^2$,
  and the complex number $\dupr{\omega}{a}$ then is called the \emph{eigenvalue of $\omega$ on $a$}.
  The set of all eigenstates of $a$ with eigenvalue $\mu \in \CC$ will be denoted by $\States_{a,\mu}(\A)$.
\end{definition}
It can also happen that one state is an eigenstate of several elements of $\A$, in which case we 
call it a \emph{common eigenstate} of these elements.
The notion of eigenstates occurs once in a while in the literature on $C^*$\=/algebras 
\cite{riedel:almostMathieuOperators:eigenstates, paschke:pureEigenstatesForTheSumOfGenOfTheFreeGroup, rinehart:eigenstatesOfCStarAlgebras},
but their basic properties are fulfilled in greater generality:

\begin{proposition} \label{proposition:eigenstates}
  Let $\A$ be a quasi-ordered $^*$\=/algebra, $a\in \A$, and $\omega \in \States(\A)$. Then the
  following are equivalent:
  \begin{enumerate}
    \item There exists a complex number $\mu$ such that $\dupr{\omega}{(a-\mu \Unit)^*(a-\mu \Unit)} = 0$.
      \label{item:eigenstates:eigenvalue}
    \item The identities $\dupr{\omega}{a^*b} = \cc{\dupr{\omega}{a}} \dupr{\omega}{b}$ and $\dupr{\omega}{b^*a} = \cc{\dupr{\omega}{b}} \dupr{\omega}{a}$ hold for all $b\in \A$.
      \label{item:eigenstates:mult}
    \item The identity $\dupr{\omega}{a^*a} = \abs{\dupr{\omega}{a}}^2$ holds, i.e.~$\omega$ is an eigenstate of $a$.
      \label{item:eigenstates:var}
  \end{enumerate}
  Moreover, if the first point \refitem{item:eigenstates:eigenvalue} holds for some $\mu \in \CC$, then $\mu = \dupr{\omega}{a}$ is the eigenvalue of $\omega$ on $a$.
\end{proposition}
\begin{proof}
  The proof is essentially the same as for eigenstates on $C^*$\=/algebras, and is repeated here for convenience of the reader:
  First assume that some $\mu \in \CC$ fulfils $\dupr{\omega}{(a-\mu \Unit)^*(a-\mu \Unit)} = 0$.
  Then it follows from the Cauchy--Schwarz inequality that $0 \le \abs{\dupr{\omega}{a-\mu \Unit}}^2 \le \dupr{\omega}{(a-\mu\Unit)^*(a-\mu\Unit)} = 0$,
  so $\dupr{\omega}{a} = \dupr{\omega}{\mu \Unit} = \mu$. Moreover, for any $b\in \A$
  the Cauchy--Schwarz inequality shows that
  \begin{align*}
    \abs[\big]{ \dupr{\omega}{a^*b} - \cc{\dupr{\omega}{a}} \dupr{\omega}{b} }^2
    =
    \abs[\big]{\dupr{\omega}{ (a-\mu \Unit)^* b }}^2
    \le
    \dupr[\big]{\omega}{ (a-\mu \Unit)^*(a-\mu \Unit)}
    \dupr{\omega}{ b^* b}
    =
    0
    \komma
  \end{align*}
  so $\dupr{\omega}{a^*b} = \cc{\dupr{\omega}{a}} \dupr{\omega}{b}$. By complex conjugation it follows that 
  $\dupr{\omega}{b^*a} = \cc{\dupr{\omega}{b}} \dupr{\omega}{a}$. We conclude that \refitem{item:eigenstates:eigenvalue}
  implies \refitem{item:eigenstates:mult}. As \refitem{item:eigenstates:mult} trivially implies \refitem{item:eigenstates:var}
  by choosing $b\coloneqq a$ and as \refitem{item:eigenstates:var} implies \refitem{item:eigenstates:eigenvalue}
  with $\mu \coloneqq \dupr{\omega}{a}$ because then
  $\dupr{\omega}{(a-\mu \Unit)^*(a-\mu \Unit)} = \dupr{\omega}{a^*a} - \abs{\dupr{\omega}{a}}^2 = 0$,
  these three statements are equivalent.
\end{proof}
The above Proposition~\ref{proposition:eigenstates} especially shows that the concept of eigenstates on a quasi-ordered $^*$\=/algebra $\A$
can be seen as a weakening of positive unital $^*$\=/homomorphisms from $\A$ to $\CC$. More precisely, a state $\omega$ on $\A$
is a positive unital $^*$\=/homomorphism if and only if it is a common eigenstate of all elements of $\A$. The next example shows that 
eigenstates can also be interpreted as generalizations of vector states associated to eigenvectors:

\begin{example} \label{example:eigenstates}
  Consider the ordered $^*$\=/algebra of operators $\Adbar(\Dom)$ on a pre-Hilbert space $\Dom$.
  For every $\psi \in \Dom$, the vector functional $\chi_\psi \colon \Adbar(\Dom)\to \CC$,
  $a \mapsto \dupr{\chi_\psi}{a} \coloneqq \skal{\psi}{a(\psi)}$ is Hermitian and positive,
  and it is a state if and only if $\seminorm{}{\psi}=1$.
  Now let $a\in \Adbar(\Dom)$, $\psi \in \Dom$ with $\seminorm{}{\psi}=1$, and $\mu \in \CC$
  be given. Then the statement $\seminorm{}{a(\psi) - \mu \psi} = 0$ is equivalent to
  $a(\psi) = \mu \psi$ and also equivalent to
  $\dupr{\chi_\psi}{(a-\mu\Unit)^*(a-\mu\Unit)} = 0$. This shows that $\psi$ is an eigenvector
  of $a$ with eigenvalue $\mu$ if and only if $\chi_\psi$ is an eigenstate of $a$ with eigenvalue $\mu$.
\end{example}
Moreover, for a normal element $a$ of a $C^*$\=/algebra $\A$
one finds that eigenstates exist precisely to eigenvalues which are elements of the spectrum of $a$,
see \cite{rinehart:eigenstatesOfCStarAlgebras}, essentially because all $\CC$-valued unital $^*$\=/homomorphisms
of the commutative $C^*$\=/subalgebra of $\A$ that is generated by $a$ can be extended to states on $\A$ by a Hahn--Banach type argument.
However, we will see in Proposition~\ref{proposition:noeigenstate} that this does not generalize to the unbounded case.

\section{Reduction of Representable Poisson \texorpdfstring{$^*$-Algebras}{*-Algebras}} 
\label{sec:reduction:general}

\subsection{Representable Poisson \texorpdfstring{$^*$\=/Algebras}{*-Algebras}} 
In applications to classical or quantum physics, ordered $^*$\=/algebras appear as the algebras of observables,
with their order induced by their states.
Such observable algebras are usually endowed with a Poisson bracket: in quantum physics, this Poisson bracket is derived from the commutator,
but in classical physics, it is an additional structure on the algebra. This leads to:
\begin{definition} \label{definition:repPoiSAlg}
  A \emph{representable Poisson $^*$\=/algebra} is an ordered $^*$\=/algebra $\A$
  whose order is induced by its states and
  that is equipped with a bilinear and antisymmetric \emph{Poisson bracket}
  $\poi{\argument}{\argument} \colon \A \times \A \to \A$ fulfilling the usual Leibniz and Jacobi identity
  and which is compatible with the $^*$\=/involution in the sense that $\poi{a}{b}^* = \poi{a^*}{b^*}$ holds for all $a,b\in \A$.
  Given two representable Poisson $^*$\=/algebras $\A$ and $\B$ and a unital $^*$\=/homomorphism
  $\Phi \colon \A \to \B$, then $\Phi$ is said to be \emph{compatible with Poisson brackets} if
  $\Phi\big(\poi{a_1}{a_2}\big) = \poi[\big]{\Phi(a_1)}{\Phi(a_2)}$ holds for all $a_1,a_2 \in \A$.
\end{definition}
Recall that in the general non-commutative case, where the order of the factors is important,
the Leibniz identity for $a,b,c\in \A$ is $\poi{ab}{c} = \mathop{\poi{a}{c}} b + 
a\mathop{\poi{b}{c}}$.
The algebras defined above are ``representable'' in the following sense: Since the order is induced 
by the states, the underlying ordered $^*$\=/algebra admits a faithful representation as an 
$O^*$\=/algebra on a pre-Hilbert space,
which can be constructed as an orthogonal sum of GNS-representations 
\cite[Chapter 4.4]{schmuedgen:invitationToStarAlgebras}.
We are especially interested in two types of representable Poisson $^*$\=/algebras:
\begin{example} \label{example:comaspoi}
  Let $\A$ be an ordered $^*$\=/algebra whose order is induced by its states, e.g.~an 
  $O^*$\=/algebra on some pre-Hilbert space $\Dom$,
  and let $\hbar \in \RR\setminus \{0\}$. Then $\A$ with the rescaled commutator as Poisson bracket,
  \begin{align}
    \poi{a}{b} \coloneqq \frac{ab-ba}{\I \hbar}
    \label{eq:comaspoi}
  \end{align}
  for all $a,b\in \A$, is a representable Poisson $^*$\=/algebra. All unital $^*$\=/homomorphisms
  between such representable Poisson $^*$\=/algebras (with the same value of $\hbar$) are automatically
  compatible with Poisson brackets.
\end{example}
If the underlying $^*$\=/algebra of a representable Poisson $^*$\=/algebra $\A$
is sufficiently non-commutative, then there exist some general conditions under which the Poisson bracket of $\A$
necessarily is of the form \eqref{eq:comaspoi}, see \cite{farkas.letzter:ringTheoryFromSymplecticGeometry}.
For this reason, more general notions of ``non-commutative Poisson algebras'' like in 
\cite{xu:noncommutativePoissonAlgebras} have been developed, an approach that we, however, will not pursue any further.
\begin{example} \label{example:Poissonmanifold}
  Let $\A$ be a commutative ordered $^*$\=/algebra  whose order is induced by its states, 
  e.g.~$\A = \Smooth(M)$,
  the smooth $\CC$-valued functions on a smooth manifold $M$ with the pointwise order. Then any bilinear and antisymmetric bracket on the real subalgebra
  $\A_\Hermitian$ of $\A$ which fulfils Leibniz and Jacobi identity gives rise to a Poisson bracket $\poi{\argument}{\argument}$
  on whole $\A$ (by $\CC$-linear extension) with which $\A$ becomes a representable Poisson $^*$\=/algebra. In the case $\A = \Smooth(M)$
  such a bracket can always be derived from a (real) Poisson tensor with which $M$ becomes a 
  Poisson manifold.
\end{example}
One might wonder why Definition~\ref{definition:repPoiSAlg} does not require any form of compatibility between the Poisson bracket
and the order. The reason for this is that in the case of smooth functions on a Poisson manifold as in Example~\ref{example:Poissonmanifold} there
does not seem to be any such compatibility.

If $\A$ is a representable Poisson $^*$\=/algebra and $\B$ a unital $^*$\=/subalgebra of $\A$ which is closed under the Poisson bracket,
then it is easy to check that $\B$ with the operations and the order inherited from $\A$ is again a representable Poisson $^*$\=/algebra,
because all positive Hermitian linear functionals on $\A$ can be restricted to $\B$. Quotients, however, are somewhat less well-behaved
as we will shortly see:
\begin{definition} \label{definition:repPoiSIdeal}
  Let $\A$ be a representable Poisson $^*$\=/algebra. A subset $\mathcal{I}$ of $\A$ is a \emph{representable Poisson $^*$\=/ideal} if $\mathcal{I}$
  is a $^*$\=/ideal of $\A$ which is also a Poisson ideal, i.e.~$\poi{a}{b} \in \mathcal{I}$ for all $a\in \A$ and all $b\in \mathcal{I}$,
  and if additionally $\mathcal{I}$ arises as the common kernel of a set of states on $\A$,
  i.e.~for all $a\in \mathcal{A} \setminus \mathcal{I}$ there exists $\omega \in \States(\A)$ for which
  $\dupr{\omega}{a} \neq 0$ and $\mathcal{I} \subseteq \ker \omega$ hold.
\end{definition}
For example, if $\Phi \colon \A \to \B$ is a positive unital $^*$\=/homomorphism between representable Poisson $^*$\=/algebras $\A$ and $\B$
and compatible with Poisson brackets, then $\ker \Phi = \set{a\in \A}{\Phi(a) =0}$ certainly is a $^*$\=/ideal and a Poisson ideal of $\A$,
and it is the common kernel of a set of states on $\A$, hence a representable Poisson $^*$\=/ideal:
Given $a \in \mathcal{A} \setminus \ker \Phi$, then $\Phi(a)\neq 0$ so that there exists $\rho \in \States(\B)$ with $\dupr{\rho}{\Phi(a)} \neq 0$
because the order on $\B$ is induced by its states. Consequently, $\rho \circ \Phi \in \States(\A)$
fulfils $\dupr{\rho\circ\Phi}{a} \neq 0$.
\begin{proposition} \label{proposition:repPoiQuotient}
  Let $\A$ be a representable Poisson $^*$\=/algebra, $\mathcal{I}$ a representable Poisson $^*$\=/ideal of $\A$,
  and $[\argument] \colon \A \to \A / \mathcal{I}$ the canonical projection onto the quotient $^*$\=/algebra $\A / \mathcal{I}$.
  Then the Poisson bracket descends to $\A / \mathcal{I}$, i.e.~there exists a (necessarily unique)
  Poisson bracket $\poi{\argument}{\argument}$ on $\A / \mathcal{I}$ fulfilling $\poi{[a]}{[b]} = [\poi{a}{b}]$
  for all $a,b \in \A$. The quotient $^*$\=/algebra $\A/\mathcal{I}$, endowed with this 
  Poisson bracket and with the order whose quadratic module of positive Hermitian elements is
  \begin{align}
    (\A / \mathcal{I})^+_\Hermitian
    \coloneqq
    \set[\big]{
      [a]
    }{
      a\in \A_\Hermitian\textup{ such that }\dupr{\omega}{a} \ge 0\textup{ for all }\omega \in \States(\A)\textup{ for which }\mathcal{I} \subseteq \ker \omega
    }
    ,
    \label{eq:repPoiQuotient:order}
  \end{align}
  becomes a representable Poisson $^*$\=/algebra.
  This way, the projection $[\argument] \colon \A \to \A / \mathcal{I}$ becomes a surjective positive unital $^*$\=/homomorphism compatible with Poisson brackets.
  Moreover, whenever a state $\omega$ on $\A$ fulfils $\ker [\argument] \subseteq \ker \omega$, then the unique algebraic state $\check{\omega}$ on 
  $\A / \mathcal{I}$ that fulfils $\omega = \check{\omega} \circ [\argument]$ is positive, hence a state on $\A / \mathcal{I}$.
\end{proposition}
\begin{proof}
  The quotient $^*$\=/algebra $\A / \mathcal{I}$ with the order defined by \eqref{eq:repPoiQuotient:order} is an ordered $^*$\=/algebra
  whose order is induced by its states, because it is obtained by the construction
  of Proposition~\ref{proposition:inducedOrder} with $S \coloneqq \set{\omega\in\States(\A)}{\mathcal{I} \subseteq \ker \omega}$;
  the condition $a \acts \omega \in S$ for all $\omega \in S$ and all $a\in \A$ with $\dupr{\omega}{a^*a} = 1$
  holds because $\mathcal{I}$ is a $^*$\=/ideal. The Poisson bracket on $\A / \mathcal{I}$ is well-defined because $\mathcal{I}$ is a Poisson ideal,
  and it is easy to check that this way $\A/\mathcal{I}$ becomes a representable Poisson $^*$\=/algebra
  and that the canonical projection $[\argument] \colon \A \to \A / \mathcal{I}$ becomes a surjective positive unital $^*$\=/homomorphism
  compatible with Poisson brackets. Given any state $\omega$ on $\A$ that fulfils $\ker [\argument] \subseteq \ker \omega$,
  then $\mathcal{I} \subseteq \ker \omega$ and it is clear that there exists a unique algebraic state $\check{\omega}$ on 
  $\A / \mathcal{I}$ that fulfils $\omega = \check{\omega} \circ [\argument]$. 
  It is an immediate consequence of \eqref{eq:repPoiQuotient:order} that $\check{\omega}$ is also positive, hence a state.
\end{proof}
Note that the order from~\eqref{eq:repPoiQuotient:order} in general does not coincide with the quotient order of $^*$\=/algebras
defined in Section~\ref{sec:preliminaries:orderedsalg}; it is the smallest order induced by states that contains the quotient order.
While the construction of the \emph{quotient representable Poisson $^*$\=/algebra} from the above 
Proposition~\ref{proposition:repPoiQuotient}
has all the properties that one would expect from an abstract point of view, there is a problem within the definition
of representable Poisson $^*$\=/ideals: Without any compatibility between Poisson bracket and order, it is unclear how
such a representable Poisson $^*$\=/ideal $\mathcal{I}$ can be constructed explicitly in the general case, because it simultaneously has
to be a Poisson ideal and the common kernel of a set of states. The solution to this problem depends on
the example at hand: In the non-commutative case of Example~\ref{example:comaspoi}, every 
$^*$\=/ideal automatically is a Poisson ideal,
while in the commutative case of Example~\ref{example:Poissonmanifold} one can often apply geometric arguments, which will be
discussed further in Section~\ref{sec:PoiMan}.

\subsection{Symmetry Reduction}
From the compatibility between Poisson bracket and $^*$\=/involution it follows that the real linear subspace of Hermitian elements of a representable Poisson
$^*$\=/algebra with the restriction of the Poisson bracket is especially a Lie algebra. This leads to a notion of well-behaved actions of real Lie algebras
on representable Poisson $^*$\=/algebras:
\begin{definition} \label{definition:momentmap}
  Let $\A$ be a representable Poisson $^*$\=/algebra and $\lie g$ a real Lie algebra. Then a 
  \emph{momentum map} from $\lie g$ to $\A$ is a morphism
  $\momentmap \colon \lie g \to \A_\Hermitian$ of real Lie algebras (with respect to the Poisson bracket on $\A_\Hermitian$).
  Given such a momentum map, then we define the \emph{induced right action}  
  $\argument\racts\argument \colon \A \times \lie g \to \A$,
  \begin{align}
    (a,\xi) \mapsto a \racts \xi \coloneqq \poi[\big] {a}{\momentmap(\xi)}
    \punkt
  \end{align}
\end{definition}
One can easily check that $\argument\racts\argument$ is indeed a right action
of the Lie algebra $\lie g$ on $\A$, i.e.\ that
\begin{align}
  \big((a \racts \xi) \racts \eta\big) - \big((a \racts \eta) \racts \xi\big)
  =
  a \racts \kom{\xi}{\eta}
\end{align}
holds for all $a\in \A$ and all $\xi, \eta \in \lie g$, where $\kom{\argument}{\argument}$ denotes the Lie bracket of $\mathfrak{g}$.
This right action is also compatible with the Poisson bracket in the sense that
\begin{align}
 \poi{a}{b} \racts \xi
 =
 \poi[\big]{a\racts \xi}{b} + \poi[\big]{a}{b\racts \xi}
 \label{eq:actpoicomp}
\end{align}
holds for all $a,b\in \A$ and all $\xi \in \lie g$.
\begin{definition}
  Let $V$ be a vector space endowed with a right action $\argument \racts \argument \colon V \times \lie g \to V$
  of a Lie algebra $\lie{g}$, then 
  \begin{align}
    V^{\lie g}
    \coloneqq
    \set[\big]{
      v \in V
    }{
      \forall_{\eta \in \lie g} : v \racts \eta = 0
    }
  \end{align}
  denotes its linear subspace of \emph{$\lie g$-invariant} elements.
\end{definition}
While in similar settings there do exist reduction procedures for e.g.~free and proper actions of arbitrary Lie groups, like Marsden--Weinstein reduction
or the reduction of formal star products via BRST cohomology from \cite{bordemann.herbig.waldmann:BRSTCohomologyAndPhaseSpaceReduction}, we will only
consider the simpler case of abelian Lie groups, in which case the Lie bracket of the associated 
Lie algebra is zero (yet we will consider arbitrary momenta).
\begin{proposition} \label{proposition:ginv}
  Let $\A$ be a representable Poisson $^*$\=/algebra and $\momentmap \colon \lie g \to \A_\Hermitian$ a momentum map for a real Lie algebra $\lie g$.
  Then $\A^{\lie g}$ with the restriction of the relevant operations of $\A$ and the restricted 
  order is again a respresentable Poisson $^*$\=/algebra.
  Moreover, if $\lie{g}$ is commutative, then $\momentmap(\xi) \in \A^{\lie g}$ for all $\xi \in \lie g$.
\end{proposition}
\begin{proof}
  As $\xi \in \lie g$ acts on $\A$ by an inner derivation $\poi{\argument}{\momentmap(\xi)}$ with Hermitian $\momentmap(\xi)$,
  their common kernel $\A^{\lie g}$ is a unital $^*$\=/subalgebra of $\A$. From 
  \eqref{eq:actpoicomp} it follows that $\A^{\lie g}$ is also
  closed under the Poisson bracket. Since the restricted order on $\A^{\lie g}$ is still induced by 
  its states, $\A^{\lie g}$ is a representable Poisson $^*$\=/algebra.
  If $\lie g$ is commutative, then $\momentmap(\xi) \racts \eta = \poi{\momentmap(\xi)}{\momentmap(\eta)} = \momentmap(\kom{\xi}{\eta}) = 0$
  holds for all $\xi,\eta \in \lie g$, so $\momentmap(\xi) \in \A^{\lie g}$.
\end{proof}
Restriction to $\A^{\lie g}$ is the first step in the reduction procedure, the second step is to divide out a suitable representable Poisson $^*$\=/ideal $\mathcal{I}_\mu$,
interpreted as the ``vanishing ideal of the levelset $\Levelset_\mu$ of the momentum map $\momentmap$ at $\mu$''. Of course, the concept of a levelset or vanishing ideal is not applicable
in the general case considered here, especially not for $O^*$\=/algebras like in Example~\ref{example:comaspoi}.
For the general definition of the reduction we therefore fall back to requiring a universal property to be fulfilled,
which essentially reduces to a characterization of $\mathcal{I}_\mu$ as the smallest representable Poisson $^*$\=/ideal
that contains the image of $\momentmap-\mu$.
An alternative description of the reduction as a quotient by an actual ``non-commutative 
vanishing ideal'' will be given in Theorem~\ref{theorem:vanishing:neu}.
\begin{definition} \label{definition:reduction}
  Let $\A$ be a representable Poisson $^*$\=/algebra and $\momentmap \colon \lie g \to \A_\Hermitian$ a momentum map for a commutative real Lie algebra $\lie g$.
  Given $\mu \in {\lie g}^*$, then the \emph{$\momentmap$-reduction of $\A$ at $\mu$} is a tuple
  $(\A_{\mu\mred}, [\argument]_\mu)$ of a representable Poisson $^*$\=/algebra $\A_{\mu\mred}$ and a positive
  unital $^*$\=/homomorphism $[\argument]_\mu \colon \A^{\lie g} \to \A_{\mu\mred}$ compatible with Poisson brackets
  and which fulfils $[\momentmap(\xi)]_\mu = \dupr{\mu}{\xi}\Unit$ for all $\xi \in \lie g$,
  such that the following universal property is fulfilled:
  Whenever $\Phi \colon \A^{\lie g} \to \B$ is another positive unital $^*$\=/homomorphism compatible with Poisson brackets
  into any representable Poisson $^*$\=/algebra $\B$ that fulfils $\Phi\big( \momentmap(\xi) \big) = \dupr{\mu}{\xi}\Unit$
  for all $\xi \in \lie g$, then there exists a unique positive unital $^*$\=/homomorphism $\Phi_{\mu\mred} \colon \A_{\mu\mred} \to \B$
  compatible with Poisson brackets for which $\Phi = \Phi_{\mu\mred} \circ [\argument]_\mu$ holds.
\end{definition}
Note that the $\momentmap$-reduction at $\mu$ (once we have shown that it exists) is determined up to unique isomorphism.
The existence of the reduction is indeed guaranteed, which will be shown in 
Theorem~\ref{theorem:reduction}.

\begin{definition} \label{definition:genIdeal}
  Let $\A$ be a representable Poisson $^*$\=/algebra, $\momentmap \colon \lie g \to \A_\Hermitian$ a momentum map for a
  commutative real Lie algebra $\lie g$, and $\mu \in {\lie g}^*$. Then denote by $\genSId{\momentmap-\mu}{}$ the $^*$\=/ideal
  of $\A^{\lie g}$ that is generated by all $\momentmap(\xi) - \dupr{\mu}{\xi} \Unit$ with $\xi \in \lie g$.
\end{definition}

\begin{theorem} \label{theorem:reduction}
  Let $\A$ be a representable Poisson $^*$\=/algebra, $\momentmap \colon \lie g \to \A_\Hermitian$ a momentum map for a commutative real Lie algebra $\lie g$,
  and $\mu \in {\lie g}^*$. Then the intersection
  \begin{align}
    \mathcal{I}_\mu
    \coloneqq
    \bigcap \set[\big]{
      \mathcal{I}
    }{
      \mathcal{I} \text{ is a representable Poisson $^*$\=/ideal of }\A^{\lie g}\text{ fulfilling }\genSId{\momentmap-\mu}{} \subseteq \mathcal{I}
    }
    .
    \label{eq:Imu}
  \end{align}
  is a well-defined representable Poisson $^*$\=/ideal of $\A^{\lie g}$. Moreover, given a representable Poisson $^*$\=/algebra $\A_{\mu\mred}$ and a
  positive unital $^*$\=/homomorphism $[\argument]_\mu \colon \A^{\lie g} \to \A_{\mu\mred}$ compatible with Poisson brackets,
  then $(\A_{\mu\mred}, [\argument]_\mu)$ is the $\momentmap$-reduction of $\A$ at $\mu$ if and only if the following two conditions are fulfilled:
  \begin{enumerate}
    \item $[\argument]_\mu$ is surjective and its kernel is $\ker{[\argument]_\mu} = \mathcal{I}_\mu$. \label{item:reduction:projection}
    \item Whenever $\omega \in \States(\A^{\lie g})$ fulfils $\ker{[\argument]_\mu} \subseteq \ker \omega$, then the unique algebraic state $\check{\omega}$
      on $\A_{\mu\mred}$ that fulfils $\omega = \check{\omega} \circ [\argument]_\mu$ is positive, hence a state on $\A_{\mu\mred}$. \label{item:reduction:states}
  \end{enumerate}
  Finally, the $\momentmap$-reduction of $\A$ at $\mu$ always exists and can e.g.~be realized as the quotient representable Poisson $^*$\=/algebra
  $\A_{\mu\mred} \coloneqq \A^{\lie g} / \mathcal{I}_\mu$ as in Proposition~\ref{proposition:repPoiQuotient} together with 
  the canonical projection onto the quotient $[\argument]_\mu \colon \A^{\lie g} \to \A^{\lie g} / \mathcal{I}_\mu$.
\end{theorem}

\begin{proof}
  $\A^{\lie g}$ itself is a representable Poisson $^*$\=/ideal of $\A^{\lie g}$ fulfilling $\genSId{\momentmap-\mu}{} \subseteq \A^{\lie g}$,
  so $\mathcal{I}_\mu$ as in \eqref{eq:Imu} is a well-defined subset of $\A^{\lie g}$, which 
  clearly is a $^*$\=/ideal and a Poisson ideal,
  and also is the common kernel of a set of states, hence a representable Poisson $^*$\=/ideal: For all
  $a\in \mathcal{A}^{\lie g} \setminus \mathcal{I}_\mu$ there is some representable Poisson $^*$\=/ideal $\mathcal{I}$ of $\A^{\lie g}$
  fulfilling $\genSId{\momentmap-\mu}{} \subseteq \mathcal{I}$ and $a \notin \mathcal{I}$,
  and therefore there is $\omega \in \States(\A^{\lie g})$ fulfilling $\dupr{\omega}{a} \neq 0$ and $\mathcal{I} \subseteq \ker \omega$,
  and especially also $\mathcal{I}_\mu \subseteq \ker \omega$.
  
  Now assume that $\A_{\mu\mred}$ and $[\argument]_\mu$ fulfil the two properties above and let 
  $\Phi \colon \A^{\lie g} \to \B$ be a 
  positive unital $^*$\=/homomorphism compatible with Poisson brackets into another representable Poisson $^*$\=/algebra $\B$ that
  fulfils $\Phi\big( \momentmap(\xi) \big) = \dupr{\mu}{\xi}\Unit_{\B}$ for all $\xi \in \lie g$. Then $\ker \Phi$ is a 
  representable Poisson $^*$\=/ideal of $\A^{\lie g}$ as discussed below 
  Definition~\ref{definition:repPoiSIdeal},
  and $\momentmap(\xi)-\dupr{\mu}{\xi}\Unit_{\A^{\lie g}} \in \ker \Phi$ holds for all $\xi \in\lie g$, hence $\genSId{\momentmap-\mu}{} \subseteq \ker \Phi$.
  It follows that $\mathcal{I}_\mu \subseteq \ker \Phi$ and as a consequence of property~\refitem{item:reduction:projection},
  there exists a unique unital $^*$\=/homomorphism $\Phi_{\mu\mred} \colon \A_{\mu\mred} \to \B$ compatible with Poisson brackets
  such that $\Phi = \Phi_{\mu\mred} \circ [\argument]_\mu$. Moreover, every state $\rho$ on $\B$ can be pulled back to a state
  $\rho \circ \Phi = \rho \circ \Phi_{\mu\mred} \circ [\argument]_\mu$ on $\A^{\lie g}$, and property~\refitem{item:reduction:states}
  then implies that $\rho \circ \Phi_{\mu\mred}$ is a state on $\A_{\mu\mred}$. By Proposition~\ref{proposition:positivitycriterium}, this
  means that $\Phi_{\mu\mred}$ is positive, and we conclude that $(\A_{\mu\mred}, [\argument]_\mu)$ fulfils the universal property
  of the $\momentmap$-reduction at $\mu$.
  
  Finally, the quotient representable Poisson $^*$\=/algebra $\A_{\mu\mred} \coloneqq \A^{\lie g} / \mathcal{I}_\mu$ together with the
  canonical projection $[\argument]_\mu \colon \A^{\lie g} \to \A^{\lie g} / \mathcal{I}_\mu$ clearly fulfils the
  first property, and it fulfils the second one by Proposition~\ref{proposition:repPoiQuotient}. So $(\A^{\lie g} / \mathcal{I}_\mu, [\argument]_\mu)$
  is the $\momentmap$-reduction of $\A$ at $\mu$. This also means that for every other realization
  $(\A_{\mu\mred}^\sim, [\argument]_\mu^\sim)$ of the $\momentmap$-reduction at $\mu$ there exist two mutually inverse positive unital $^*$\=/homomorphisms
  $\phi \colon \A_{\mu\mred}^\sim \to \A^{\lie g} / \mathcal{I}_\mu$ and $\psi \colon \A^{\lie g} / \mathcal{I}_\mu \to \A_{\mu\mred}^\sim$
  fulfilling $\phi \circ [\argument]_\mu^\sim = [\argument]_\mu$ and $\psi \circ [\argument]_\mu = [\argument]_\mu^\sim$ and which
  are compatible with Poisson brackets. Using these it is easy to check that $\A_{\mu\mred}^\sim$ and $[\argument]_\mu^\sim$ also
  fulfil properties \refitem{item:reduction:projection} and \refitem{item:reduction:states}.
\end{proof}
In order to determine the reduction of a representable Poisson $^*$\=/algebra it therefore is crucial to determine the representable
Poisson $^*$\=/ideal $\mathcal{I}_\mu$ constructed above. This, however, might be a rather hard task in general without any compatibility
between Poisson bracket and order available.

Finally, we note that in the non-commutative Example \ref{example:comaspoi} the reduction can also 
be characterized via its representations:

\begin{example} \label{example:representationsFactorThroughReduction:1}
  Let $\A$ be a representable Poisson $^*$\=/algebra like in Example~\ref{example:comaspoi}, i.e.~assume that there exists $\hbar \in \RR \setminus \{0\}$
  such that the Poisson bracket on $\A$ fulfils $\poi{a}{b} = (ab-ba)/(\I \hbar)$ for all $a,b\in \A$. Moreover, let $\momentmap \colon \lie g \to \A_\Hermitian$
  be a momentum map for a commutative real Lie algebra $\lie g$ and $\mu \in \lie g^*$, and let 
  $(\A_{\mu\mred}, [\argument]_\mu)$ be the $\momentmap$-reduction
  of $\A$ at $\mu$. Then the Poisson bracket on $\A_{\mu\mred}$ is again derived from the commutator, more precisely,
  \begin{align*}
    \poi[\big]{[a]_\mu}{[b]_\mu} = \big[\poi{a}{b}\big]_\mu = \frac{[ab-ba]_\mu}{\I\hbar} = \frac{[a]_\mu[b]_\mu-[b]_\mu[a]_\mu}{\I\hbar} .
  \end{align*}
  for all $a,b\in \A^{\lie g}$.
  Consider the quotient $^*$\=/algebra $\A^{\lie g} / \genSId{\momentmap-\mu}$ with the quotient order and
  $[\argument] \colon\A^{\lie g} \to \A^{\lie g} / \genSId{\momentmap-\mu}$ the canonical projection, 
  then $\genSId{\momentmap-\mu} \subseteq \ker [\argument]_\mu$ so that
  there exists a unique positive unital $^*$\=/homomorphism $\iota \colon \A^{\lie g} / \genSId{\momentmap-\mu} \to \A_{\mu\mred}$
  fulfilling $\iota \circ [\argument] = [\argument]_\mu$. Now let $\Dom$ be any pre-Hilbert space
  and $\Phi \colon \A^{\lie g} / \genSId{\momentmap-\mu} \to \Adbar(\Dom)$ any positive unital $^*$\=/homomorphism, i.e.~any representation of
  $\A^{\lie g} / \genSId{\momentmap-\mu}$ on $\Dom$. Then $\Phi \circ [\argument]$ is a positive unital $^*$\=/homomorphism from 
  $\A^{\lie g}$ to $\Adbar(\Dom)$ which is compatible with Poisson brackets if $\Adbar(\Dom)$ is also endowed with the commutator bracket
  \eqref{eq:comaspoi}, and $\Phi([\momentmap(\xi)]) = \dupr{\mu}{\xi} \Unit$ holds for all $\xi \in \lie g$ because
  $\momentmap(\xi)-\dupr{\mu}{\xi}\Unit \in \genSId{\momentmap-\mu}$. By definition of the reduction, there exists a unique
  positive unital $^*$\=/homomorphism $(\Phi\circ [\argument])_{\mu\mred} \colon \A_{\mu\mred} \to \Adbar(\Dom)$
  such that $(\Phi\circ [\argument])_{\mu\mred} \circ [\argument]_\mu = \Phi \circ [\argument]$, hence
  $(\Phi\circ [\argument])_{\mu\mred} \circ \iota = \Phi$. In this sense, every representation $\Phi$ of $\A^{\lie g} / \genSId{\momentmap-\mu}$
  factors through $\A_{\mu\mred}$. This property even characterizes $\A_{\mu\mred}$ completely because
  $\A_{\mu\mred}$ admits a faithful representation as discussed under Definition~\ref{definition:repPoiSAlg}.
\end{example}

\subsection{Non-commutative Vanishing Ideals}
There is an important special case in which the representable Poisson $^*$\=/ideal 
$\mathcal{I}_\mu$ of Theorem~\ref{theorem:reduction}
can be described explicitly as a ``non-commutative vanishing ideal''. Before discussing this, 
however, we need some definitions:

\begin{definition}
  Let $\A$ be a representable Poisson $^*$\=/algebra. Then we say that \emph{Poisson-commuting elements in $\A$ commute} if
  $ab = ba$ holds for all $a,b\in \A$ that fulfil $\poi{a}{b} = 0$.
\end{definition}
It is immediately clear that Poisson-commuting elements commute in Example~\ref{example:comaspoi}, where the Poisson bracket is induced by the commutator,
and in Example~\ref{example:Poissonmanifold} of commutative ordered $^*$\=/algebras with an 
arbitrary Poisson bracket. However, there are also more pathological examples which do not have this property:
Take e.g.~any non-commutative ordered $^*$-algebra whose order is induced by its states, and endow it with the zero Poisson bracket.

\begin{definition}
  Let $\A$ be a representable Poisson $^*$\=/algebra and $\momentmap \colon \lie{g} \to \A_\Hermitian$ a momentum map for a commutative real Lie algebra $\lie{g}$.
  Following the notation introduced in Definition~\ref{definition:eigenstates}, 
  the sets of common eigenstates of all $\momentmap(\xi)$, $\xi \in \lie{g}$, with eigenvalues $\dupr{\mu}{\xi}$ will be denoted by
  \begin{align}
    \States_{\momentmap\!,\mu}(\A) \coloneqq \bigcap_{\xi \in \lie g} \States_{\momentmap(\xi), \dupr{\mu}{\xi}}(\A)
    \quad\quad\text{and}\quad\quad
    \States_{\momentmap\!,\mu}(\A^{\lie g}) \coloneqq \bigcap_{\xi \in \lie g} \States_{\momentmap(\xi), \dupr{\mu}{\xi}}(\A^{\lie g})
    \,.
  \end{align}
\end{definition}
If Poisson-commuting elements commute, then the $^*$\=/ideal $\genSId{\momentmap-\mu}{}$ is especially well-behaved:

\begin{proposition} \label{proposition:genIdeal}
  Let $\A$ be a representable Poisson $^*$\=/algebra in which Poisson-commuting elements commute, 
  $\momentmap \colon \lie{g} \to \A_\Hermitian$ a momentum map for a commutative real Lie algebra $\lie{g}$,
  and $\mu \in \lie g^*$. Then the $^*$\=/ideal $\genSId{\momentmap-\mu}{}$ of $\A^{\lie g}$ can explicitly be described as
  \begin{align}
    \genSId{\momentmap-\mu}{}
    =
    \set[\Big]{
      \sum\nolimits_{m=1}^M a_m \,\big(\momentmap(\xi_m) - \dupr{\mu}{\xi_m} \Unit\big)
    }{
      M\in \NN_0;\,
      a_1,\dots,a_M \in \A^{\lie g};\,
      \xi_1,\dots,\xi_M \in \lie g
    }
    \label{eq:genIdeal}
  \end{align}
  and $\genSId{\momentmap-\mu}{}$ automatically is also a Poisson ideal of $\A^{\lie g}$.
  Moreover, any state $\omega$ on $\A^{\lie g}$ fulfils $\genSId{\momentmap-\mu}{} \subseteq \ker \omega$
  if and only if $\omega \in \States_{\momentmap\!,\mu}(\A^{\lie g})$.
\end{proposition}
\begin{proof}
  The inclusion ``$\supseteq$'' in \eqref{eq:genIdeal} is clear. Conversely, the right-hand side
  of \eqref{eq:genIdeal} certainly contains $\momentmap(\xi)-\dupr{\mu}{\xi}\Unit$ for all $\xi \in \lie g$,
  and clearly is a linear subspace of $\A^{\lie g}$ and even a left ideal. 
  Now let $a\in \A^{\lie g}$ and $\xi \in \lie g$ be given, then 
  $\poi{a}{\momentmap(\xi)} = a\racts \xi = 0$ implies that $a\,\momentmap(\xi) = \momentmap(\xi)\,a$ because
  Poisson-commuting elements in $\A$ commute by assumption. Therefore
  \begin{align*}
    \big( a \,\big(\momentmap(\xi)-\dupr{\mu}{\xi}\Unit\big) \big)^*
    =
    \big( \big(\momentmap(\xi)-\dupr{\mu}{\xi}\Unit\big)\,a \big)^*
    =
    a^*\big(\momentmap(\xi)-\dupr{\mu}{\xi}\Unit\big)
  \end{align*}
  from which it follows that the right-hand side of \eqref{eq:genIdeal} is also stable under $\argument^*$, hence a $^*$\=/ideal,
  and consequently ``$\subseteq$'' holds in \eqref{eq:genIdeal}. Similarly, for $a,b\in \A^{\lie g}$ and $\xi \in \lie g$ one finds that
  \begin{align*}
    \poi[\big]{b}{ a \,\big(\momentmap(\xi)-\dupr{\mu}{\xi}\Unit\big) }
    &=
    \underbrace{\poi{b}{a}\big(\momentmap(\xi)-\dupr{\mu}{\xi}\Unit\big)}_{\in 
    \genSId{\momentmap-\mu}{}} \!
    {}+
    a \underbrace{\poi{b}{\momentmap(\xi)}}_{ = 0}\!
    {}-
    a\underbrace{\poi{b}{\dupr{\mu}{\xi}\Unit}}_{=0}
  \end{align*}
  because $\poi{b}{\momentmap(\xi)} = b\racts \xi = 0$. It thus follows from
  \eqref{eq:genIdeal} that $\genSId{\momentmap-\mu}{}$ is a Poisson ideal of $\A^{\lie g}$.
  
  Now let a state $\omega$ on $\A^{\lie g}$ be given. If $\genSId{\momentmap-\mu}{} \subseteq \ker \omega$,
  then especially $(\momentmap(\xi)-\dupr{\mu}{\xi}\Unit)^2 \in \ker \omega$ for all $\xi \in \lie g$, so $\omega \in \States_{\momentmap\!,\mu}(\A^{\lie g})$
  by Proposition~\ref{proposition:eigenstates}. Conversely, if $\omega \in \States_{\momentmap\!,\mu}(\A^{\lie g})$,
  then it follows from Proposition~\ref{proposition:eigenstates} that
  \begin{align*}
    \dupr[\big]{\omega}{a\,\big(\momentmap(\xi)-\dupr{\mu}{\xi}\Unit\big)}
    =
    \dupr{\omega}{a\,\momentmap(\xi)} - \dupr{\omega}{a}\dupr{\mu}{\xi}
    =
    \dupr{\omega}{a} \big( \dupr{\omega}{\momentmap(\xi)} - \dupr{\mu}{\xi} \big)
    =
    0
  \end{align*}
  holds for all $a\in \A^{\lie g}$ and all $\xi \in \lie g$, and therefore $\genSId{\momentmap-\mu}{} \subseteq \ker \omega$ by \eqref{eq:genIdeal}.
\end{proof}
However, $\genSId{\momentmap-\mu}{}$ is not necessarily the common kernel of a set of states,
hence in general not a representable Poisson $^*$\=/ideal. In those cases where 
$\genSId{\momentmap-\mu}{}$ is a representable Poisson $^*$\=/ideal, it coincides with 
$\mathcal{I}_\mu$ as an immediate
consequence of the definition of $\mathcal{I}_\mu$ in Theorem~\ref{theorem:reduction}.

As common eigenstates of the momentum map $\momentmap$ with eigenvalues $\mu$ are precisely the states that vanish on
the $^*$\=/ideal generated by $\momentmap-\mu$ by the above Proposition~\ref{proposition:genIdeal}, one might be tempted to
interpret these as a generalization of the evaluation functionals on the levelset $\Levelset_\mu$ of $\momentmap$ at $\mu$ in
the geometric approach to symmetry reduction. This idea leads to:

\begin{definition} \label{definition:regular}
  Let $\A$ be a representable Poisson $^*$\=/algebra in which Poisson-commuting elements commute,
  $\momentmap \colon \lie g \to \A_\Hermitian$ a momentum map for a commutative real Lie algebra $\lie g$, and $\mu \in {\lie g}^*$.
  We write
  \begin{align}
    \mathcal{R}_\mu 
    &\coloneqq \set[\big]{a\in \A^{\lie g}_\Hermitian }{\dupr{\omega}{a} \ge 0 \textup{ for all }\omega \in \States_{\momentmap\!,\mu}(\A^{\lie g})}
    \label{eq:regular:RQm}
  \intertext{and define the \emph{vanishing ideal} of $\momentmap$ at $\mu$ as}
    \mathcal{V}_\mu
    &\coloneqq
    \set[\big]{a\in\A^{\lie g}}{\dupr{\omega}{a} = 0 \textup{ for all }\omega \in \States_{\momentmap\!,\mu}(\A^{\lie g})}
    \punkt
    \label{eq:regular:VIdeal}
  \end{align}
  We say that $\mu$ is \emph{regular} for $\momentmap$ if $\mathcal{V}_\mu$ is a Poisson ideal 
  of $\A^{\lie g}$.
\end{definition}
Note that $\mathcal{V}_\mu = \big(\mathcal{R}_\mu \cap (-\mathcal{R}_\mu)\big) + \I \big(\mathcal{R}_\mu \cap (-\mathcal{R}_\mu)\big)$.
In the special case of regular momenta, the ideal $\mathcal{I}_\mu$, which is the key to the $\momentmap$-reduction at $\mu$,
coincides with this vanishing ideal $\mathcal{V}_\mu$:

\begin{proposition} \label{proposition:vanishing}
  Let $\A$ be a representable Poisson $^*$\=/algebra in which Poisson-commuting elements commute,
  $\momentmap \colon \lie g \to \A_\Hermitian$ a momentum map for a commutative real Lie algebra $\lie g$, and $\mu \in {\lie g}^*$.
  Then $\mathcal{R}_\mu$ is a quadratic module of $\A^{\lie g}$ and $\mathcal{V}_\mu$ is a $^*$\=/ideal of $\A^{\lie g}$.
  Moreover, let $\mathcal{I}_\mu$ be the representable Poisson $^*$\=/ideal defined in Theorem~\ref{theorem:reduction}.
  Then $\genSId{\momentmap-\mu}{} \subseteq \mathcal{V}_\mu \subseteq \mathcal{I}_\mu$ holds, and
  if $\mu$ additionally is regular for $\momentmap$, then even $\mathcal{V}_\mu = \mathcal{I}_\mu$.
\end{proposition}
\begin{proof}
  Given $\omega \in \States_{\momentmap\!,\mu}(\A^{\lie g})$ and $a\in \A^{\lie g}$ with $\dupr{\omega}{a^*a} = 1$, then $a\acts \omega$ clearly is a state on $\A^{\lie g}$.
  Moreover, $a^*( \momentmap(\xi) - \dupr{\mu}{\xi} \Unit )^2a \in \genSId{\momentmap-\mu}{}$ holds for all $\xi \in \lie g$,
  and as $\genSId{\momentmap-\mu}{} \subseteq \ker \omega$ by Proposition~\ref{proposition:genIdeal}, this implies
  $\dupr{a\acts \omega}{ ( \momentmap(\xi) - \dupr{\mu}{\xi} \Unit )^2 } = 0$ for all $\xi \in \lie g$, i.e.~$a\acts \omega \in \States_{\momentmap\!,\mu}(\A^{\lie g})$.
  Proposition~\ref{proposition:inducedOrder} now applies to $\A^{\lie g}$ and $S \coloneqq \States_{\momentmap\!,\mu}(\A^{\lie g})$
  and shows that $\mathcal{R}_\mu$ and $\mathcal{V}_\mu$ are a quadratic module and a $^*$\=/ideal of $\A^{\lie g}$, respectively.
  From $\genSId{\momentmap-\mu}{} \subseteq \ker \omega$ for all $\omega \in \States_{\momentmap\!,\mu}(\A^{\lie g})$
  it also follows that $\genSId{\momentmap-\mu}{} \subseteq \mathcal{V}_\mu$.
  
  Now given $a \in \A^{\lie g} \setminus \mathcal{I}_\mu$, then there exists a state $\omega$ on $\A^{\lie g}$
  with $\mathcal{I}_\mu \subseteq \ker \omega$ such that $\dupr{\omega}{a} \neq 0$ because $\mathcal{I}_\mu$
  is a representable Poisson $^*$\=/ideal of $\A^{\lie g}$.
  As $\genSId{\momentmap-\mu}{} \subseteq \mathcal{I}_\mu$ by definition of $\mathcal{I}_\mu$, this implies that
  $\omega \in \States_{\momentmap\!,\mu}(\A^{\lie g})$ by Proposition~\ref{proposition:genIdeal} again.
  But now it follows from $\dupr{\omega}{a} \neq 0$ that $a \notin \mathcal{V}_\mu$, and we conclude that $\mathcal{V}_\mu \subseteq \mathcal{I}_\mu$.
  
  Finally $\mathcal{V}_\mu$ by definition is the common kernel of a set of states on $\A^{\lie g}$, and if $\mu$ is regular for $\momentmap$,
  then $\mathcal{V}_\mu$ also is a Poisson ideal, hence a representable Poisson $^*$\=/ideal of $\A^{\lie g}$.
  It then follows from $\genSId{\momentmap-\mu}{} \subseteq \mathcal{V}_\mu$ that $\mathcal{I}_\mu \subseteq \mathcal{V}_\mu$,
  hence $\mathcal{V}_\mu = \mathcal{I}_\mu$.
\end{proof}
The most obvious case of a momentum $\mu$ that is regular for a momentum map $\momentmap$
is the one of Example~\ref{example:comaspoi} where the Poisson bracket is derived from the commutator,
because in this case, every $^*$\=/ideal is a Poisson ideal.
In Section~\ref{sec:PoiMan} we will discuss Poisson manifolds, and we show that all momenta are regular for the momentum map in this case, too.

\begin{corollary} \label{corollary:vanishing}
  Let $\A$ be a representable Poisson $^*$\=/algebra in which Poisson-commuting elements commute,
  $\momentmap \colon \lie g \to \A_\Hermitian$ a momentum map for a commutative real Lie algebra $\lie g$,
  and $\mu \in {\lie g}^*$ regular for $\momentmap$.
  Moreover, let $\omega$ be a state on $\A^{\lie g}$, then the equivalences
  \begin{align}
    \mathcal{I}_\mu \subseteq \ker \omega
    \quad\Longleftrightarrow\quad
    \mathcal{V}_\mu \subseteq \ker \omega
    \quad\Longleftrightarrow\quad
    \genSId{\momentmap-\mu} \subseteq \ker \omega
    \quad\Longleftrightarrow\quad
    \omega \in \States_{\momentmap\!,\mu}(\A^{\lie g})
  \end{align}
  hold.
\end{corollary}
\begin{proof}
  The previous Proposition~\ref{proposition:vanishing} shows that 
  $\mathcal{I}_\mu \subseteq \ker \omega\,\Longleftrightarrow\,
  \mathcal{V}_\mu \subseteq \ker \omega
  \,\Longrightarrow\,
  \genSId{\momentmap-\mu} \subseteq \ker \omega$  
  holds. 
  The last equivalence holds by Proposition~\ref{proposition:genIdeal}.
  If $\omega \in \States_{\momentmap\!,\mu}(\A^{\lie g})$, then $\mathcal{V}_\mu \subseteq \ker 
  \omega$ by definition of $\mathcal{V}_\mu$.
\end{proof}
The key to understanding the reduced algebra is to understand the common eigenstates of the momentum map,
which determine the quadratic module $\mathcal{R}_\mu$ and the non-commutative vanishing ideal $\mathcal{V}_\mu$.
This way, we can simplify the characterization of the reduced algebra from Theorem~\ref{theorem:reduction}:

\begin{theorem} \label{theorem:vanishing:neu}
	Let $\A$ be a representable Poisson $^*$\=/algebra in which Poisson-commuting elements commute,
	$\momentmap \colon \lie g \to \A_{\Hermitian}$ a momentum map for a commutative real Lie 
	algebra $\lie g$, and $\mu \in \lie g^*$ regular for $\momentmap$. Moreover, let
	$\A_{\mu\mred}$ be any representable Poisson $^*$\=/algebra and $[\argument]_\mu \colon \A^{\lie g} \to \A_{\mu\mred}$
	any positive unital $^*$\=/homomorphism compatible with Poisson brackets, then $(\A_{\mu\mred}, [\argument]_\mu)$ is the $\momentmap$-reduction
	of $\A$ at $\mu$ if and only if the following two conditions are fulfilled:
	\begin{enumerate}
	  \item $[\argument]_\mu$ is surjective and $\ker [\argument]_\mu = \mathcal{V}_\mu$.
	  \label{item:theorem:vanishing:1}
	  \item Whenever an element $a\in \A^{\lie g}_\Hermitian$ fulfils $[a]_\mu \in (\A_{\mu\mred})^+_\Hermitian$, then $a\in \mathcal{R}_\mu$.
	  \label{item:theorem:vanishing:2}
	\end{enumerate}
	Finally, if $(\A_{\mu\mred}, [\argument]_\mu)$ is the $\momentmap$-reduction of $\A$ at $\mu$, then
	$\mathcal{R}_\mu = \set[\big]{a\in\A^{\lie g}_\Hermitian}{[a]_\mu \in (\A_{\mu\mred})^+_\Hermitian}$.
\end{theorem}
\begin{proof}
  The two conditions given here  for $(\A_{\mu\mred}, [\argument]_\mu)$ being the $\momentmap$-reduction of $\A$ at $\mu$
  are equivalent to those from Theorem~\ref{theorem:reduction}:
  
  By Proposition~\ref{proposition:vanishing}, 
  $\mathcal{I}_\mu = \mathcal{V}_\mu$ holds so that the first condition here and in Theorem~\ref{theorem:reduction} are equivalent.
  Now assume that $[\argument]_\mu\colon \A^{\lie g} \to \A_{\mu\mred}$ is surjective with $\ker [\argument]_\mu = \mathcal{V}_\mu = \mathcal{I}_\mu$.
  Given a state $\omega$ on $\A^{\lie g}$, then Corollary~\ref{corollary:vanishing} shows that
  $\mathcal{I}_\mu \subseteq \ker \omega$ if and only if $\omega \in \States_{\momentmap\!,\mu}(\A^{\lie g})$.
  In this case, write $\check{\omega} \colon \A_{\mu\mred} \to \CC$ for the unique algebraic state on $\A_{\mu\mred}$
  that fulfils $\check{\omega}\circ [\argument]_\mu = \omega$. On the one hand, if every
  element $a\in \A^{\lie g}_\Hermitian$ with $[a]_\mu \in (\A_{\mu\mred})^+_\Hermitian$ fulfils $a\in \mathcal{R}_\mu$,
  then $\check\omega$ is positive on $\A_{\mu\mred}$ for every $\omega \in \States_{\momentmap\!,\mu}(\A^{\lie g})$, because
  $\dupr{\check\omega}{[a]_\mu} = \dupr{\omega}{a} \ge 0$ for all $a\in \mathcal{R}_\mu$, 
  and in particular for all $a\in \A^{\lie g}_\Hermitian$ with $[a]_\mu \in (\A_{\mu\mred})^+_\Hermitian$.
  On the other hand, if $\check\omega$ for every $\omega \in \States_{\momentmap\!,\mu}(\A^{\lie g})$ is positive on $\A_{\mu\mred}$,
  then every $a\in \A^{\lie g}_\Hermitian$ with $[a]_\mu \in (\A_{\mu\mred})^+_\Hermitian$ fulfils $a\in \mathcal{R}_\mu$
  because $\dupr{\omega}{a} = \dupr{\check\omega}{[a]_\mu} \ge 0$ for all $\omega \in \States_{\momentmap\!,\mu}(\A^{\lie g})$.
  
  Finally, if $(\A_{\mu\mred}, [\argument]_\mu)$ is the $\momentmap$-reduction of $\A$ at $\mu$, then
	$\mathcal{R}_\mu \supseteq \set{a\in\A^{\lie g}_\Hermitian}{[a]_\mu \in (\A_{\mu\mred})^+_\Hermitian}$
	by the second condition above. Conversely,
	as $[\argument]_\mu \colon \A^{\lie g} \to \A_{\mu\mred}$ is a positive unital $^*$\=/homomorphism,
	every $\check\omega \in \States(\A_{\mu\mred})$ can be pulled back to a state $\omega \coloneqq \check\omega \circ [\argument]_\mu$
	on $\A^{\lie g}$, and even $\omega\in\States_{\momentmap\!,\mu}(\A^{\lie g})$ by Corollary~\ref{corollary:vanishing}.
	So given $a\in \mathcal{R}_\mu$, then $\dupr{\check\omega}{[a]_\mu} = \dupr{\omega}{a} \ge 0$ for all $\check\omega \in \States(\A_{\mu\mred})$,
	which shows that $[a]_\mu \in (\A_{\mu\mred})^+_\Hermitian$.
\end{proof}
If Poisson-commuting elements commute and if $\mu\in\lie g^*$ is regular for $\momentmap$ so that $\mathcal{V}_\mu = \mathcal{I}_\mu$
by Proposition~\ref{proposition:vanishing}, then Theorems~\ref{theorem:reduction} and \ref{theorem:vanishing:neu} especially show that the
$\momentmap$-reduction of a representable Poisson $^*$\=/algebra $\A$ at $\mu$ can be constructed
as the quotient $^*$\=/algebra $\A_{\mu\mred} \coloneqq \A^{\lie g} / \mathcal{V}_\mu$ together with the canonical projection
$[\argument]_\mu \colon \A^{\lie g} \to \A_{\mu\mred}$ onto the quotient, and endowed with the order whose quadratic module of positive Hermitian elements
is $(\A_{\mu\mred})^+_\Hermitian \coloneqq \set{[a]_\mu}{a\in \mathcal{R}_\mu}$ and with the Poisson bracket on $\A_{\mu\mred}$
that is defined as $\poi{[a]_\mu}{[b]_\mu} \coloneqq [\poi{a}{b}]_\mu$ for all $a,b\in \A^{\lie g}$.

As an application we continue the discussion of representations as operators from Example~\ref{example:representationsFactorThroughReduction:1}:
\begin{example}
	Let $\A$ be a representable Poisson $^*$\=/algebra in which Poisson-commuting elements commute,
	$\momentmap \colon \lie g \to \A_\Hermitian$ a momentum map for a commutative real Lie algebra 
	$\lie g$,	and $\mu \in {\lie g}^*$.
	Like in Example \ref{example:representationsFactorThroughReduction:1},
	consider again a representation $\Phi$ of the quotient $^*$\=/algebra $\A^{\lie g} / 
	\genSId{\momentmap-\mu}{}$ on a pre-Hilbert space $\Dom$, i.e.\ a positive unital 
	$^*$\=/homomorphism $\Phi\colon \A^{\lie g} / \genSId{\momentmap-\mu}{} \to \Adbar(\Dom)$.
	Then the pullback of any vector state on $\Adbar(\Dom)$ with $\Phi$ is a common eigenstate of 
	$\momentmap$ with eigenvalues $\mu$ as a consequence of Proposition~\ref{proposition:genIdeal}.
	Therefore, disregarding Poisson brackets,
	$\Phi$ factors through the quotient $^*$\=/algebra $\A^{\lie g} / \mathcal V_\mu$,
	which coincides with $\A_{\mu\mred}$ for regular $\mu$.
\end{example}

\subsection{Reduction of States} \label{sec:reducedStates}
We are finally in the position to discuss the reduction of states:
\begin{definition} \label{definition:reducedStates}
  Let $\A$ be a representable Poisson $^*$\=/algebra, $\momentmap \colon \lie{g} \to \A_\Hermitian$ 
  a momentum map for a commutative real 
  Lie algebra $\lie{g}$, and $\mu \in {\lie g}^*$. Moreover, let $(\A_{\mu\mred}, [\argument]_\mu)$ 
  be the $\momentmap$-reduction of $\A$ at $\mu$.
  We say that a state $\omega$ on $\A$ is \emph{$\momentmap$-reducible at $\mu$}
  if there exists a (necessarily unique) state $\omega_{\mu\mred}$ on $\A_{\mu\mred}$ fulfilling 
  $\dupr{\omega}{a} = \dupr{\omega_{\mu\mred}}{[a]_\mu}$ for all $a\in \A^{\lie g}$. In this case
  $\omega_{\mu\mred}$ will be called the \emph{$\momentmap$-reduction of $\omega$ at $\mu$}.
\end{definition}
Note that this definition is independent of the realization of the reduction, 
i.e.~every construction
that fulfils the universal property from Definition~\ref{definition:reduction} leads to the same notion
of reducibility of states, because all realizations of the reduced algebra are isomorphic.

Theorem~\ref{theorem:reduction} already shows how states on the reduced representable Poisson 
$^*$\=/algebra $\A_{\mu\mred}$
are related to those on $\A^{\lie g}$. Moreover, there are many cases in which, for geometric reasons, all states on $\A^{\lie g}$
can be obtained by restricting states on $\A$ to $\A^{\lie g}$:

\begin{definition} \label{definition:averaging}
  Let $\A$ be a representable Poisson $^*$\=/algebra and 
  $\momentmap \colon \lie{g} \to \A_\Hermitian$ 
  a momentum map for a commutative real Lie algebra $\lie{g}$.
  An \emph{averaging operator} for the induced action of $\lie g$ on $\A$ is a linear map 
  $\argument_{\av} \colon \A \to \A^{\lie g}$ which is a projection onto $\A^{\lie g}$, Hermitian, 
  and positive,
  i.e.~$a_\av = a$ for all $a\in \A^{\lie g}$, $a_\av \in \A^{\lie g}_\Hermitian$ for all $a\in \A_\Hermitian$, and $a_\av \in (A^{\lie g})^+_\Hermitian$
  for all $a\in \A^+_\Hermitian$.
\end{definition}
Note that necessarily $\Unit \in \A^{\lie g}$, so $\Unit_\av = \Unit$.
If the action of $\lie g$ is obtained from differentiating the action of a connected Lie group ${G}$
that is compatible with $\argument^*$ and order-preserving, then an averaging operator can oftentimes be constructed by averaging over
the action of ${G}$. Two examples of this will be discussed in Sections~\ref{sec:PoiMan} and \ref{sec:polynomials}.

Averaging operators, if they exist, are a quite useful tool because they clearly allow the extension of states on a $^*$\=/subalgebra
of invariant elements to the whole space:

\begin{proposition} \label{proposition:avapplication}
  Let $\A$ be a representable Poisson $^*$\=/algebra and $\momentmap \colon \lie{g} \to \A_\Hermitian$ a momentum map for a commutative real Lie algebra $\lie{g}$.
  Denote by $\argument\at{\A^{\lie g}} \colon \A^* \to (\A^{\lie g})^*$ the restriction of linear functionals on $\A$ to $\A^{\lie g}$.
  Then $\omega \at{\A^{\lie g}} \in \States_{\momentmap\!,\mu}(\A^{\lie g})$ for all $\omega \in \States_{\momentmap\!,\mu}(\A)$,
  and conversely, if there exists an averaging operator $\argument_{\av} \colon \A \to \A^{\lie g}$, then
  for every $\tilde{\omega} \in \States_{\momentmap\!,\mu}(\A^{\lie g})$ there exists $\omega \in 
  \States_{\momentmap\!,\mu}(\A)$
  with $\omega\at{\A^{\lie g}} = \tilde{\omega}$, e.g.~$\omega \coloneqq \tilde{\omega} \circ 
  \argument_\av$.
\end{proposition}
\begin{proof}
  Clear.
\end{proof}
So if there exists an averaging operator $\argument_\av \colon \A \to \A^{\lie g}$, then 
in the definitions of the quadratic module $\mathcal{R}_\mu$ in \eqref{eq:regular:RQm}
and of the non-commutative vanishing ideal $\mathcal{V}_\mu$ in \eqref{eq:regular:VIdeal},
the condition ``$\omega \in \States_{\momentmap\!,\mu}(\A^{\lie g})$'' can be replaced by ``$\omega \in \States_{\momentmap\!,\mu}(\A)$''.
For the reduction of states, this yields:

\begin{theorem} \label{theorem:reducedStates}
  Let $\A$ be a representable Poisson $^*$\=/algebra, $\momentmap \colon \lie{g} \to \A_\Hermitian$ 
  a momentum map for a commutative real 
  Lie algebra $\lie{g}$, and $\mu \in {\lie g}^*$.
  \begin{enumerate}
    \item If a state $\omega \in \States(\A)$ is $\momentmap$-reducible at $\mu$, then $\omega \in \States_{\momentmap\!,\mu}(\A)$.
    \item If there exists an averaging operator $\argument_\av \colon \A \to \A^{\lie g}$,
      then for every $\rho \in \States(\A_{\mu\mred})$ there is a state $\omega \in \States(\A)$ that is
      $\momentmap$-reducible at $\mu$ with $\omega_{\mu\mred} = \rho$ (thus especially $\omega \in \States_{\momentmap,\mu}(\A)$).
      \label{item:reducedStates:lift}
    \item If Poisson-commuting elements of $\A$ commute and if $\mu$ additionally is regular for 
    $\momentmap$,
      then every $\omega \in \States_{\momentmap\!,\mu}(\A)$ is $\momentmap$-reducible at $\mu$.
      \label{item:reducedStates:reducibility}
  \end{enumerate}
\end{theorem}
\begin{proof}
  Let $(\A_{\mu\mred}, [\argument]_\mu)$ be the $\momentmap$-reduction of $\A$ at $\mu$ constructed 
  in Theorem~\ref{theorem:reduction}.
  
  If for some $\omega \in \States(\A)$ there exists $\omega_{\mu\mred} \in \States(\A_{\mu\mred})$ such that $\dupr{\omega}{a} = \dupr{\omega_{\mu\mred}}{[a]_\mu}$
  holds for all $a\in\A^{\lie g}$, then $\dupr{\omega}{(\momentmap(\xi)-\dupr{\mu}{\xi}\Unit)^2} = 0$ for all $\xi \in \lie g$ because 
  $\ker[\argument]_\mu$ is a $^*$\=/ideal of $\A^{\lie g}$ that contains $\momentmap(\xi)-\dupr{\mu}{\xi}\Unit$ for all $\xi \in \lie g$
  by definition of $[\argument]_\mu$. This proves the first point.
  
  Now let any state $\rho$ on $\States(\A_{\mu\mred})$ be given. If there exists an averaging operator $\argument_\av \colon \A \to \A^{\lie g}$,
  then $\omega \coloneqq \rho \circ [\argument]_\mu \circ \argument_\av \colon \A \to \CC$ is a state on $\A$ that fulfils
  $\dupr{\omega}{a} = \dupr{\rho}{[a_\av]_\mu} = \dupr{\rho}{[a]_\mu}$ for all $a\in\A^{\lie g}$,
  i.e.~$\omega$ is $\momentmap$-reducible at $\mu$ with $\omega_{\mu\mred} = \rho$ and especially $\omega \in \States_{\momentmap\!,\mu}(\A)$
  by the first part. This proves the second point.
  
  For the third point, given $\omega \in \States_{\momentmap\!,\mu}(\A)$, then
  $\omega \at{\A^{\lie g}} \in \States_{\momentmap\!,\mu}(\A^{\lie g})$ and by 
  Theorem~\ref{theorem:reduction} and Corollary~\ref{corollary:vanishing}
  there exists a state $\omega_{\mu\mred}$ on $\A_{\mu\mred}$ fulfilling $\omega_{\mu\mred} \circ [\argument]_\mu = \omega\at{\A^{\lie g}}$.
\end{proof}
For a representable Poisson $^*$\=/algebra $\A$ in which Poisson-commuting elements commute, equipped with a momentum map $\momentmap \colon \lie{g} \to \A_\Hermitian$
that admits an averaging operator, and for regular momenta $\mu \in \lie g^*$, the states on the reduced algebra $\A_{\mu\mred}$
thus are just the reductions of the common eigenstates of $\momentmap$ with eigenvalues given by $\mu$. This matches
the heuristic about the reduction of states discussed in the introduction.

Note also how part~\refitem{item:reducedStates:lift} above makes use of the notion of ordered $^*$\=/algebras: This statement is only true
because the reduced algebra is endowed with an order obtained from the reduction procedure, but would fail if one considers $^*$\=/algebras
endowed always with the algebraic order. We will see examples of this in Section~\ref{sec:polynomials}.

Since many properties of the reduction that were discussed in this section depend on the correct 
choice of the order on the reduced algebra,
it is always desirable to find a clear description of this order, or at least of all its states.
A description as reductions of common eigenstates is already given by the above 
Theorem~\ref{theorem:reducedStates},
but one might still strive for a description independent of the reduction.
This can be seen as a problem of (non-commutative) real algebraic geometry and will be discussed 
further in the following examples.

\section{Reduction of Poisson Manifolds} \label{sec:PoiMan}
The first example to discuss is the reduction of Poisson manifolds, i.e.~of the representable 
Poisson $^*$\=/algebra
$\Smooth(M)$ of smooth $\CC$-valued functions with the pointwise order on a smooth manifold $M$, endowed with a Poisson bracket
which, in this case, can always be obtained from a Poisson tensor.

Note that the pointwise order on $\Smooth(M)$ is in general not the algebraic order:
For example, if 
$f \in \Smooth(M)_\Hermitian$ is (in some local coordinates) a homogeneous polynomial function which cannot be expressed as a sum of squares of polynomial functions,
like the homogeneous Motzkin polynomial, then considering Taylor expansions at $0$ shows that $f$ cannot even be expressed as a sum of squares of smooth functions,
see \cite{bony.broglia.colombini.pernazza:NonnegativeFunctionsAsSumsOfSquares} for details.
Nevertheless, all algebraically positive Hermitian linear functionals on $\Smooth(M)$ are also positive with respect to the pointwise order
because $\sqrt{f+\epsilon \Unit}$ is smooth for every smooth function $f \colon M \to {[0,\infty[}$ and all $\epsilon \in {]0,\infty[}$.

For the rest of this section, we will assume that $M$ is a Poisson manifold
so that $\Smooth(M)$ with the pointwise order
is a representable Poisson $^*$\=/algebra. Moreover, we assume the following:
\begin{itemize}
  \item The Poisson manifold $M$ is endowed with a smooth left action $\argument \acts \argument \colon G \times M \to M$, $(g,x) \mapsto g \acts x$,
    of an abelian connected Lie group $G$, which induces a right action $\argument \racts \argument \colon \Smooth(M) \times G \to \Smooth(M)$
    by pullbacks, i.e.~$(f\racts g)(x) \coloneqq f(g \acts x)$ for all $f\in \Smooth(M)$, $g\in G$, and all $x\in M$.
  \item The right action $\argument \racts \argument \colon \Smooth(M) \times \lie g \to \Smooth(M)$ of the (finite dimensional)
    Lie algebra $\lie g$ of $G$ on $\Smooth(M)$, which one obtains by differentiating the right action of $G$,
    is induced by a momentum map $\momentmap \colon \lie g \to \Smooth(M)_\Hermitian$ as in Definition~\ref{definition:momentmap};
    especially $f\racts \xi = \poi{f}{\momentmap(\xi)}$ for all $f\in\Smooth(M)$, $\xi \in \lie g$.
    Note that we use the same symbol $\racts$ to denote the actions of the Lie group $G$ and of its Lie algebra $\lie g$.
\end{itemize}
We also fix a momentum $\mu \in \lie g^*$ and define the \emph{$\mu$-levelset of $\momentmap$},
\begin{align}
  \Levelset_\mu \coloneqq \set[\big]{x\in M}{\momentmap(\xi)(x) = \dupr{\mu}{\xi} \textup{ for all }\xi\in\lie g}
  .
  \label{eq:levelsetmu}
\end{align}

\begin{proposition} \label{proposition:manifoldVanishing}
  For every $\mu \in \lie g^*$, the quadratic module $\mathcal{R}_\mu$ and the generalized vanishing ideal from Definition~\ref{definition:regular} are
  \begin{align}
    \mathcal{R}_\mu &= \set[\big]{f\in \Smooth(M)^{\lie g}_\Hermitian}{f(x) \ge 0 \textup{ for all } x \in \Levelset_\mu}
    \label{eq:manifoldVanishing:R}
  \shortintertext{and}
    \mathcal{V}_\mu &= \set[\big]{f\in\Smooth(M)^{\lie g}}{f(x) = 0 \textup{ for all } x \in \Levelset_\mu}
    .
    \label{eq:manifoldVanishing:V}
  \end{align}
\end{proposition}
\begin{proof}
  We start with \eqref{eq:manifoldVanishing:R}. Every evaluation functional $\delta_x \colon \Smooth(M)^{\lie g} \to \CC$,
  $f\mapsto\dupr{\delta_x}{f}\coloneqq f(x)$ with $x\in \Levelset_\mu$ is an element of $\States_{\momentmap\!,\mu}\big(\Smooth(M)^{\lie g}\big)$
  by definition of $\Levelset_\mu$.
  Therefore $f(x)\ge 0$ for all $f\in \mathcal{R}_\mu$, $x\in \Levelset_\mu$, which proves the inclusion ``$\subseteq$''.
  
  Conversely, consider any $f\in \Smooth(M)^{\lie g}_\Hermitian$ fulfilling $f(x) \ge 0$ for all $x\in \Levelset_\mu$.
  Let $\alpha \colon \RR \to {[{-2},\infty[}$ be a smooth function that fulfils $\alpha(y) = y$ for all $y\in {[{-1},\infty[}$ and
  $\alpha(y) = -2$ for all $y\in {]{-\infty}, {-2}]}$, and define the smooth function $\beta \colon \RR \to \RR$, $y \mapsto \beta(y) \coloneqq y- \alpha(y)$.
  Note that $\beta(y) = 0$ for all $y\in {[{-1},\infty[}$.
  Then for all $k \in \NN$ the identity $kf = (\alpha \circ kf) + (\beta \circ kf)$ holds, and $\alpha \circ kf \ge - 2\cdot\Unit$.
  
  Moreover, $\beta \circ kf \in \genSId{\momentmap-\mu}$:
  Indeed, let $\eta_1,\dots,\eta_\ell \in \lie g$ with $\ell\in \NN_0$ be a basis of $\lie g$ (which is finite-dimensional by assumption)
  and define $h \coloneqq \sum_{j=1}^\ell (\momentmap(\eta_j) - \dupr{\mu}{\eta_j} \Unit)^2 \in 
  \genSId{\momentmap-\mu}{}$, then any $x\in M$ fulfils $x\in \Levelset_\mu$ if and only if $h(x) = 0$. Consequently, 
  $\beta \circ kf = h g_k \in \genSId{\momentmap-\mu}$ with $g_k \in \Smooth(M)^{\lie g}_\Hermitian$
  defined by requiring that $g_k (x) = (\beta \circ kf)(x) / h(x)$ for all $x\in M$ with $f(x)<0$ and 
  $g_k(x)=0$ for all $x\in M$ with $f(x) > - 1/k$. Note that $g_k$ is indeed $\lie g$-invariant because $f$ and $h$ are $\lie g$-invariant.
  
  For any $\omega \in \States_{\momentmap\!,\mu}\big( \Smooth(M)^{\lie g} \big)$ we now have $\dupr{\omega}{\beta \circ kf} = 0$
  by Proposition~\ref{proposition:genIdeal}, so
  $\dupr{\omega}{f} = k^{-1} \dupr{\omega}{\alpha \circ kf} + k^{-1}\dupr{\omega}{\beta \circ kf} \ge -2k^{-1}$ for all $k\in \NN$,
  hence $\dupr{\omega}{f} \ge 0$, and therefore $f\in \mathcal{R}_\mu$. This shows that \eqref{eq:manifoldVanishing:R} holds,
  which also implies \eqref{eq:manifoldVanishing:V} because
  $\mathcal{V}_\mu = \big(\mathcal{R}_\mu \cap (-\mathcal{R}_\mu)\big) + \I \big(\mathcal{R}_\mu \cap (-\mathcal{R}_\mu)\big)$.
\end{proof}

\begin{corollary} \label{corollary:manifoldVanishing}
  Every $\mu \in \lie g^*$ is regular for $\momentmap$ in the sense of Definition~\ref{definition:regular}.
\end{corollary}
\begin{proof}
  We have to check that the $^*$\=/ideal $\mathcal{V}_\mu$ of $\Smooth(M)^{\lie g}$ is a Poisson ideal,
  i.e.~$\poi{f}{g} \in \mathcal{V}_\mu$ for all $f\in\Smooth(M)^{\lie g}$, $g\in \mathcal{V}_\mu$.
  This has been proven in \cite[Thm.~1]{arms.cushman.gotay:universalReductionProcedure}, for convenience of the reader,
  we repeat the relevant part here:

  Let $\gamma \colon {]{-\epsilon},\epsilon[} \to M$ be an integrating curve of the Hamiltonian flow of $f$ through $x$,
  i.e.~$\gamma(0) = x$ and $\frac{\D}{\D t} \at[\big]{\tau} h(\gamma(t)) = \poi{f}{h}(\gamma(\tau))$ for all $\tau \in {]{-\epsilon},\epsilon[}$
  and all $h\in \Smooth(M)$.
  For all $\xi \in \lie g$ it follows that
  \begin{equation*}
    \frac{\D}{\D t} \at[\bigg]{\tau} \momentmap(\xi)\big(\gamma(t)\big) =  \poi[\big]{f}{\momentmap(\xi)}\big(\gamma(\tau)\big) = \big(f \racts \xi\big)\big(\gamma(\tau)\big) = 0
  \end{equation*}
  by $\lie g$-invariance of $f$,
  so $\momentmap(\xi)(\gamma(t)) = \momentmap(\xi)(\gamma(0)) = \momentmap(\xi)(x) = \dupr{\mu}{\xi}$ for all $t \in {]{-\epsilon},\epsilon[}$.
  This shows that $\gamma(t) \in \Levelset_\mu$ for all $t \in {]{-\epsilon},\epsilon[}$ and therefore
  $\poi{f}{g}(x) = \frac{\D}{\D t} \at[\big]{0} g(\gamma(t)) = 0$ because $g$ is constantly $0$ on $\Levelset_\mu$ by 
  Proposition~\ref{proposition:manifoldVanishing}.
\end{proof}
We thus recover the universal reduction procedure of \cite{arms.cushman.gotay:universalReductionProcedure}, adapted to Poisson manifolds:

\begin{theorem} \label{theorem:universalPoisson}
  Retain the assumptions from the beginning of this section, fix $\mu \in \lie g^*$, and define the quotient topological space 
  $M_{\mu\mred} \coloneqq \Levelset_\mu / G$ and the $^*$\=/algebra $\Stetig(M_{\mu\mred})$ of continuous
  $\CC$-valued functions on $M_{\mu\mred}$ with the pointwise operations.
  For every $f\in \Smooth(M)^{\lie g}$ define $[f]_\mu \in \Stetig(M_{\mu\mred})$ as
  \begin{align}
    [f]_\mu\big([x]_G\big) \coloneqq f(x)
    \label{eq:universalPoisson:bracketmu}
  \end{align}
  for all $[x]_G \in M_{\mu\mred}$ with representative $x\in\Levelset_\mu$.
  Then the $\momentmap$-reduction of $\Smooth(M)$ at $\mu$ is given by $\big( \mathcal{W}^\infty(M_{\mu\mred}), [\argument]_\mu \big)$,
  where $\mathcal{W}^\infty(M_{\mu\mred})$ is the unital $^*$\=/subalgebra
  \begin{equation}
    \mathcal{W}^\infty(M_{\mu\mred}) \coloneqq \set[\big]{[f]_\mu}{f\in\Smooth(M)^{\lie g}}
  \end{equation}
  of $\Stetig(M_{\mu\mred})$ with the pointwise order, equipped with the Poisson bracket that is given by
  \begin{align}
    \poi[\big]{[f]_\mu}{[g]_\mu} \coloneqq \big[\poi{f}{g}\big]_\mu
    \label{eq:universalPoisson:poisson}
  \end{align}
  for all $f,g \in \Smooth(M)^{\lie g}$, and where $[\argument]_\mu \colon \Smooth(M)^{\lie g} \to \mathcal{W}^\infty(M_{\mu\mred})$
  is the map from \eqref{eq:universalPoisson:bracketmu}.
\end{theorem}
\begin{proof}
  Poisson-commuting elements of $\Smooth(M)$ commute trivially, and $\mu$ is regular for $\momentmap$ by the previous Corollary~\ref{corollary:manifoldVanishing}.
  Therefore we only need to check that the conditions of Theorem~\ref{theorem:vanishing:neu} are fulfilled:

  As the Lie group $G$ is connected by assumption, every $f\in \Smooth(M)^{\lie g}$ is $G$-invariant,
  so the function $[f]_\mu \colon M_{\mu\mred} \to \CC$ of \eqref{eq:universalPoisson:bracketmu} is well-defined,
  and $[f]_\mu$ is continuous by definition of the quotient topology on $M_{\mu\mred}$. It is now easy to check that
  $\mathcal{W}^\infty(M_{\mu\mred})$ is a unital $^*$\=/subalgebra of $\Stetig(M_{\mu\mred})$.
  The Poisson bracket \eqref{eq:universalPoisson:poisson} is well-defined as a consequence
  of the previous Corollary~\ref{corollary:manifoldVanishing}, and the pointwise order on $\mathcal{W}^\infty(M_{\mu\mred})$
  is induced by its states by definition. So $\mathcal{W}^\infty(M_{\mu\mred})$ is a representable Poisson $^*$\=/algebra.
  
  The map $[\argument]_\mu \colon \Smooth(M)^{\lie g} \to \mathcal{W}^\infty(M_{\mu\mred})$ clearly is a positive unital
  $^*$\=/homomorphism with kernel $\mathcal{V}_\mu$. It is surjective and compatible with Poisson brackets by definition
  of $\mathcal{W}^\infty(M_{\mu\mred})$.
  
  Now consider an element $f \in \Smooth(M)^{\lie g}_\Hermitian$ such that $[f]_\mu \in \mathcal{W}^\infty(M_{\mu\mred})^+_\Hermitian$,
  i.e.~$[f]_\mu$ is pointwise positive. This means $f(x) = [f]_\mu([x]_G) \ge 0$ for all $x\in \Levelset_\mu$, so $f\in\mathcal{R}_\mu$
  by Proposition~\ref{proposition:manifoldVanishing}.
\end{proof}
The algebra $\mathcal{W}^\infty(M_{\mu\mred})$ consists of the Whitney smooth functions on $M_{\mu\mred}$
as in \cite{arms.cushman.gotay:universalReductionProcedure}.
Under some additional standard assumptions (proper and free action, $\mu$ a regular value), $M_{\mu\mred}$ can even be given the structure of a Poisson manifold and
$\mathcal{W}^\infty(M_{\mu\mred}) = \Smooth(M_{\mu\mred})$. Especially in the symplectic case it was shown in
\cite{arms.cushman.gotay:universalReductionProcedure} that this construction then gives back 
Marsden--Weinstein reduction.

We now turn our attention to the reduction of states on $\Smooth(M)$: By Theorem~\ref{theorem:reducedStates} and Corollary~\ref{corollary:manifoldVanishing},
a state $\omega$ on $\Smooth(M)$ is $\momentmap$-reducible at $\mu$ if and only if $\omega \in \States_{\momentmap\!,\mu}\big( \Smooth(M)\big)$.
Especially for the evaluation functionals $\delta_x \colon \Smooth(M)^{\lie g} \to \CC$, $f\mapsto\dupr{\delta_x}{f}\coloneqq f(x)$ with $x\in M$
this means that $\delta_x$ is $\momentmap$-reducible at $\mu$ if and only if $x \in \Levelset_\mu$.
It is also easy to see that in this case, the corresponding reduced state is just $(\delta_x)_{\mu\mred} = 
\delta_{[x]_G} \in \States\big( \mathcal{W}^\infty(M_{\mu\mred})\big)$, the evaluation functional at $[x]_G \in M_{\mu\mred}$.
Note that the evaluation functionals at points of $M$ are precisely the multiplicative states on $\Smooth(M)$ (this result
is sometimes referred to as ``Milnor and Stasheff's exercise'', see e.g.~\cite{kriegl.michor.schachermayer:charactersOnAlgebrasOfSmoothFunctions}
for a proof in a setting much more general than just smooth manifolds) and the multiplicative states in turn are
precisely the extreme points of $\States\big(\Smooth(M)\big)$, see e.g.~\cite[Cor.~2.61, Thm.~2.63]{schmuedgen:invitationToStarAlgebras} or 
\cite{schoetz:preprintGelfandNaimarkTheorems}.
The reduced space $M_{\mu\mred}$ therefore is just a geometric manifestation of the 
reduction of the reducible extremal states.

Theorem~\ref{theorem:reducedStates} also shows that every state on $\mathcal{W}^\infty(M_{\mu\mred})$ can be obtained as a reduction
of some reducible state on $\Smooth(M)$ if there exists an averaging operator for the action of $\lie g$. We close this section
with the observation that this is indeed the case if the action of $G$ on $M$ is proper:

\begin{proposition} \label{proposition:properAveraging}
  If the action of $G$ on $M$ is proper, then one can construct an averaging operator $\argument_\av \colon \Smooth(M) \to \Smooth(M)^{\lie g}$.
\end{proposition}
\begin{proof}
  The averaging operator can be constructed essentially like in the proof of
  \cite[Prop.~5.3]{miaskiwskyi:invariantHochschildCohomologyOfSmoothFunctions}, see also \cite[Chap.~4]{pflaum}: 
  Let $\nu$ be a non-zero right-invariant volume form on $G$
  and let $\Smooth_0(M)$ be the $^*$\=/ideal of compactly supported functions in $\Smooth(M)$. Then 
  the 
  compactly supported averaging operator
  $\argument_{\mathrm{cp\textup{-}av}} \colon \Smooth_0(M) \to \Smooth(M)^{\lie g}$, $f\mapsto f_{\mathrm{cp\textup{-}av}}$ with
  \begin{align*}
    f_{\mathrm{cp\textup{-}av}}(x) \coloneqq \integral{g\in G}{f(g \acts x)}{ \nu(g) }
  \end{align*}
  for all $x\in M$ is well-defined because the action of $G$ on $M$ is proper, so that for every 
  compact neighbourhood $K$ of
  any $x\in M$ the set $\set[\big]{g \in G}{ \exists_{y,z\in K}: g \acts y = z }$ is compact,
  and because for all $f\in \Smooth_0(M)$, all $\tilde{g} \in G$ and all $x\in M$ one finds that
  \begin{align*}
    \big(f_{\mathrm{cp\textup{-}av}} \racts \tilde{g}\big) (x)
    =
    \integral{g\in G}{f(g \tilde{g} \acts x)}{ \nu(g) }
    =
    \integral{g'\in G}{f(g' \acts x)}{ \nu(g' \tilde{g}^{-1}) }
    =
    \integral{g'\in G}{f(g' \acts x)}{ \nu(g') }
    =
    f_{\mathrm{cp\textup{-}av}}(x)
  \end{align*}
  by right-invariance of $\nu$,
  so that indeed $f_{\mathrm{cp\textup{-}av}} \in \Smooth(M)^{\lie g}$ for all $f\in \Smooth_0(M)$.
  It is clear that $f_{\mathrm{cp\textup{-}av}} \in (\Smooth(M)^{\lie g})_\Hermitian$
  for all $f \in \Smooth_0(M)_\Hermitian$ and $f_{\mathrm{cp\textup{-}av}} \in (\Smooth(M)^{\lie g})_\Hermitian^+$
  for all $f \in \Smooth_0(M)_\Hermitian^+$ hold. We also make the following observation: Given $f\in \Smooth(M)$, $h \in \Smooth_0(M)$,
  and $x\in M$ such that $f(g \acts x) = f(x)$ for all $g\in G$, then $f h \in \Smooth_0(M)$ and 
  $(fh)_{\mathrm{cp\textup{-}av}}(x) = f(x) h_{\mathrm{cp\textup{-}av}}(x)$.
  
  Now let $(\theta_\ell)_{\ell\in \NN}$ be a locally finite and smooth partition of unity on $M$ such that each 
  $\theta_\ell$ with $\ell\in \NN$ has compact support. Define $\sigma_\ell \coloneqq (\theta_\ell)_{\mathrm{cp\textup{-}av}}$
  for all $\ell\in \NN$, and the preimages $U_\ell \coloneqq \sigma_\ell^{-1}(]0,\infty[)$. Then $(U_\ell)_{\ell\in \NN}$
  is an open cover of $M$ by $G$-invariant sets. By \cite[Prop.~2.3.8, v)]{ortega.ratiu:momentumMapsAndHamiltonianReduction} there exists
  a locally finite and smooth partition of unity $(\tau_\ell)_{\ell\in \NN}$ of $M$ such that each $\tau_\ell$ with $\ell\in \NN$
  has support contained in $U_\ell$ and is $G$-invariant, so $\tau_\ell \in \Smooth(M)^{\lie g}$
  (essentially, one can construct a partition of unity on the quotient topological space $M/G$ out of those real-valued and
  pointwise positive continuous functions that can be pulled back to smooth functions on $M$).
  Define $\argument_\av \colon \Smooth(M) \to \Smooth(M)^{\lie g}$,
  \begin{align*}
    f \mapsto f_\av \coloneqq \frac{ \sum_{\ell\in \NN} (f \theta_\ell)_{\mathrm{cp\textup{-}av}} \tau_\ell }{ \sum_{\ell\in \NN} \sigma_\ell \tau_\ell }
    ,
  \end{align*}
  where the infinite sums in the numerator and denominator are locally finite because $\tau_\ell$ is a locally finite partition of unity,
  and where the denominator is strictly positive at every point $x\in M$, because for every $x\in M$ there is $\ell \in \NN$
  with $\tau_\ell(x) > 0$, which means $x \in U_\ell$, i.e.~$\sigma_\ell(x) > 0$. It is clear that indeed $f_\av \in \Smooth(M)^{\lie g}$
  for all $f\in \Smooth(M)$, so $\argument_\av$ is well-defined, and that $f_\av \in (\Smooth(M)^{\lie g})_\Hermitian$
  for all $f \in \Smooth(M)_\Hermitian$ and $f_\av\in (\Smooth(M)^{\lie g})_\Hermitian^+$
  for all $f \in \Smooth_0(M)_\Hermitian^+$ hold. Finally, given $f\in \Smooth(M)$ and $x\in M$ such that 
  $f(x) = f(g \acts x)$ holds for all $g\in G$, then
  \begin{align*}
    f_\av(x)
    =
    \frac{ \sum_{\ell\in \NN} (f \theta_\ell)_{\mathrm{cp\textup{-}av}}(x) \tau_\ell(x) }{ \sum_{\ell\in \NN} \sigma_\ell(x) \tau_\ell(x) }
    =
    f(x) \frac{ \sum_{\ell\in \NN} (\theta_\ell)_{\mathrm{cp\textup{-}av}}(x) \tau_\ell(x) }{ \sum_{\ell\in \NN} \sigma_\ell(x) \tau_\ell(x) }
    =
    f(x)
  \end{align*}
  because $\sigma_\ell = (\theta_\ell)_{\mathrm{cp\textup{-}av}}$ by definition. This especially also shows that $f_\av = f$ for all
  $f\in \Smooth(M)^{\lie g}$ because all $\lie g$-invariant elements of $\Smooth(M)$ are also $G$-invariant, and so $\argument_\av$
  is indeed an averaging operator in the sense of Definition~\ref{definition:averaging}.
\end{proof}

\section{Reduction of the Weyl Algebra} \label{sec:weyl}
In this section, let $\Schwartz(\RR^m)$ with $m\in \NN$ be the Schwartz space of rapidly decreasing functions on $\RR^m$
with the usual $\mathrm{L}^2$ inner product over the Lebesgue measure on $\RR^m$. The \emph{Weyl 
algebra}
$\Weyl(\RR^m)$ is the unital $^*$\=/subalgebra of $\Adbar\big( \Schwartz(\RR^m) \big)$ that is generated by the usual position
and momentum operators $q_j, p_j \in \Adbar\big( \Schwartz(\RR^m) \big)$, which are defined as
$q_j f \coloneqq x_j f$ and $p_j f \coloneqq -\I \frac{\partial f}{\partial x_j}$ for all $f\in \Schwartz(\RR^m)$,
where $x_j$ denotes the $j$-th standard coordinate function on $\RR^m$. 
We endow $\Weyl(\RR^m)$ with the usual operator order like in 
Section~\ref{subsec:quadraticModules}.
As discussed in Example~\ref{example:comaspoi},
the Weyl algebra $\Weyl(\RR^m)$ together with the Poisson bracket that is defined as $\poi{a}{b} \coloneqq -\I (ab-ba)$
for all $a,b\in \Weyl(\RR^m)$ is a representable Poisson $^*$\=/algebra. 
A basis of $\Weyl(\RR^m)$ is given by $\set{p^kq^\ell}{k,\ell \in \NN_0^m}$
where
\begin{equation}
  p^k q^\ell \coloneqq (p_1)^{k_1} \dots (p_m)^{k_m} (q_1)^{\ell_1} \dots (q_m)^{\ell_m}
\end{equation}
for all $k,\ell \in \NN_0^m$.

As noted in \cite{woronowicz:quantumProblemOfMoments}, the order on $\Weyl(\RR^m)$
is not the algebraic order, even if $m=1$: Consider the number operator $N \coloneqq \frac{1}{2}(q+\I p)^*(q+\I p) \in \Weyl(\RR)^+_\Hermitian$,
then it follows from $N$ being essentially self-adjoint with spectrum $\NN_0$ that
$(N-\Unit)(N-2\Unit) \in \Weyl(\RR)^+_\Hermitian$, but one can check that $(N-\Unit)(N-2\Unit) \notin \Weyl(\RR)^{++}_\Hermitian$
because $(N-\Unit)(N-2\Unit) \stackrel{!}{=} \sum_{k=1}^K b_k^* b_k$ with $K\in \NN$ and $b_1,\dots,b_K \in \Weyl(\RR)$
would require all $b_k$ with $k\in \{1,\dots,K\}$ to be of degree at most $2$ in the generators $a \coloneqq q+\I p$ and $a^*$ and to have the 
$1$- and $2$-eigenspaces of $N$ in their kernel. From \cite[Thms.~10.36, 10.37]{schmuedgen:invitationToStarAlgebras}
it then follows that there even exists an algebraically positive state on $\Weyl(\RR)$ which is not positive.
An algebraic characterization of the positive Hermitian linear functionals on $\Weyl(\RR^m)$ can be obtained from
the Positivstellensatz for the Weyl algebra from \cite{schmuedgen:StrictPositivstellensatzForWeylAlgebra}:
A Hermitian linear functional $\omega$ on $\Weyl(\RR^m)$ is positive if and only if
$\dupr{\omega}{a} \ge 0$ holds for all those $a\in \Weyl(\RR^m)_\Hermitian$ for which there exists
$b\in \mathcal{N}$ such that $bab \in \Weyl(\RR^m)^{++}_\Hermitian$, where 
$\mathcal{N}$ is the set of all finite products of elements $N-\lambda \Unit$ with $\lambda \in \RR\setminus \NN_0$
and with $N \coloneqq \frac{1}{2}\sum_{j=1}^m (q_j+\I p_j)^*(q_j+\I p_j) \in \Weyl(\RR^m)^+_\Hermitian$
the $m$-dimensional number operator.

Throughout the rest of this section, we make the following assumptions:
\begin{itemize}
  \item $n\in \NN$ is a fixed dimension, the coordinate functions on $\RR^{1+n}$ will be numbered $x_0,\dots,x_n$, and those on $\RR^n$ will be numbered $x_1,\dots,x_n$.
  \item $\lie g \cong \RR$ is the $1$-dimensional Lie algebra and we choose any momentum $\mu \in \lie g^*$, 
    which we will identify by abuse of notation with $\mu \coloneqq \mu(1) \in \RR$.
  \item The momentum map is $\momentmap \colon \lie g \to \Weyl(\RR^{1+n})_\Hermitian$, $\lambda \mapsto \momentmap(\lambda) \coloneqq \lambda p_0$,
    and therefore is completely described by $p_0 = \momentmap(1)$.
\end{itemize}
Note that the corresponding action of the Lie algebra $\lie g$ 
can be obtained by differentiating the action of the Lie group $\RR$ by translation of the $0$-component.
The space of invariant elements under this action can of course be described more explicitly:

\begin{lemma} \label{lemma:weylalgdec}
  $\Weyl(\RR^{1+n})^{\lie g}$ is the unital $^*$\=/subalgebra of $\Weyl(\RR^{1+n})$ that is
  generated by $q_1,\dots,q_n$ and $p_0,p_1,\dots,p_n$, which is the linear subspace spanned by the
  basis elements $p^k q^\ell$ with $k,\ell \in \NN_0^{1+n}$, $\ell_0=0$. Moreover, we have the decomposition
  $\Weyl(\RR^{1+n})^{\lie g} = \genSId{p_0-\mu} \oplus \genSAlg{q_1,\dots,q_n,p_1,\dots,p_n}$,
  where $\genSId{p_0-\mu}{} = \genSId{\momentmap-\mu}{}$ is the $^*$\=/ideal of $\Weyl(\RR^{1+n})^{\lie g}$ generated by $p_0-\mu\Unit$,
  and where $\genSAlg{q_1,\dots,q_n,p_1,\dots,p_n}$ is the unital $^*$\=/subalgebra of $\Weyl(\RR^{1+n})^{\lie g}$ generated
  by $\{q_1,\dots,q_n,p_1,\dots,p_n\}$.
\end{lemma}
\begin{proof}
  It is clear that $\Weyl(\RR^{1+n})^{\lie g}$ is a unital $^*$\=/subalgebra of $\Adbar\big(\Schwartz(\RR^{1+n})\big)$ and that
  $q_j \in \Weyl(\RR^{1+n})^{\lie g}$ for all $j\in \{1,\dots,n\}$ and $p_j \in \Weyl(\RR^{1+n})^{\lie g}$ for all $j\in \{0,\dots,n\}$.
  Conversely, from $\poi{q_0}{p_0} = 1$ it follows for all $k,\ell \in \NN_0^{1+n}$
  that $\poi{p^k q^\ell}{p_0}\, q_0  = \ell_0 p^k q^\ell$. So given any 
  $a = \sum_{k,\ell\in \NN_0^{1+n}} \alpha_{k,\ell} \,p^k q^\ell \in \Weyl(\RR^{1+n})^{\lie g}$
  with complex coefficients $\alpha_{k,\ell}$, then $\poi{a}{p_0} = 0$ implies $\poi{a}{p_0} \, 
  q_0  = 0$,
  which implies $\alpha_{k,\ell} = 0$ for all $k,\ell \in \NN_0^{1+n}$ with $\ell_0 \neq 0$.
  
  It is clear that $\genSId{p_0-\mu} \cap \genSAlg{q_1,\dots,q_n,p_1,\dots,p_n} = \{0\}$ so that the sum of these linear subspaces
  of $\Weyl(\RR^{1+n})^{\lie g}$ is direct, and from
  $p_0^{k_0} - \mu^{k_0}\Unit = (p_0-\mu\Unit) \sum_{m=0}^{k_0-1} \mu^{k_0-1-m}p_0^{m} \in \genSId{p_0-\mu}{}$
  it follows that $\Weyl(\RR^{1+n})^{\lie g} = \genSId{p_0-\mu} \oplus \genSAlg{q_1,\dots,q_n,p_1,\dots,p_n}$.
\end{proof}
\begin{definition} \label{definition:redmapweyl}
  The map $[\argument]_\mu \colon \Weyl(\RR^{1+n})^{\lie g} \to \Weyl(\RR^n)$ is defined as the unique linear map
  that fulfils $[p^kq^\ell]_\mu = \mu^{k_0} p^{k'} q^{\ell'} \in \Weyl(\RR^n)$ for all $k,\ell \in 
  \NN_0^{1+n}$ with $\ell_0 = 0$,
  where $k' \coloneqq (k_1,\dots,k_n) \in \NN_0^n$ and $\ell' \coloneqq (\ell_1,\dots,\ell_n) \in \NN_0^n$.
\end{definition}
It is easy to check that the kernel of $[\argument]_\mu$ is $\genSId{p_0-\mu}$, and that the restriction of $[\argument]_\mu$ to the complement $\genSAlg{q_1,\dots,q_n,p_1,\dots,p_n}$
is the unital $^*$\=/homomorphism that maps $p^k q^\ell \in \Weyl(\RR^{1+n})^{\lie g}$ with $k,\ell \in \NN_0^{1+n}$ and $k_0 = \ell_0 = 0$ to
$p^{k'}q^{\ell'} \in \Weyl(\RR^n)$. Therefore $[\argument]_\mu \colon \Weyl(\RR^{1+n})^{\lie g} \to \Weyl(\RR^n)$
is a unital $^*$\=/homomorphism, and one might expect that the $p_0$-reduction of $\Weyl(\RR^{1+n})$ at $\mu$ is given by 
$\big(\Weyl(\RR^n), [\argument]_\mu\big)$. 
However, it is not so easy to show that $[\argument]_\mu$ is positive:
As $\ker [\argument]_\mu = \genSId{p_0-\mu}$, Proposition~\ref{proposition:positivitycriterium} and Corollary~\ref{corollary:vanishing}
show that $[\argument]_\mu$ is positive if and only if $\chi_\phi \circ [\argument]_\mu$, for every 
vector state $\chi_\phi$ on $\Weyl(\RR^n)$ with $\phi \in \Schwartz(\RR^n)$, $\seminorm{}{\phi}=1$,
is an eigenstate of $p_0$ with eigenvalue $\mu$. Such eigenstates, however, could only be obtained as 
weak-$^*$-limits of vector states because $p_0$ does not admit any eigenvector.

\begin{definition}
  For $k\in \NN$, let $\iota_k \colon \Schwartz(\RR^n) \to \Schwartz(\RR^{1+n})$, $\phi\mapsto \iota_k(\phi)$ be defined as
  \begin{align}
    \iota_k(\phi)(x_0,x_1,\dots,x_n) \coloneqq k^{-1/2} \E^{\I \mu x_0} \E^{-\pi x_0^2/(2k^2)} \phi(x_1,\dots,x_n)
  \end{align}
  for all $x\in \RR^{1+n}$.
\end{definition}
\begin{lemma} \label{lemma:iotak}
  For every $k\in \NN$, the map $\iota_k \colon \Schwartz(\RR^n) \to \Schwartz(\RR^{1+n})$ is an isometry and adjointable,
  with adjoint $\iota_k^* \colon \Schwartz(\RR^{1+n}) \to \Schwartz(\RR^{n})$, $\psi \mapsto \iota_k^*(\psi)$
  given explicitly as
  \begin{align}
    \iota_k^*(\psi)(x_1,\dots,x_n) = k^{-1/2}\integral{\RR}{\psi(x_0,x_1,\dots,x_n) \E^{-\I \mu x_0} \E^{-\pi x_0^2/(2k^2)}}{ \D x_0}
    \label{eq:iotastar}
    .
  \end{align}
\end{lemma}
\begin{proof}
	For all $\phi \in \Schwartz(\RR^n)$ and all $x_1, \dots, x_n \in \RR$ one has
	\begin{equation*}
	\iota^*_k \iota_k (\phi)(x_1, \dots, x_n)
	=
	\frac{1}{k}\integral{\RR}{\E^{-\pi x_0^2/k^2} \phi(x_1,\dots,x_n) }{\D x_0} 
	= 
	\phi(x_1, \dots, x_n)	\komma
	\end{equation*}
	and it is also easy to check that $\skal{\iota_k^*(\psi)}{\phi} = \skal{\psi}{\iota_k(\phi)}$ 
	for all $\psi\in \Schwartz(\RR^{1+n})$ and all $\phi\in \Schwartz(\RR^n)$
	with $\iota_k^*$ as in \eqref{eq:iotastar}.
\end{proof}
The isometries $\iota_k$ can be used to map vector functionals from $\Adbar\big( 
\Schwartz(\RR^n)\big)$ 
to $\Adbar\big( \Schwartz(\RR^{1+n})\big)$.
We will eventually be interested in their limit for $k\to \infty$.
\begin{lemma} \label{lemma:p0iotak}
  For every $k\in \NN$, the identities $q_j \iota_k = \iota_k q_j$ and $p_j \iota_k = \iota_k p_j$ hold for all $j\in \{1,\dots,n\}$,
  where $q_j$ and $p_j$ denote the position and momentum operators on both $\Schwartz(\RR^n)$ and $\Schwartz(\RR^{1+n})$.
  Moreover, for every $\ell \in \NN$ there exist sequences $(\alpha_{\ell,m;k})_{k\in \NN}$ in $\CC$ for all $m \in \{0,\dots,\ell\}$ such that
  \begin{align}
    (p_0- \mu \Unit)^\ell \iota_k = \sum_{m=0}^\ell \alpha_{\ell,m;k} q_0^m \iota_k
    \label{eq:p0iotak}
  \end{align}
  holds for all $k\in \NN$, and $\lim_{k\to \infty} k^m \alpha_{\ell,m;k} = 0$ for all $\ell \in \NN$ and all $m\in \{0,\dots,\ell\}$.
\end{lemma}
\begin{proof}
  The identities for $q_j$ and $p_j$ with $j\in \{1,\dots,n\}$ are immediately clear. Now consider the case $\ell = 1$ in \eqref{eq:p0iotak}, 
  then for all $\phi\in \Schwartz(\RR^n)$ one has
  \begin{align*}
    \big(p_0 \iota_k \phi \big)(x_0,x_1,\dots,x_n)
    &=
    -\I \frac{\partial}{\partial x_0} k^{-1/2} \E^{\I \mu x_0} \E^{-\pi x_0^2/(2k^2)} \phi(x_1,\dots,x_n)
    \\
    &=
    (\mu + \I \pi x_0 / k^2) k^{-1/2} \E^{\I \mu x_0} \E^{-\pi x_0^2/(2k^2)} \phi(x_1,\dots,x_n)
    \\
    &=
    (\mu + \I \pi x_0 / k^2) \big(\iota_k \phi \big)(x_0,x_1,\dots,x_n)
  \end{align*}
  and therefore $(p_0-\mu\Unit) \iota_k = \I \pi k^{-2} q_0 \iota_k$. So \eqref{eq:p0iotak}
  is fulfilled for $\ell = 1$ with complex coefficients $\alpha_{1,0;k} = 0$ and $\alpha_{1,1;k} = \I \pi k^{-2}$
  and clearly $\lim_{k\to \infty} k^m \alpha_{1,m;k} = 0$ for both $m\in \{0,1\}$.
  The general case of \eqref{eq:p0iotak} is proven inductively:
  
  Assume that for some $L \in \NN$ there exist complex sequences $(\alpha_{\ell,m;k})_{k\in \NN}$ such that \eqref{eq:p0iotak} holds for both $\ell \in \{L,L-1\}$ and all $k\in \NN$.
  This is especially true for $L=1$ because the case $\ell = 1$ has just been discussed, and in the somewhat exceptional case $\ell = 0$
  one has $\alpha_{0,0,k} = 1$ for all $k\in \NN$.
  Using the commutator formula $(p_0-\mu\Unit)^L q_0 - q_0(p_0-\mu\Unit)^L = \I \poi{(p_0-\mu\Unit)^L}{q_0} = -\I L (p_0-\mu\Unit)^{L-1}$
  one then finds:
  \begin{align*}
    (p_0-\mu\Unit)^{L+1} \iota_k
    &=
    (p_0-\mu\Unit)^{L} \I \pi k^{-2} q_0 \iota_k
    \\
    &=
    \I \pi k^{-2} \big( q_0 (p_0-\mu\Unit)^{L} - \I L (p_0-\mu\Unit)^{L-1} \big)\iota_k
    \\
    &=
    \sum_{m=0}^L \I \pi k^{-2} \alpha_{L,m;k} q_0^{m+1} \iota_k + \sum_{m=0}^{L-1} \pi k^{-2} L \alpha_{L-1,m;k} q_0^{m} \iota_k
  \end{align*}
  This shows that \eqref{eq:p0iotak} is again fulfilled for $\ell = L+1$ with suitably chosen complex coefficients $\alpha_{L+1,m;k}$,
  and if the sequences $(k^m \alpha_{\ell,m;k})_{k\in \NN}$ for all $m\in\{0,\dots,\ell\}$ and both $\ell \in \{L,L-1\}$
  are bounded, then $\lim_{k\to \infty} k^m \alpha_{L+1,m;k} = 0$ for all $m\in \{0,\dots,L+1\}$.
\end{proof}

\begin{proposition} \label{proposition:PhiWeyl}
  The unital $^*$\=/homomorphism $[\argument]_\mu \colon \Weyl(\RR^{1+n})^{\lie g} \to \Weyl(\RR^n)$ is positive.
\end{proposition}
\begin{proof}
  Given a Hermitian and positive elment $a$ of $\Weyl(\RR^{1+n})^{\lie g}$, then in order to show that $[a]_\mu \in \Weyl(\RR^n)^+_\Hermitian$
  it is sufficient to show that $\lim_{k\to \infty} \big(\iota_k^* a \iota_k\big)(\phi) = [a]_{\mu}(\phi)$
  holds for all $\phi\in \Schwartz(\RR^n)$ with respect to the topology on $\Schwartz(\RR^n)$ that is induced by the inner product, 
  because then $\skal{\phi}{[a]_\mu(\phi)} = \lim_{k\to \infty} \skal{\iota_k(\phi)}{a\big(\iota_k(\phi)\big)} \ge 0$
  for all $\phi\in \Schwartz(\RR^n)$.
  
  As $q_j \iota_k = \iota_k q_j$ and $p_j \iota_k = \iota_k p_j$ for all $j\in \{1,\dots,n\}$ and all $k\in \NN$
  by the previous Lemma~\ref{lemma:p0iotak}, and as $\iota_k^* \iota_k = \id_{\Schwartz(\RR^n)}$ by Lemma~\ref{lemma:iotak},
  one has $\big(\iota_k^*a\iota_k\big)(\phi) = [a]_\mu(\phi)$ for all $a \in \genSAlg{q_1,\dots,q_n,p_1,\dots,p_n}$
  and all $k\in \NN$, and this clearly also holds in the limit $k \to \infty$.
  Using that $p_0$ is central in $\Weyl(\RR^{1+n})^{\lie g}$ and starting with the highest power of $p_0$,
  every $a \in \Weyl(\RR^{1+n})^{\lie g}$ admits a decomposition as $a = \sum_{\ell=0}^L (p_0-\mu \Unit)^\ell a_\ell$
  with some $L\in \NN_0$ and with $a_\ell \in \genSAlg{q_1,\dots,q_n,p_1,\dots,p_n}$ for all $\ell \in \{0,\dots,L\}$.
  Then $a_\ell \iota_k = \iota_k [a_\ell]_\mu$ for all $k \in \NN$, $\ell \in \{0,\dots,L\}$, and the previous Lemma~\ref{lemma:p0iotak} shows that
  \begin{align*}
    \iota_k^* a \iota_k
    =
    \sum_{\ell=0}^L \iota_k^* (p_0-\mu \Unit)^\ell \iota_k [a_\ell]_\mu
    =
    [a_0]_\mu + \sum_{\ell=1}^L \sum_{m=0}^\ell \alpha_{\ell,m;k} \iota_k^* q_0^m \iota_k [a_\ell]_\mu
  \end{align*} 
  for all $k\in \NN$. It is easy to see that $\iota_k^* q_0^m \iota_k = c_{m,k} \id_{\Schwartz(\RR^n)}$
  for all $m\in \NN_0$ and all $k\in \NN$, with prefactors $c_{m,k} \in \CC$ given explicitly by
  \begin{align*}
    c_{m,k}
    =
    \frac{1}{k}\integral{\RR}{x_0^m \E^{-\pi x_0^2 / k^2}}{\D x_0}
    =
    k^m \integral{\RR}{y^m \E^{-\pi y^2}}{\D y}
    ,
  \end{align*}
  which is proportional to $k^m$ (for fixed $m \in \NN_0$). So
  $\alpha_{\ell,m;k} \iota_k^* q_0^m \iota_k [a_\ell]_\mu \xrightarrow{k\to \infty} 0$
  for all $\ell \in \{1,\dots,L\}$ and all $m\in \{0,\dots,\ell\}$ as a consequence of the estimates
  from the previous Lemma~\ref{lemma:p0iotak}.
  It follows that $\lim_{k\to \infty} \big(\iota_k^* a \iota_k\big)(\phi) = [a_0]_\mu(\phi) = [a]_\mu(\phi)$ for all $\phi\in \Schwartz(\RR^n)$.
\end{proof}
As the Weyl algebra $\Weyl(\RR^{1+n})$ is an instance of Example~\ref{example:comaspoi}, i.e.~its Poisson bracket is just
the rescaled commutator, the assumptions of Theorem~\ref{theorem:vanishing:neu} are automatically fulfilled (Poisson commuting
elements commute and the momentum $\mu$ is regular). Because of this it makes sense to determine the
quadratic module $\mathcal{R}_\mu$ and the non-commutative vanishing ideal $\mathcal{V}_\mu$ from Definition~\ref{definition:regular}:
\begin{proposition} \label{proposition:weylvr}
  We have $\mathcal{V}_\mu = \genSId{p_0-\mu}$ and 
  \begin{equation}
    \big(\genSId{p_0-\mu}\big)_\Hermitian + \big(\Weyl(\RR^{1+n})^{\lie g} \big)^+_\Hermitian
    =
    \mathcal{R}_\mu
    =
    \set[\big]{a\in \Weyl(\RR^{1+n})^{\lie g}_\Hermitian }{[a]_\mu \in \Weyl(\RR^{n})^+_\Hermitian }
    ,
  \end{equation}
  where $(\genSId{p_0-\mu})_\Hermitian \coloneqq \genSId{p_0-\mu} \cap \Weyl(\RR^{1+n})^{\lie g}_\Hermitian$.
\end{proposition}
\begin{proof}
  As $\ker [\argument]_\mu = \genSId{p_0-\mu}$, the previous Proposition~\ref{proposition:PhiWeyl} especially shows that
  every linear functional $\chi_\phi \circ [\argument]_\mu \colon \Weyl(\RR^{1+n})^{\lie g} \to \CC$, for any vector state
  $\chi_\phi \in \States\big( \Weyl(\RR^n) \big)$ with $\phi \in \Schwartz(\RR^{n})$, $\seminorm{}{\phi} = 1$,
  is a state on $\Weyl(\RR^{1+n})^{\lie g}$, and $\chi_\phi \circ [\argument]_\mu$ even is an eigenstate of $p_0$ with eigenvalue $\mu$ by Proposition~\ref{proposition:genIdeal}.
  From this it follows that $[\mathcal{V}_\mu]_\mu \subseteq \{0\}$, so $\mathcal{V}_\mu \subseteq \ker [\argument]_\mu = \genSId{p_0-\mu}$,
  and $[\mathcal{R}_\mu]_\mu \subseteq \Weyl(\RR^{n})^+_\Hermitian$.
  Conversely, $\mathcal{V}_\mu \supseteq \genSId{p_0-\mu}$ holds in general, see Theorem~\ref{theorem:vanishing:neu}, so $\mathcal{V}_\mu = \genSId{p_0-\mu}$.
  
  The inclusion $(\genSId{p_0-\mu})_\Hermitian \subseteq \mathcal{R}_\mu$ holds as a consequence of Proposition~\ref{proposition:genIdeal},
  it is clear that $(\Weyl(\RR^{1+n})^{\lie g} )^+_\Hermitian \subseteq \mathcal{R}_\mu$, and
  $\mathcal{R}_\mu \subseteq \set{a\in \Weyl(\RR^{1+n})^{\lie g}_\Hermitian }{[a]_\mu \in \Weyl(\RR^{n})^+_\Hermitian }$
  holds because $[\mathcal{R}_\mu]_\mu \subseteq \Weyl(\RR^{n})^+_\Hermitian$.
  Finally, let any $a \in \Weyl(\RR^{1+n})^{\lie g}_\Hermitian$ with $[a]_\mu \in \Weyl(\RR^{n})^+_\Hermitian$ be given.
  By Lemma~\ref{lemma:weylalgdec} there exist unique $b \in \genSId{p_0-\mu}$ and $c \in \genSAlg{q_1,\dots,q_n,p_1,\dots,p_n}$
  such that $a = b+c$, and from $a=a^*$ it follows that $b=b^*$ and $c=c^*$ because $\genSId{p_0-\mu}$ and $\genSAlg{q_1,\dots,q_n,p_1,\dots,p_n}$
  are stable under $\argument^*$. It only remains to show that $c \in \big(\Weyl(\RR^{1+n})^{\lie g} \big)^+_\Hermitian$:
  
  Note that $[c]_\mu = [a]_\mu \in \Weyl(\RR^n)^+_\Hermitian$. For all $\psi \in \Schwartz(\RR^{1+n})$ and all $x_0 \in \RR$, the function 
  $\psi_{x_0} \colon \RR^n \to \CC$, $(x_1,\dots,x_n) \mapsto \psi_{x_0}(x_1,\dots,x_n) \coloneqq \psi(x_0,x_1,\dots,x_n)$
  is an element of $\Schwartz(\RR^n)$, and $(c(\psi))_{x_0} = [c]_\mu(\psi_{x_0})$ holds because this can easily be checked
  for products of the generators $q_j, p_j$ with $j\in \{1,\dots,n\}$. Therefore
  \begin{align*}
    \skal[\big]{\psi}{c(\psi)}
    =
    \integral{\RR}{\skal[\big]{\psi_{x_0}}{(c(\psi))_{x_0}}}{\D x_0}
    =
    \integral{\RR}{\skal[\big]{\psi_{x_0}}{[c]_\mu (\psi_{x_0})}}{\D x_0}
    \ge
    0
  \end{align*}
  holds for all $\psi \in \Schwartz(\RR^{1+n})$, so indeed $c \in (\Weyl(\RR^{1+n})^{\lie g} )^+_\Hermitian$.
\end{proof}

\begin{theorem}
  The tuple $\big( \Weyl(\RR^n), [\argument]_\mu \big)$ is the $p_0$-reduction at $\mu$ of the Weyl algebra $\Weyl(\RR^{1+n})$.
\end{theorem}
\begin{proof}
  As $[\argument]_\mu \colon \Weyl(\RR^{1+n})^{\lie g} \to \Weyl(\RR^n)$ is a positive unital $^*$\=/homomorphism
  by Proposition~\ref{proposition:PhiWeyl} and automatically is compatible with Poisson brackets, 
  we only have to check that the two conditions of Theorem~\ref{theorem:vanishing:neu} for the $p_0$-reduction at $\mu$ 
  are fulfilled:
  
  Surjectivity of $[\argument]_\mu \colon \Weyl(\RR^{1+n})^{\lie g} \to \Weyl(\RR^n)$ is clear, and 
  the previous Proposition~\ref{proposition:weylvr} shows that $\ker [\argument]_\mu = \genSId{p_0-\mu} = \mathcal{V}_\mu$ and that
  every $a\in \Weyl(\RR^{1+n})^{\lie g}_\Hermitian$ which fulfils $[a]_\mu \in \Weyl(\RR^n)^+_\Hermitian$ is an element of $\mathcal{R}_\mu$.
\end{proof}
While this result for the reduction of the Weyl algebra is the naively expected one, which might be seen as further justification for the general definition
of the reduction in Section~\ref{sec:reduction:general}, it should be noted that this example is ill-behaved with respect to the 
reduction of states discussed in Section~\ref{sec:reducedStates}, so the reduction of representable Poisson $^*$\=/algebras
behaves better than the reduction of their states:
\begin{proposition} \label{proposition:noeigenstate}
  There is no eigenstate of $p_0$ on $\Weyl(\RR^{1+n})$.
\end{proposition}
\begin{proof}
  By Proposition~\ref{proposition:eigenstates}, any eigenstate $\omega$ of $p_0$ on $\Weyl(\RR^{1+n})$ would have to fulfil
  \begin{align*}
    0 = \dupr{\omega}{q_0}\dupr{\omega}{p_0} - \dupr{\omega}{p_0}\dupr{\omega}{q_0} = \dupr{\omega}{q_0p_0-p_0q_0} = \I \dupr{\omega}{\Unit} = \I
    .
  \end{align*}
\end{proof}
\begin{corollary}
  There is no averaging operator $\argument_\av \colon \Weyl(\RR^{1+n}) \to \Weyl(\RR^{1+n})^{\lie g}$.
\end{corollary}
\begin{proof}
  The existence of an averaging operator would lead to a contradiction between the previous Proposition~\ref{proposition:noeigenstate}
  and Theorem~\ref{theorem:reducedStates}, part~\refitem{item:reducedStates:lift}.
\end{proof}
The non-existence of an averaging operator in this case is surprising in so far as in the analogous 
commutative case there does exist an averaging operator for the translation in the $0$-coordinate:
Consider the polynomial algebra with pointwise order
on the cotangent space $\Cotangent \RR^{1+n}$ with standard coordinates $q_0,\dots,q_n,p_0,\dots,p_n$ and Poisson bracket obtained from the canonical
symplectic form. An averaging operator is given by restricting polynomials to the hyperplane
of $\Cotangent \RR^{1+n}$ where $q_0$ vanishes and extending the result to polynomials constant in 
$q_0$-direction. This, however, does no longer work for the Weyl algebra
because $q_0^2 + p_0^2 - \Unit \in \Weyl(\RR^{1+n})^+_\Hermitian$ but $p_0^2 - \Unit \notin \Weyl(\RR^{1+n})^+_\Hermitian$.

\section{Reduction of the Polynomial Algebra} \label{sec:polynomials}
Fix some $n \in \NN$ for the rest of this section.
It is well-known that the complex projective space $\CC\PP^n$ can be obtained as the
quotient of the $(1+2n)$-sphere $\mathbb S^{1+2n} \subseteq \CC^{1+n}$ by the action of $\group 
U(1)$ via multiplication.
A similar construction for different signatures yields the hyperbolic disc as a quotient, see e.g.\ \cite{schmitt.schoetz:PreprintWickRotationsInDQ}.
This procedure can be understood as Marsden--Weinstein reduction,
and in this section we discuss its algebraic analogue in terms of representable Poisson $^*$\=/algebras.
One has a choice of the class of functions that one wants to consider:
Taking the smooth functions $\Smooth(\CC^{1+n})$ results in a special case of the
reduction of Poisson manifolds as described in Section~\ref{sec:PoiMan}.
Another choice is to work with the polynomial functions $\Polynomials(\CC^{1+n})$,
which allows a non-formal deformation to the non-commutative setting.
This deformation is the main topic of Part II of this article,
and here we lay the foundations by describing the reduction in the commutative case.

For $i\in \{0,\dots,n\}$, let $z_i, \cc z_i \colon \CC^{1+n} \to \CC$, $z_i(w) = w_i$, $\cc 
z_i(w) = \cc w_i$ be the standard coordinates and their complex conjugates on $\CC^{1+n}$.
The algebra of polynomials in $z_i$ and $\cc z_i$ will be denoted by $\Polynomials(\CC^{1+n})$.
It is $(\ZZ \times \ZZ)$-graded by the holomorphic and antiholomorphic degree,
i.e.\ $\Polynomials(\CC^{1+n}) = \bigoplus_{K, L \in \ZZ} \Polynomials(\CC^{1+n})^{K,L}$ 
with $\Polynomials(\CC^{1+n})^{K,L}$ spanned by $z_0^{k_0} \dots z_{n}^{k_n} \cc 
z_0^{\ell_0} \dots \cc z_n^{\ell_n}$ with $k_0, \dots, k_n, \ell_0, \dots, \ell_n \in \NN_0$
satisfying $k_0 + \dots + k_n = K$ and $\ell_0 + \dots + \ell_n = L$
if $K,L \geq 0$, and $\Polynomials(\CC^{1+n})^{K,L} =\{0\}$ otherwise.

Various objects in this section will depend on the choice of a signature $s \in \lbrace 1, \dots, 1+n \rbrace$.
One of them is a tuple of coefficients $\nu^{(s)} \in \lbrace -1, 1 \rbrace^{1+n}$,
which is defined as $\nu^{(s)}_i \coloneqq 1$ if $i \in \{0, \dots, s-1\}$ and $\nu^{(s)}_i 
\coloneqq -1$ if $i \in \{ s, \dots, n\}$.
We will usually omit the superscript $^{(s)}$ from our notation, e.g.\ we will write $\nu_i$
instead of $\nu_i^{(s)}$. 
It is easy to check that for any signature $s$,
\begin{equation}
	\poi f g \coloneqq \poi f g^{(s)} \coloneqq \frac 1 \I \sum_{j=0}^n \nu_j \bigg(\frac{\partial f}{\partial \cc z_j} \frac{\partial g}{\partial z_j} - \frac{\partial f}{\partial z_j} \frac{\partial g}{\partial \cc z_j}\bigg)
\end{equation}
defines a Poisson bracket on $\Polynomials(\CC^{1+n})$.
Let $\lie u_1 = \set {\I \alpha \Unit} {\alpha \in \RR}$ be the abelian Lie algebra of the Lie 
group $\group U(1)$ and consider the momentum map $\momentmap \colon \lie u_1 \to 
\Polynomials(\CC^{1+n})$, $\I \alpha \Unit \mapsto \alpha \sum_{i=0}^n \nu_i z_i \cc z_i$.
Since $\lie u_1$ is $1$\=/dimensional, $\momentmap$ is uniquely determined by the image of $\I 
\Unit$ and we abuse notation and write $\momentmap \coloneqq \momentmap^{(s)} \coloneqq 
\momentmap(\I \Unit) = \sum_{i=0}^n \nu_i z_i \cc z_i$ also for this image.
The momentum map $\momentmap$ induces a right action $\argument \racts \argument \colon \Polynomials(\CC^{1+n}) 
\times \lie u_1 \to \Polynomials(\CC^{1+n})$ by derivations,
which is, on monomials, given explicitly by
\begin{equation}
	z_k \racts (\I \alpha \Unit) 
	= 
	\alpha \poi{z_k}{\momentmap} 
	=
	\I\alpha z_k 
	\quad\text{and}\quad
	\cc z_\ell \racts (\I \alpha \Unit) 
	= 
	\alpha \poi{\cc z_\ell}{\momentmap} 
	=
	-\I\alpha \cc z_\ell
\end{equation}
for all $k,\ell \in \{0,\dots, n\}$.
In particular, this action integrates to the usual action of the Lie group $\group U(1)$ on 
$\Polynomials(\CC^{1+n})$ by automorphisms,
$z_k \racts \E^{\I \alpha} = \E^{\I \alpha} z_k$ and
$\cc z_\ell \racts \E^{\I \alpha} = \E^{-\I \alpha} \cc z_\ell$. Again, we use the same symbol for the actions of the Lie group $\group{U}(1)$
and of its Lie algebra $\lie u_1$.
Since $\group U(1)$ is connected, $f \in \Polynomials(\CC^{1+n})$ is $\lie u_1$-invariant if and 
only if it is $\group U(1)$-invariant, 
which is the case if and only if $f \in \bigoplus_{k \in \NN_0} \Polynomials(\CC^{1+n})^{k,k}$.

We endow $\Polynomials(\CC^{1+n})$ with the pointwise order, i.e.\ the order induced by the evaluation functionals
$\delta_w \colon \Polynomials(\CC^{1+n}) \to \CC$, $f \mapsto \dupr{\delta_w}{f} \coloneqq 
f(w)$ with $w\in \CC^{1+n}$ like in Proposition~\ref{proposition:inducedOrder}. 
This order is, by construction, induced by its states and therefore 
$\Polynomials(\CC^{1+n})$ becomes a representable Poisson $^*$-algebra. It is well-known that the 
pointwise order on the polynomial algebra does not coincide with the algebraic order,
and, contrary to the case of smooth functions, there even exist algebraically positive Hermitian linear functionals on $\Polynomials(\CC^{1+n})$
that are not positive with respect to the pointwise order, see e.g.\ \cite[Thms.~10.36, 
10.37]{schmuedgen:invitationToStarAlgebras}.

Since the $\group U(1)$-action on $\Polynomials(\CC^{1+n})$ preserves the degree,
it follows for fixed $f \in \Polynomials(\CC^{1+n})$ that all polynomials $f \racts u$ for $u \in 
\group U(1)$ lie in a finite dimensional subspace of $\Polynomials(\CC^{1+n})$.
This makes it easy to check that 
$\argument_\av \colon \Polynomials(\CC^{1+n}) \to \Polynomials(\CC^{1+n})^{\lie u_1}$,
\begin{equation} \label{eq:averagingOpForUOne}
	f \mapsto f_\av 
	\coloneqq 
	\frac 1 {2 \pi} \int_0^{2 \pi} (f \racts \E^{\I \alpha}) \,\D \alpha \komma
\end{equation}
defines an averaging operator in the sense of Definition~\ref{definition:averaging}.

\subsection{The Reduction}

For a momentum $\mu \in \lie u_1^*$ we again abuse notation 
and write $\mu \coloneqq \mu(\I \Unit) \in \RR$ also for the image of $\I \Unit$. In this sense, we will always assume that $\mu > 0$.
The goal of this section is to determine the $\momentmap$\=/reduction  
of $\Polynomials(\CC^{1+n})$ at $\mu$, so we first determine an explicit description of the quadratic
module $\mathcal{R}_\mu \coloneqq \mathcal{R}_\mu^{(s)}$ and of the $^*$\=/ideal $\mathcal{V}_\mu \coloneqq \mathcal{V}_\mu^{(s)}$ from Definition~\ref{definition:regular}.
We denote the $\mu$-levelset of $\momentmap$ by
\begin{equation}
  \mathcal Z_\mu \coloneqq \mathcal Z_\mu^{(s)} \coloneqq \set[\big] { w \in \CC^{1+n} }{ \momentmap(w) = \mu }
  \punkt
\end{equation}

\begin{lemma} \label{lemma:vrPoly}
  Using the notation $(\genSId{\momentmap-\mu})_\Hermitian \coloneqq \genSId{\momentmap-\mu} \cap \Polynomials(\CC^{1+n})^{\lie u_1}_\Hermitian$ we have
  \begin{align}
    \genSId{\momentmap-\mu}
    &=
    \mathcal{V}_\mu
    =
    \set[\big]{f\in \Polynomials(\CC^{1+n})^{\lie u_1}}{ f(w) = 0 \textup{ for all }w\in \Levelset_\mu }
  \shortintertext{and}
    \big(\genSId{\momentmap-\mu}\big)_\Hermitian + \big( \Polynomials(\CC^{1+n})^{\lie u_1} \big)^+_\Hermitian
    &=
    \mathcal{R}_\mu
    =
    \set[\big]{f \in \Polynomials(\CC^{1+n})^{\lie u_1}_\Hermitian}{ f(w) \ge 0 \textup{ for all }w\in \Levelset_\mu }
    \punkt
    \label{eq:RmuPolyChar}
  \end{align}
\end{lemma}
\begin{proof}
  From Proposition~\ref{proposition:genIdeal} it follows that $\genSId{\momentmap-\mu} \subseteq \mathcal{V}_\mu$ and that
  $(\genSId{\momentmap-\mu})_\Hermitian \subseteq \mathcal{R}_\mu$. The inclusion
  $( \Polynomials(\CC^{1+n})^{\lie u_1} )^+_\Hermitian \subseteq \mathcal{R}_\mu$ is clear. Moreover,
  as all evaluation functionals $\delta_w \colon \Polynomials(\CC^{1+n})^{\lie u_1} \to \CC$, $f\mapsto \dupr{\delta_w}{f} \coloneqq f(w)$
  with $w\in \Levelset_\mu$ are eigenstates of $\momentmap$ with eigenvalue $\mu$, every $f\in \mathcal{V}_\mu$
  fulfils $f(w) = 0$ for all $w\in \Levelset_\mu$, and every $f\in \mathcal{R}_\mu$
  fulfils $f(w) \ge 0$ for all $w\in \Levelset_\mu$.
  It only remains to show that every $f \in \Polynomials(\CC^{1+n})^{\lie u_1}$
  that fulfils $f(w) = 0$ for all $w\in \Levelset_\mu$ is an element of $\genSId{\momentmap-\mu}$,
  and that every $f \in \Polynomials(\CC^{1+n})^{\lie u_1}_\Hermitian$
  that fulfils $f(w) \ge 0$ for all $w\in \Levelset_\mu$ is an element of $(\genSId{\momentmap-\mu})_\Hermitian + ( \Polynomials(\CC^{1+n})^{\lie u_1} )^+_\Hermitian$:
  
  For any $f \in \Polynomials(\CC^{1+n})^{\lie u_1}$, define its homogenization
  $f_{\mathrm h} \coloneqq \sum_{\ell = 0}^d (\momentmap / \mu)^{d-\ell}f_\ell \in \Polynomials(\CC^{1+n})^{d,d}$,
  where $f_\ell \in \Polynomials(\CC^{1+n})^{\ell,\ell}$, $\ell \in \NN_0$,
  are the homogeneous components of $f$ so that $f = \sum_{\ell = 0}^\infty f_\ell$, 
  and where $d\in \NN_0$ is minimal such that $f_\ell = 0$ for all $\ell \in \NN_0$ with $\ell > d$.
  Then $f_{\mathrm h}(w) = f(w)$ for all $w\in \Levelset_\mu$ and
 	\begin{equation*}
    f - f_{\mathrm h}
		=
		\sum_{\ell=0}^{d} \frac{\mu^{d-\ell}\Unit - \momentmap^{d-\ell}}{\mu^{d-\ell}} f_\ell 
		= 
		(\mu \Unit- \momentmap) \sum_{\ell=0}^{d-1} \frac{f_\ell}{\mu^{d-\ell}} \sum_{k=1}^{d-\ell} \mu^{d-\ell-k} \momentmap^{k-1}
		\in
		\genSId{\momentmap-\mu}
		.
		\tag{$\clubsuit$}
		\label{eq:vrPoly:intern}
	\end{equation*}
	Note also that $f_{\mathrm h}$ is Hermitian if $f$ is Hermitian.
	
	Now consider the case of a polynomial $f \in \Polynomials(\CC^{1+n})^{\lie u_1}$
  that fulfils $f(w) = 0$ for all $w\in \Levelset_\mu$. Then 
  $f_{\mathrm h}(w) = \momentmap(w)^{d/2} f_{\mathrm h}\big(\momentmap(w)^{-1/2} w\big)= \momentmap(w)^{d/2} f\big(\momentmap(w)^{-1/2} w\big) = 0$
  for all $w \in \CC^{1+n}$ with $\momentmap(w) > 0$, which form an open and non-empty subset of $\CC^{1+n}$. Therefore $f_{\mathrm h} = 0$
  and $f = f-f_{\mathrm h} \in \genSId{\momentmap-\mu}$ by \eqref{eq:vrPoly:intern}.
  
  Similarly, if a polynomial $f \in \Polynomials(\CC^{1+n})^{\lie u_1}_\Hermitian$
  fulfils $f(w) \ge 0$ for all $w\in \Levelset_\mu$, then also $f_{\mathrm h}(w) \ge 0$
  for all $w \in \CC^{1+n}$ with $\momentmap(w) > 0$. Consequently,
  $f^\sim \coloneqq (2 \mu^2)^{-1} ( \momentmap - \mu \Unit )^2 \big((f_{\mathrm h})^2+\Unit\big) + f_{\mathrm h}$
  fulfils $f^\sim (w) \ge 0$ for all $w\in \CC^{1+n}$ with $\momentmap(w) > 0$,
  and for $w\in \CC^{1+n}$ with $\momentmap(w) \le 0$ one also finds that $f^\sim(w) \ge 0$
  because then $(2 \mu^2)^{-1} ( \momentmap(w) - \mu )^2 \ge 1/2$ and $f_{\mathrm h}(w)^2+1 \ge 2 \abs{ f_{\mathrm h}(w) }$.
  We therefore have $f^\sim \in (\Polynomials(\CC^{1+n})^{\lie u_1})^+_\Hermitian$, and as
  $f-f_{\mathrm h} \in (\genSId{\momentmap-\mu})_\Hermitian$ by \eqref{eq:vrPoly:intern}
  and clearly also $f_{\mathrm h} - f^\sim \in (\genSId{\momentmap-\mu})_\Hermitian$, it follows that
  $f \in (\genSId{\momentmap-\mu} )_\Hermitian + ( \Polynomials(\CC^{1+n})^{\lie u_1} )^+_\Hermitian$.
\end{proof}
Recall that the complex projective space $\CC\PP^n$ is the quotient manifold $\big( \CC^{1+n} \setminus \{0\} \big) / \CC^*$,
where the multiplicative group $\CC^* = \CC\setminus \{0\}$ acts on $\CC^{1+n} \setminus \{0\}$ by scalar multiplication.

\begin{definition} \label{definition:polynomialsOnCPn}
  Let
  \begin{equation}
    \Mred \coloneqq \Mred^{(s)} \coloneqq \set[\big] {[w] \in \CC\PP^n }{ \momentmap(w) > 0}
    \punkt
  \end{equation}
  For $f \in \Polynomials(\CC^{1+n})^{\lie u_1}$ we define $\reductionMap{f} \colon \Mred \to \CC$
	by $\reductionMap{f}([w]) \coloneqq f\at{\Levelset_\mu}(w)$
	where $w \in \mathcal Z_\mu$ is a representative of $[w] \in \Mred$.
	We call 
	\begin{equation}
	\Polynomials(\Mred) 
	\coloneqq 
  \set[\big]{[f]_\mu}{f\in \Polynomials(\CC^{1+n})^{\lie u_1}}
	\end{equation} 
	the space of \emph{polynomials} on $\Mred$.
\end{definition}
Note that $\Mred$ is well-defined since the choice of representative $w \in \CC^{1+n} \setminus \{0\}$
for $[w] \in \CC\PP^n$ has no influence on the sign of $\momentmap(w)$,
that $\mathcal Z_\mu$ and $\Mred$ depend on the choice of signature $s$ since $\momentmap$ does,
and that every $[w] \in \Mred$ has a (non-unique) representative $w \in \Levelset_\mu$.
It is easy to check that $\Polynomials(\Mred)$ with the usual pointwise multiplication of functions,
the pointwise complex conjugation as $^*$\=/involution, and the pointwise order becomes an ordered 
$^*$\=/algebra whose order is induced by its states and
$\reductionMapSign \colon \Polynomials(\CC^{1+n})^{\lie u_1} \to \Polynomials(\Mred)$ is a positive unital $^*$\=/homomorphism.

\begin{theorem} \label{theorem:polyreduction}
  The momentum $\mu > 0$ is regular for $\momentmap$, and $\ker[\argument]_\mu = \genSId{\momentmap-\mu}$.
  Moreover,
  \begin{equation}
    \poi[\big]{[f]_\mu}{[g]_\mu} \coloneqq \big[ \poi{f}{g} \big]_\mu
    \label{eq:polyreduction:bracket}
  \end{equation}
  for all $f,g \in \Polynomials(\CC^{1+n})^{\lie u_1}$ defines a Poisson bracket on $\Polynomials(\Mred)$ and
  the tuple $\big(\Polynomials(\Mred), \reductionMapSign\big)$ is the $\momentmap$\=/reduction of $\Polynomials(\CC^{1+n})$ at $\mu$.
\end{theorem}
\begin{proof}
  The kernel of $\reductionMapSign$ clearly consists of precisely those elements of $\Polynomials(\CC^{1+n})^{\lie u_1}$ that vanish on $\mathcal Z_\mu$, so
  $\genSId{\momentmap-\mu} = \mathcal{V}_\mu = \ker \reductionMapSign$ by Lemma~\ref{lemma:vrPoly}, which is
  a Poisson ideal of $\Polynomials(\CC^{1+n})^{\lie u_1}$ by Proposition~\ref{proposition:genIdeal}.
  So $\mu$ is regular for $\momentmap$ and the Poisson bracket of $\Polynomials(\CC^{1+n})^{\lie u_1}$ descends to a well-defined Poisson 
  bracket \eqref{eq:polyreduction:bracket} on $\Polynomials(\Mred)$. With this Poisson bracket and with the pointwise order,
  $\Polynomials(\Mred)$ is a representable Poisson $^*$\=/algebra and $\reductionMapSign$ is a positive unital $^*$\=/homomorphism
  compatible with Poisson brackets.
  Moreover, Theorem~\ref{theorem:vanishing:neu} applies,
  and as Lemma~\ref{lemma:vrPoly} also shows that $\mathcal{R}_\mu$ coincides with the set of Hermitian elements in 
  the preimage of $\Polynomials(\Mred)^+_\Hermitian$
  under $[\argument]_\mu$, it is easy to check that $\big(\Polynomials(\Mred), \reductionMapSign\big)$ is the
  $\momentmap$\=/reduction of $\Polynomials(\CC^{1+n})$ at $\mu$.
\end{proof}
Note that this Poisson bracket on $\Mred$ coincides with the one of the Fubini--Study symplectic form (with signature),
which is obtained by Marsden--Weinstein reduction of $\CC^{1+n}$ with respect to the $\group U(1)$-action.
In particular, if $s=1+n$, then $\Mred^{(1+n)} \cong \CC\PP^n$ as Poisson manifolds,
and if $s=1$, then $\Mred^{(1)}$ and the complex hyperbolic disc $\DD^n$ are isomorphic.
See \cite{schmitt.schoetz:PreprintWickRotationsInDQ} for details.
Moreover, only the Poisson bracket of $\Polynomials(\Mred)$ actually depends on $\mu$, but not the underlying ordered $^*$\=/algebra.

\subsection{Comparing Pointwise and Algebraic Order}
One key point why the reduction of representable Poisson $^*$\=/algebras works well 
is that we have a certain freedom in choosing the order on the reduction, 
in the sense that we do not always endow $^*$-algebras with their canonical algebraic order 
(the one whose positive elements are sums of Hermitian squares).
It is a priori unclear how the order on the reduced algebra $\Polynomials(\Mred)$,
i.e.\ the pointwise order, relates to the algebraic order. In this last section we want to discuss this relation,
which can be described using the methods of real algebraic geometry.

The first step is to identify the unital $^*$\=/homomorphisms from $\Polynomials(\Mred)$ to $\CC$. 
As we show below, they are in bijection with
\begin{equation}
\MredExt \coloneqq \MredExt^{(s)} \coloneqq \set[\big] {[w] \in \CC\PP^n }{ \momentmap(w) \neq 0} 
\punkt
\end{equation}
If 
$s = 1+n$, then $\momentmap(w) > 0$ for all $w \in \CC^{1+n} \setminus \{0\}$
and therefore $\Mred^{(1+n)} = \MredExt^{(1+n)}$. For $s \neq 1+n$, however, $\Mred^{(s)} \subsetneq \MredExt^{(s)}$.
In the special case $s = 1$ it was already observed in \cite{kraus.roth.schoetz.waldmann:OnlineConvergentStarProductOnPoincareDisc}
that there are more $\CC$-valued unital $^*$\=/homomorphisms on $\Polynomials(\Mred^{(1)})$
than one would naively expect, i.e.~not only evaluation functionals at points of $\Mred^{(1)}$.
	
\begin{lemma} \label{lemma:matrixLemma}
	For every $[w] \in \MredExt$, the matrix 
	$X = \big(\momentmap(w)^{-1} w_i \cc{w_j} \big)_{i,j\in\{0,\dots,n\}} \in \CC^{(1+n) \times (1+n)}$
	is well-defined and fulfils 
	\begin{equation}
    (\hat \nu X)^2 = \hat \nu X \komma\quad
    X^* = X \komma \quad{and}\quad	 
    \tr(\hat \nu X) = 1
    \komma
    \label{eq:matrixLemma}
	\end{equation}
	where $\hat\nu \coloneqq \mathrm{diag}(\nu_0, \dots, \nu_n)$ is the diagonal matrix
  	with entries $\nu_0, \dots, \nu_n$ along the diagonal.
  	Conversely, for any $X \in \CC^{(1+n) \times (1+n)}$ satisfying the conditions \eqref{eq:matrixLemma}
  	there exists a unique $[w] \in \MredExt$ such that $X = \big(\momentmap(w)^{-1} w_i \cc{w_j} \big)_{i,j\in\{0,\dots,n\}}$. 
\end{lemma}

\begin{proof}
	It is clear that the matrix $X \coloneqq \big(\momentmap(w)^{-1} w_i \cc{w_j} \big)_{i,j\in\{0,\dots,n\}}$
	is well-defined and fulfils $X=X^*$. Moreover, 
	$((\hat \nu X)^2)_{ij} = \sum_{k=0}^n \momentmap(w)^{-2} \nu_i w_i \cc{w_k} \nu_k w_k \cc{w_j} 
	= \momentmap(w)^{-1} \nu_i w_i \cc{w_j} = (\hat \nu X)_{ij}$ 
	holds for all $i,j\in\{0,\dots,n\}$ and $\tr(\hat{\nu} X) = \momentmap(w)^{-1}\sum_{i=0}^n \nu_i w_i \cc{w_i} = 1$.
	
	Conversely, assume that $X \in \CC^{(1+n) \times (1+n)}$ satisfies \eqref{eq:matrixLemma}.
	Then $\hat \nu X$ is diagonalizable with eigenvalues contained in $\lbrace 0, 1 \rbrace$ because $\hat \nu X$ is 
	idempotent. 
	Since $\tr(\hat \nu X) = 1$ it follows that the eigenvalue $1$ occurs precisely once,
	so that the image of $\hat \nu X$ has dimension $1$. 
	Consequently, there exist non-zero $v, w \in \CC^{1+n}$ such that
	$(\hat \nu X)_{ij} = v_i \cc{w_j}$ for all $i,j\in \{0,\dots,n\}$,
	or equivalently $X_{ij} = \nu_{i} v_i \cc{w}_j$.
	Next, $X^* = X$ implies that $w_i \cc v_j \nu_j = \nu_i v_i \cc w_j$
	holds for all $ i, j \in \{0,\dots,n\}$,
	and it is an easy algebraic manipulation to show that this is enough to guarantee that $\hat\nu v = \lambda w$ for some $\lambda \in \CC 
	\setminus \lbrace 0 \rbrace$, even $\lambda \in \RR \setminus \lbrace 0 \rbrace$.
	Consequently $X_{ij} = \lambda w_i \cc{w_j}$ for all $i,j\in\{0,\dots,n\}$,
	and $\tr(\hat \nu X) = 1$ implies $1 = \lambda \momentmap(w)$,
	so $X_{ij} = \momentmap(w)^{-1} w_i \cc{w_j}$.
	Since the $1$-dimensional range of $X$ is spanned by $w$,
	$w$ is determined uniquely up to multiplication with a non-zero complex constant, so $[w] \in 
	\MredExt$ is determined uniquely.
\end{proof}
For every $[w] \in \MredExt$ we define the unital $^*$\=/homomorphism  $\delta_{[w]} \colon \Polynomials(\Mred) \to \CC$,
\begin{equation}  \label{eq:evaluationFunctionals}
	[f]_\mu \mapsto \dupr{\delta_{[w]}}{[f]_\mu} \coloneqq \sum_{\ell = 0}^\infty f_\ell(w) 
	\bigg(\frac{\mu}{\momentmap (w)}\bigg)^{\!\!\ell}
\end{equation}
where $f=\sum_{\ell = 
	0}^\infty f_\ell \in \Polynomials(\CC^{1+n})^{\lie u_1}$ is any representative of $[f]_\mu \in \Polynomials(\Mred)$ with homogeneous components $f_\ell \in \Polynomials(\CC^{1+n})^{\ell,\ell}$,
and where $w \in \CC^{1+n} \setminus \lbrace 0 \rbrace$ is any representative of $[w] \in \MredExt$.
It is easy to check that $\delta_{[w]}$ is well-defined. 
For $[w] \in \Mred$, this unital $^*$\=/homomorphism $\delta_{[w]}$ is just the usual evaluation functional at $[w]$.

\begin{proposition} \label{proposition:starHomsArePointEvs} 
	For every unital $^*$\=/homomorphism $\varphi \colon \Polynomials(\Mred) \to \CC$
	there exists a unique point $[w] \in \MredExt$ such that $\varphi=\delta_{[w]}$.
\end{proposition}

\begin{proof}
	Consider the matrix $X \in \CC^{(1+n) \times (1+n)}$ with entries $X_{ij} = \mu^{-1} \dupr{\varphi}{[z_i \cc z_j]_\mu}$. 
	Then
	\begin{equation*}
		\big((\hat \nu X)^2\big)_{ij} 
		= \mu^{-2} \sum_{k=0}^n \nu_i  \nu_k \dupr[\big]{\varphi}{[z_i \cc z_k z_k \cc z_j]_\mu} 
		= \mu^{-2}  \nu_i \dupr[\big]{\varphi}{[z_i \cc z_j \momentmap ]_\mu}
		= \mu^{-1}  \nu_i \dupr[\big]{\varphi}{[z_i \cc z_j]_\mu} = (\hat \nu X)_{ij} \komma
	\end{equation*}
	and the other two assumptions of the previous Lemma~\ref{lemma:matrixLemma} are easily checked.
	So there exists a unique $[w] \in \MredExt$ such that $X_{ij} = \momentmap(w)^{-1} w_i \cc w_j$ for all $i,j \in \{0, \dots,n\}$. 
	This means that $\delta_{[w]}$ and $\varphi$ coincide on the generators $[z_i \cc z_j]_\mu$, or equivalently on all of $\Polynomials(\Mred)$.
\end{proof}
The above Proposition~\ref{proposition:starHomsArePointEvs} shows that $\MredExt$ is a real algebraic set, while $\Mred$ is not if $s \neq 1+n$.
However, $\Mred$ is a subset of $\MredExt$ that can be described by a polynomial inequality:
For a commutative $^*$\=/algebra $\A$ and a subset $\mathcal G \subseteq \A_\Hermitian$ we write
\begin{equation}
	\semialgebraicSet{\A}{\mathcal{G}}
	\coloneqq
	\set[\big]{ \varphi \colon \A \to \CC }{ \varphi\text{ a unital $^*$\=/homomorphism fulfilling } \dupr{\varphi}{g} \ge 0 \text{ for all }g \in \mathcal G }
	\punkt
	\label{eq:fanaPositivstellensatz}
\end{equation}

\begin{proposition} \label{proposition:semialgebraic}
	The identity 
	$\semialgebraicSet[\big]{\Polynomials(\Mred)}{ \{\sum_{i=s}^n\reductionMap{z_i \cc z_i} \} } = \set{\delta_{[w]}}{[w] \in \Mred}$ 
	holds.
\end{proposition}

\begin{proof}
	By the previous Proposition~\ref{proposition:starHomsArePointEvs} all unital $^*$\=/homomorphisms from $\Polynomials(\Mred)$ to $\CC$ are of the form $\delta_{[w]}$ with $[w] \in \MredExt$.
	The identity $\dupr{\delta_{[w]}}{\sum_{i=s}^n \reductionMap{z_i \cc z_i}} = \momentmap(w)^{-1} \mu 
  \sum_{i=s}^n \abs{w_i}^2$ holds for all $[w] \in \MredExt$ with representative $w \in \CC^{1+n} \setminus \lbrace 0 \rbrace$.
  Therefore $\dupr{\delta_{[w]}}{\sum_{i=s}^n \reductionMap{z_i \cc z_i}} \geq 0$ holds if and only if 
  $\momentmap (w) > 0$: On the one hand, $\momentmap(w) > 0$ clearly implies $\dupr{\delta_{[w]}}{\sum_{i=s}^n \reductionMap{z_i \cc z_i}} \geq 0$. On the other hand, assume that $\dupr{\delta_{[w]}}{\sum_{i=s}^n \reductionMap{z_i \cc z_i}} \geq 0$. 
  Then either $\sum_{i=s}^n \abs{w_i}^2 > 0$ so that 
  $\momentmap(w) > 0$ by the above identity, 
  or $\sum_{i=s}^n \abs{w_i}^2 = 0$ so that $\momentmap(w) = \sum_{i=0}^{s-1} \abs{w_i}^2 - 
  \sum_{i=s}^n \abs{w_i}^2 > 0$ because $w \neq 0$.
\end{proof}

\begin{corollary}
For a unital $^*$\=/homomorphisms $\varphi \colon \Polynomials(\Mred) \to \CC$ the following are equivalent:
\begin{enumerate}
	\item\label{item:equivalence:i} $\varphi$	is positive with respect to the pointwise order,
	\item\label{item:equivalence:ii} $\dupr{\varphi}{\sum_{i=s}^n\reductionMap{z_i \cc z_i}} \geq 0$,
	\item\label{item:equivalence:iii} there exists $[w] \in \Mred$ such that $\varphi = \delta_{[w]}$.
\end{enumerate}	
\end{corollary}

\begin{proof}
	Using the previous Proposition~\ref{proposition:semialgebraic}, it is easy to see that \refitem{item:equivalence:i} $\Rightarrow$ \refitem{item:equivalence:ii} $\Rightarrow$ \refitem{item:equivalence:iii} $\Rightarrow$ \refitem{item:equivalence:i}. 
\end{proof}
In particular, if $s \neq 1+n$, then there are multiplicative algebraic states on $\Polynomials(\Mred)$ which are not positive.
This shows that it is not possible in general to describe the order on the reduced algebra in the commutative case as the one induced by all
multiplicative algebraic states, even if the order on the original algebra is of this type. From the point of view of
real algebraic geometry, this is not surprising: The type of structures that enjoy a reasonable stability under many constructions
are not the real algebraic sets, but the semialgebraic sets, i.e.~sets that can be defined by a finite number of polynomial inequalities.

One problem arising in real algebraic geometry is to give an algebraic description of quadratic modules or preorderings
that are defined via their (multiplicative) states. If one succeeds, such a result also yields an algebraic description of the states,
and in some cases, these two results are even equivalent.
Recall the definition of the preordering generated by a subset of Hermitian elements of a commutative $^*$\=/algebra
from \eqref{eq:genPO}.
The following theorem is an adaption of the Positivstellensätze of Marshall \cite{marshall:extendingArchimedeanPositivstellensatzToTheNonCOmpactCase}
and Schmüdgen \cite{schmuedgen:kMomentProblemForCompactSemialgebraicSets}
to our setting.
\begin{theorem} \label{theorem:fanaPositivstellensatz}
  Let $\A$ be a finitely generated commutative $^*$\=/algebra and $\{x_1,\dots,x_\ell\}$
  a finite set of Hermitian generators of $\A$.
  Moreover, let $\mathcal G$ be any finite subset of $\A_\Hermitian$. 
  Given $p \in \Unit + \genPO{\mathcal{G}}$ for which there exists $\lambda \in {[0,\infty[}$
  such that $\abs{ \dupr{\varphi}{x_j} } \le \lambda \dupr{\varphi}{p}$ holds for all
  $\varphi\in \semialgebraicSet{\A}{\mathcal{G}}$ and all $j \in \{1,\dots,\ell\}$, then for every $a\in  \A_\Hermitian$, the following are equivalent:
  \begin{enumerate}
    \item $\dupr{\varphi}{a} \ge 0$ for all $\varphi \in \semialgebraicSet{\A}{\mathcal{G}}$.
    \item There is $m_1\in \NN_0$ such that for all $\epsilon \in {]0,\infty[}$ there is $m_2\in \NN_0$
      for which $p^{m_2}(a+\epsilon p^{m_1}) \in \genPO{\mathcal{G}}$.
  \end{enumerate}
  Moreover, such an element $p \in \Unit + \genPO{\mathcal{G}}$ with the property required above exists:
  For example, $p \coloneqq \Unit + \sum_{j=1}^\ell x_j^2$ would always be a valid choice, and if $\semialgebraicSet{\A}{\mathcal{G}}$
  is weak-$^*$-compact, then one can even take $p\coloneqq \Unit$.
\end{theorem}
\begin{proof}
  As $\A$ is commutative and finitely generated, its Hermitian elements form a finitely generated real commutative unital algebra
  $\A_\Hermitian \cong \RR[x_1,\dots,x_\ell] / \mathcal{I}$, where $\mathcal{I}$ is an ideal of the polynomial algebra $\RR[x_1,\dots,x_\ell]$
  and finitely generated because $\RR[x_1,\dots,x_\ell]$ is Noetherian.
  Let $h_1,\dots,h_m \in \mathcal{I}$ be generators of $\mathcal{I}$ and consider the preordering
  $T \coloneqq \genPO{\{g'_1,\dots,g'_k,h_1,\dots,h_m,-h_1,\dots,-h_m\} }$ of $\RR[x_1,\dots,x_\ell]$,
  where $g'_1,\dots,g'_k \in \RR[x_1,\dots,x_\ell]$ denote any representatives of the elements $g_1,\dots,g_k$ of $\mathcal G$.
  Then \cite[Cor.~3.1]{marshall:extendingArchimedeanPositivstellensatzToTheNonCOmpactCase} applies and gives a
  characterization of those elements $a' \in \RR[x_1,\dots,x_\ell]$ which fulfil $a'(y) \ge 0$
  for all those $y\in \RR^\ell$ for which $g'_i(y) \ge 0$ and $h_j(y) = 0$ hold for all $i\in \{1,\dots,k\}$ and all $j\in \{1,\dots,m\}$.
  After projecting down onto $\A_\Hermitian$, one obtains the above statement. The special case of
  weak-$^*$-compact $\semialgebraicSet{\A}{\mathcal{G}}$ and $p\coloneqq \Unit$ has appeared earlier in \cite[Cor.~3]{schmuedgen:kMomentProblemForCompactSemialgebraicSets}.
\end{proof}
This Positivstellensatz especially applies to the quadratic module $\mathcal{R}_\mu$, which gives some
insight into what types of orderings can be obtained as a result of the reduction procedure:

\begin{corollary} \label{corollary:fanaPositivstellensatz}
  Write $\mathcal G \coloneqq \{\sum_{i=s}^n z_i \cc z_i\} \subseteq \Polynomials(\CC^{1+n})^{\lie u_1}_\Hermitian$
  and let $p\in \Unit + \genQM{\mathcal G}$ be an element for which there exists $\lambda \in {[0,\infty[}$ such that
  $\lambda p(w) \ge \abs{w_i}^2$ holds for all $w \in \Levelset_\mu$ and all $i \in \{0,\dots,n\}$.
  Given $f \in \Polynomials(\CC^{1+n})^{\lie u_1}_\Hermitian$, then $f\in \mathcal{R}_\mu$ if and only if
  there is $m_1\in \NN_0$ such that for all $\epsilon \in {]0,\infty[}$ there is $m_2\in \NN_0$
  for which
  $p^{m_2}(f+\epsilon p^{m_1}) \in (\genSId{\momentmap-\mu})_\Hermitian + \genQM{\mathcal{G}}$.
  Here $(\genSId{\momentmap-\mu})_\Hermitian \coloneqq \genSId{\momentmap-\mu} \cap \Polynomials(\CC^{1+n})^{\lie u_1}_\Hermitian$.
\end{corollary}

\begin{proof}
	Note first that $\genQM{\mathcal{ G}} = \genPO{\mathcal{ G}}$ because $\mathcal{ G}$ contains only one element. The finite subset
	$\bigcup_{i,j \in \{0, \dots,n\}} \{[z_i \cc z_j + z_j \cc z_i]_\mu, \I [ z_i \cc z_j - z_j \cc z_i ]_\mu \} \subseteq \Polynomials(\Mred)_\Hermitian$
	generates the $^*$\=/algebra $\Polynomials(\Mred)$.
	Moreover, the estimate
	\begin{equation*}
    2\lambda \dupr{\delta_{[w]}}{[p]_\mu}
    =
    2\lambda p(w)
    \ge
    \abs{w_i}^2 + \abs{w_j}^2
    \ge
    \abs{w_i \cc{w_j} + w_j \cc{w_i}}
    =
    \abs{\dupr{\delta_{[w]}}{[z_i \cc z_j + z_j \cc z_i]_\mu} }	
	\end{equation*}
	holds for all $i,j \in \{0, \dotsc, n \}$ and $[w] \in \Mred$ with representative $w \in \Levelset_\mu$,
	and similarly also $2\lambda \dupr{\delta_{[w]}}{[p]_\mu} \geq \abs{\dupr{\delta_{[w]}}{\I[z_i \cc z_j - z_j \cc z_i]_\mu} }$.
	Using Proposition~\ref{proposition:semialgebraic}, this means that Theorem~\ref{theorem:fanaPositivstellensatz}
	can be applied to $\Polynomials(\Mred)$, $[\mathcal G]_\mu$, and $[p]_\mu$:
	
	It follows from Lemma~\ref{lemma:vrPoly} that $f \in \mathcal R_\mu$ if and only if $[f]_\mu \in \Polynomials(\Mred)$ is pointwise positive. Using Proposition~\ref{proposition:semialgebraic} this is equivalent to $\dupr{\varphi}{[f]_\mu} \geq 0$ for all
	$\varphi \in \semialgebraicSet{\Polynomials(\Mred)}{[\mathcal G]_\mu}$.
	By Theorem~\ref{theorem:fanaPositivstellensatz} this is the case if and only if there exists $m_1 \in \NN_0$ such that for all $\epsilon \in {]0,\infty[}$ there is $m_2 \in \NN_0$ for which $[p^{m_2}(f+\epsilon p^{m_1})]_\mu \in \genPO{[\mathcal G]_\mu}$, or equivalently
	$p^{m_2}(f+\epsilon p^{m_1}) \in (\genSId{\momentmap-\mu})_\Hermitian + \genPO{\mathcal G}$,
	because $\ker [\argument]_\mu = \genSId{\momentmap-\mu}$ and $[\genPO{\mathcal{G}}]_\mu = \genPO{[\mathcal G]_\mu}$.
\end{proof}
In contrast to \eqref{eq:RmuPolyChar} of Lemma~\ref{lemma:vrPoly}, the above Corollary~\ref{corollary:fanaPositivstellensatz}
gives a purely algebraic characterization of the quadratic module $\mathcal{R}_\mu$. For the special case of $\CC\PP^n$ we even obtain:

\begin{corollary}
	For signature $s = 1+n$, i.e.~$\Mred^{(1+n)} \cong \CC\PP^n$, the identity
\begin{equation}
  \mathcal{R}_\mu^{(1+n)}
  =
  \set[\big]{
    f \in \Polynomials(\CC^{1+n})^{\lie u_1}_\Hermitian
  }{
    f+\epsilon \Unit \in \big(\genSId{\momentmap-\mu}\big)_\Hermitian + \big(\Polynomials(\CC^{1+n})^{\lie u_1}\big)^{++}_\Hermitian \textup{ for all }\epsilon \in {]0,\infty[}
  }
  \label{eq:fanaPositivstellensatzCPn}
\end{equation}
holds,
where $(\genSId{\momentmap-\mu})_\Hermitian \coloneqq \genSId{\momentmap-\mu} \cap \Polynomials(\CC^{1+n})^{\lie u_1}_\Hermitian$
and where $(\Polynomials(\CC^{1+n})^{\lie u_1})^{++}_\Hermitian$ denotes the algebraically positive elments of
$\Polynomials(\CC^{1+n})^{\lie u_1}$ like in \eqref{eq:algebraicOrder}, i.e.~the sums of Hermitian squares.
Moreover, every algebraic state $\omega$ on $\Polynomials(\CC\PP^n)$ is positive, hence a state.
\end{corollary}

\begin{proof}
	Equation \eqref{eq:fanaPositivstellensatzCPn} is just Corollary~\ref{corollary:fanaPositivstellensatz}
	for $s=1+n$ and $p\coloneqq \Unit$. Given any algebraic state $\omega$ on $\Polynomials(\CC\PP^n)$
	then this shows that $\dupr{\omega}{[f]_\mu} = \dupr{\omega}{{[f+\epsilon \Unit]_\mu}} - \epsilon \ge -\epsilon$
	for all $\epsilon \in {]0,\infty[}$ and all $f\in \mathcal{R}_\mu^{(1+n)}$ because $[\genSId{\momentmap-\mu}]_\mu = \{0\}$
	and $[(\Polynomials(\CC^{1+n})^{\lie u_1})^{++}_\Hermitian]_\mu \subseteq \Polynomials(\CC\PP^n)^{++}_\Hermitian$.
	So $\omega$ is positive with respect to the pointwise order on $\Polynomials(\CC\PP^n)$.
\end{proof}
However, for other signatures $s\in \{1,\dots,n\}$, it is necessary to add an additional generator like
$\sum_{i=s}^n z_i \cc z_i$ to the description of $\mathcal{R}_\mu^{(s)}$, and there do exist algebraic
states on $\Polynomials(\Mred^{(s)})$ which are not positive because they yield negativ results on $\sum_{i=s}^n [z_i\cc z_i]_\mu$,
e.g.~the functionals $\delta_{[w]}$ defined in \eqref{eq:evaluationFunctionals} with $[w] \in \MredExt^{(s)} \setminus \Mred^{(s)}$.
From a purely algebraic point of view, this might be rather unexpected.

These algebraic characterizations of the quadratic module $\mathcal{R}_\mu$, hence of the order on $\Polynomials(\Mred)$,
are especially interesting with respect to a possible generalization to the non-commutative
case described in \cite{schmitt.schoetz:PreprintWickRotationsInDQ}.
If one treats the reduction of non-formal star products from $\CC^{1+n}$ to $\Mred$ in the
context of representable Poisson $^*$\=/algebras, can one give similar characterizations of the corresponding quadratic module $\mathcal{R}_\mu$?
In the special case $s=1+n$, i.e.~for the deformation quantization of $\CC\PP^n$, such a characterization will be obtained in Part II of this article, \cite{schmitt.schoetz:SymmetryReductionOfStatesII}.

\bibliographystyle{chairx_copy}

\end{onehalfspace}

\end{document}